%% file: main.tex
\documentclass[a4paper,11pt]{article}
\usepackage{biblatex}
\usepackage{graphicx} 
\usepackage[margin=1in]{geometry}
\usepackage{booktabs}
\usepackage{mathtools}
\usepackage{amsmath}
\usepackage{amsthm}
\usepackage{amssymb}
\usepackage{amsfonts}
\usepackage{bbm}
\usepackage{graphicx}
\usepackage{titling}
\usepackage{multicol}
\usepackage{pdfpages}
\usepackage[font=small]{caption}
\usepackage[capposition=top]{floatrow}
\usepackage{setspace}
\newcommand{\qcite}[1]{\citeauthor{#1} (\citeyear{#1})}
\newcommand{\qcitenp}[1]{\citeauthor{#1} \citeyear{#1}}

\usepackage{pgfplots}
\usetikzlibrary{matrix}
\usetikzlibrary{positioning}
\usetikzlibrary{calc}

\tikzset{
    empty/.style={
        draw opacity=0,
        fill opacity=1,
        align=center
    },
    payoffmatrix/.style={
        matrix of nodes,
        nodes in empty cells,
        column sep=-\pgflinewidth,
        row sep=-\pgflinewidth,
        nodes={
            draw, 
            minimum height=1cm,
            text width=2cm,
            align=center,
            anchor=center,
            inner sep=0pt,
            outer sep=0pt
        },
        column 1/.style={nodes={draw=none, font=\bfseries, text width=2.2cm}},
        row 1/.style={nodes={draw=none, font=\bfseries, minimum height=0.6cm}}
    }
}

\DeclareMathOperator*{\argmax}{\arg\!\max}
\DeclareMathOperator*{\argmin}{\arg\!\min}
\newcommand{\ind}{\perp\!\!\!\!\perp} 
\newcommand{\expectation}[1]{\mathbb{E}[#1]}

\newcommand{\CE}[2][]{\mathbb{E}[#2|#1]}
\newcommand{\distexpectation}[2][]{\mathbb{E}_{#1}[#2]}
\DeclarePairedDelimiterX{\Earg}[1]{[}{]}{#1}
\DeclarePairedDelimiterX{\Vararg}[1]{(}{)}{#1}
\newcommand{\E}[2][]{\mathbb{E}_{#1}\Earg#2}
\newcommand{\CELR}[2][]{\mathbb{E}\left[#2\middle|#1\right]}

\newcommand{\Supp}[1]{\text{Supp}(#1)}

\newtheorem{theorem}{Theorem}
\newtheorem{proposition}{Proposition}
\theoremstyle{definition}

\newtheorem{lemma}{Lemma}
\newtheorem{corollary}{Corollary}
\newtheorem{assumption}{Assumption}
\newtheorem{remark}{Remark}
\newtheorem{example}{Example}

\title{Improving and Balancing Task Allocation with Market Design}
\author{Terence Highsmith II}
\date{Version: \today}

\begin{document}

\maketitle
\begin{center}
    \begin{minipage}{0.9\textwidth}
\textit{In balanced task allocation, a firm allocates stochastically arriving tasks to workers. Balance requires that workers complete the same number of tasks or have equally difficult task allocations. Prominently, U.S. foster care agencies rotate allocation across workers to satisfy balance. The Rotation mechanism is highly inefficient and could contribute to poorer welfare and productivity for workers. I take a market design approach to construct the Dynamic Pseudomarket Mechanism (DPM). It satisfies the institutional balance constraint, Pareto efficiency, and approximate incentive compatibility. I develop a data-driven bound on DPM's expected productivity because DPM could increase or decrease productivity in practice. Simulation results indicate this approach can accurately predict DPM's effects with minimal data. The theoretical results underpinning the bound establish that, in large random markets, workers' welfare converges to the maximum in fractional, multi-unit Pareto efficient matchings, a finding of independent interest.} (JEL D47)	
\end{minipage}
\end{center}

\vspace{1em}

\input{Paper/Introduction}

\input{Paper/Preliminaries}

\input{Paper/DPM}

\input{Paper/Statistic}

\input{Paper/Simulations}

\input{Paper/Conclusion}

\printbibliography

\appendix
\counterwithin{table}{section}
\counterwithin{figure}{section}
\counterwithin{theorem}{section}
\counterwithin{proposition}{section}
\counterwithin{lemma}{section}
\counterwithin{assumption}{section}
\counterwithin{remark}{section}
\counterwithin{corollary}{section}
\counterwithin{example}{section}

\input{Paper/Examples.tex}

\input{Paper/Relationship.tex}

\input{Paper/WeightedRotation.tex}

\input{Paper/Appendix}

\input{Paper/ReducedForm.tex}

\end{document}

%% file: Paper/Introduction.tex
\section{Introduction}

In a task allocation problem, a firm immediately allocates dynamically arriving tasks to workers. Many firms require these allocations to be \textit{balanced}: workers must receive the same average number of cases or workload. Foster care in the United States is a primary example. County social agencies take custody of children who are victims of abuse or neglect, and supervisors assign each child a caseworker who manages the child's social services. Federal guidelines explicitly direct supervisors toward balance: ``The goal of a supervisor is to achieve equitable caseloads'' (\qcitenp{salus2004supervising}). In practice, counties guarantee balance with Rotation: the next case goes to the caseworker who has least recently received one (\qcitenp{baron2024mechanism}).

I take balance as given. Following \qcite{sonmez2023}'s minimalist market design, I treat it as a policy objective of the institution rather than a modeling choice: the designer's task is to fix what the existing mechanism does poorly while preserving what its stakeholders value. Balance is also natural on its own terms. A caseworker can observe how many cases her colleagues receive, but not their preferences, which underlie most fairness criteria in economics. Most importantly, a mechanism that keeps the institution's constraint is far more likely to be adopted than one that asks agencies to abandon it.

What Rotation does poorly is efficient allocation. Caseworkers' preferences vary over the child's age, language, geographic location, and other case factors, and Rotation ignores them entirely. Evidence suggests that this inefficiency contributes to perennially high caseworker turnover, poor caseworker productivity, and poor outcomes for foster children (\qcitenp{carnochan2013}; \qcitenp{marsh2020caseworker}; \qcitenp{steen2026case}). Agencies on fixed, state-governed budgets cannot simply raise salaries, nor can they safely turn away cases. Better allocation is one of the few levers that remain. Moreover, changing the allocation mechanism to improve allocations must avoid unintended effects on caseworker productivity. If a mechanism increases worker welfare but decreases productivity, then it could harm service quality for vulnerable children. An agency would be rightfully hesitant to implement the mechanism. This is a market design problem: is there a mechanism that preserves balance and improves worker welfare without harming productivity and, through it, child outcomes?

I model task allocation in a framework that allows me to (i) design an efficient, balanced mechanism and (ii) predict its effect on productivity. A firm allocates tasks that arrive at each time $t$ over a horizon $T$ to workers with arbitrary preferences over tasks. Each task $x$ has a weight $w^x$ that reflects its workload. An allocation is \textit{ex-post balanced} if workers' total weights are approximately equal at the end of the horizon, and \textit{ex-post efficient} if no reallocation of the realized tasks makes every worker weakly better off and some worker strictly better off. A mechanism elicits preferences at the beginning of the horizon and uses them to allocate tasks as they arrive.

I propose the Dynamic Pseudomarket Mechanism (DPM). DPM elicits worker preferences then internally constructs a \textit{pseudomarket}. In the pseudomarket, workers receive fiat budgets of money. The pseudomarket commodifies future tasks into probability shares that workers purchase using their fiat. The mechanism uses reported preferences to compute a competitive equilibrium allocation. On its own, a competitive equilibrium of this market is the classical \qcite{hylland1979} approach, and it is not necessarily balanced. DPM modifies the pseudomarket to respect the balance constraint through manipulating the fiat budgets. The intuition is that whenever a worker would ``over-consume''---obtaining more tasks than is desired---the market reduces her budget. Whenever she would ``under-consume'', the market increases her budget. Using this idea, I prove that a market-clearing, balanced equilibrium allocation always exists (Theorem \ref{thm:existence}). DPM computes it once and allocates each arriving task according to the equilibrium allocation.

DPM preserves Rotation's balance in expectation and in the long run, and short-run deviations are controlled: the median U.S. caseworker receives about fourteen cases per quarter, and under DPM her worst-case expected deviation from this is about four cases. Some relaxation is unavoidable because no ex-post efficient mechanism is ex-post balanced (Proposition \ref{proposition:expost}). DPM is ex-post efficient, and its distribution over allocations cannot be Pareto improved before tasks arrive (Theorem \ref{thm:efficient}). \qcite{echenique2021} show that pseudomarkets are typically \textit{constrained} efficient: no allocation satisfying the constraints Pareto improves on them. DPM is efficient among \textit{all} allocations, so long-run balance does not exclude unconstrained efficiency. Rotation, by contrast, is inefficient with probability one in the long run (Proposition \ref{proposition:inefficient}). 

No mechanism is ex-post efficient, asymptotically balanced, and strategyproof (Proposition \ref{proposition:impossibility}), so DPM is manipulable. I study large random markets with a random sample of a firm's workers. This represents the limit case of a large firm where workers have non-trivial preference heterogeneity. In these markets, DPM is ``approximately'' incentive compatible: the probability that a worker can gain more than any fixed amount by manipulating vanishes (Proposition \ref{proposition:approximate}).

An agency could have serious concerns about implementing DPM if it would adversely impact caseworker productivity; this could cause service quality to decline for vulnerable children. This is possible if caseworkers prefer cases in which they are less productive, for example, because those cases require less effort. The gold standard for measuring a mechanism's effects is a randomized field experiment (\qcitenp{kapor2020}; \qcitenp{abdul2022}; \qcitenp{budish2022}), but in foster care an experiment exposes children to an untested mechanism and can be prohibitively costly for a resource-constrained agency.

My main methodological contribution is a bounding approach that predicts DPM's directional effect on expected productivity from data that the agency already has or can cheaply collect. With two task types, I characterize DPM exactly: a worker receives a task type if and only if her relative valuation for it clears an endogenous threshold. DPM's expected productivity is a weighted sum of the productivities of workers ranked by the endogenous threshold. This creates a sorting effect: workers with a higher relative value for a task receive it. I characterize when this sorting either increases or decreases productivity as a function of available data.

The bound depends on the share of workers who prefer each task type and their average productivity in each type. Because Rotation is random assignment, these moments can be estimated from existing outcome records and a one-question preference survey. The bound lies between DPM's and Rotation's expected productivity (Theorem \ref{thm:expected}). That is, the bound has the correct sign, and it understates the magnitude of DPM's effect. The bound becomes more conservative with more task types. In markets with many workers, DPM's expected productivity with many task types exceeds its expected productivity with two task types (Theorem \ref{thm:general}). This occurs because DPM exploits comparative advantage across many tasks rather than only pairwise advantages as the two task type case assumes. These results imply that the bound is a correct directional prediction in general markets.

The proof rests on an independently interesting result that workers receive near maximum utility in any efficient and balanced allocation. It also implies that nearly all workers’ gains from manipulation fall below any fixed amount in large enough markets. Efficient and balanced allocation nests fractional, multi-unit Pareto efficient matchings as a special case, so the result is also true for the canonical pseudomarket model. The proof is similar to \qcite{che2018}'s in random matching markets but uses a new technique to address technical difficulties that arise from fractional, multi-unit allocations.

My bounding approach has tradeoffs. It relies on a linear specification for expected productivity and cannot capture non-linear allocation effects. However, this specification is testable under the status quo. The bound is also conservative, and it can predict the wrong direction due to sampling error. Sampling error falls quickly: in simulations, the bound predicts the direction of DPM's effect correctly about 90\% of the time with about twenty workers, even when preference reports are very noisy. Simulation results also confirm the theoretical results that the bound's predictions can be robust to strategic incentives in the sense that \textit{any possible} manipulation has asymptotically small gain. Workers' incentives to manipulate decrease quickly as the market grows when balance is allocating the same number of tasks to workers; this can fail if balance is weighted by task workload. 

A resource-constrained agency can therefore assess DPM's effect on productivity at the cost of a preference survey and without exposing children to an untested mechanism. More broadly, the bound predicts a counterfactual mechanism's effect from moments observable under the status quo, without new random variation or a fully specified structural model. The main takeaway from these results is that DPM is a mechanism that (i) satisfies the institutional constraint, (ii) improves worker allocations, and (iii) has a predictable effect on expected productivity. This positions it as a viable alternative to Rotation in agencies where it has a demonstrable positive effect on productivity.

\subsection*{Running Example}

As of 2022, 368,530 children in the United States were in foster care (\qcitenp{afcars}). Local authorities employ caseworkers to investigate allegations of abuse and manage ongoing care for children removed from their homes. This involves two task allocation problems: allocating investigations (\qcitenp{baron2024mechanism}) and allocating ongoing cases. Many localities separate the two, either with separate teams or by contracting ongoing care to private agencies (\qcitenp{steen2013impact}). I focus on ongoing cases. Upon removal, a local authority or agency assumes responsibility for the child's care, and a caseworker supervisor routes the child's case to a specific caseworker. A recent survey of one-hundred and fifty caseworkers and sixty-five supervisors confirms that balance is a common constraint and that some agencies/supervisors use Rotation. The following quotes are from \qcite{steen2026case}:

\begin{quote}
    Respondents described eight principles for case assignment. The first four principles focus on distribution methods that aim to equalize workloads. In these methods, cases are assigned based on the time the last case was assigned to the case manager (rotation), the number of cases the case manager currently holds (equalization of caseload), the number of children and/or families the case manager currently serves (equalization of children and/or families served), and the complexity/difficulty of the case (equalization of case complexity). 
\end{quote}

The first and fourth principles are Rotation mechanisms; the fourth balances \textit{weighted} allocations\footnote{These two are \textit{procedural} balance: caseworkers receive the same number of cases or same case complexity over some long time horizon $T$, even if outcomes differ due to some caseworkers' propensity to close cases faster. The second and third are \textit{outcome} balance: workload is equalized across caseworkers at every time $t$, even if this means some caseworkers could receive more work if they are more productive. Note that these principles all coincide if all caseworkers close cases at the same rate. I abstract away from this in my model and focus on procedural balance.}. Some supervisors did match cases to caseworkers' preferences over child age and location and to competencies such as proficiency in Spanish. One caseworker described poignantly: ``Years ago I had a baby die on my caseload and my unit manager has always been sensitive and asks when assigning a case with a baby.'' This highlights that there are many dimensions of preference heterogeneity. Still, because preferential matching contradicted balance, experiences were mixed:
\begin{quote}
    Matching to case manager competence had negative implications for those receiving and those not receiving the case. Those receiving the case, typically the seasoned workers, were described as becoming ``overwhelmed'' and ``stressed'' by repeated assignment of difficult cases.
\end{quote}

Caseworker productivity can also vary across cases. Productivity is the quality of a caseworker's matches, measured by \textit{placement stability}. A placement is unstable if the foster home prematurely returns the child to the authorities due to relational breakdown. Caseworkers are responsible for identifying matches that are likely to be stable and for ongoing check-ins with the child and the foster home. Stability is also how agencies are evaluated: private agencies hold performance-based contracts with local authorities, and local authorities face federal penalties for poor performance (\qcitenp{taylor2010performance}). Caseworkers assigned to cases they prefer might provide higher-quality care, and the survey echoes this concern: ``When case assignment prioritizes the equalization of caseload over the preference and convenience of the case manager ... this can impact the quality of case management.'' For these reasons, using a Pareto efficient and balanced allocation mechanism to assign cases to caseworkers could be poised to improve worker welfare and productivity.

While this foster care example will recur in this paper, Rotation and balanced task allocation are also important in other applications: allocating cases to investigators, defendant cases to public lawyers (\qcitenp{dogan2025}), support tickets to IT professionals\footnote{Popular customer service software provider HubSpot explicitly provides an automatic ticket rotation mechanism. See: \qcite{hubspot2026}.}, and more. Benefits from Pareto efficient and balanced allocation could extend to many different firms.

\subsection*{Literature Contribution}

\textbf{Pseudomarkets.} The cornerstone of pseudomarket mechanisms is \qcite{hylland1979}. Allocating probability shares lets indivisible objects function as divisible goods, so agents with equal budgets of fake money can trade in a Walrasian economy, and the competitive equilibrium is Pareto efficient and envy-free. Subsequent work applies the same formula to school choice, course allocation, matching with constraints, collective choice, combinatorial allocation, and food bank allocation (\qcitenp{budish2011}; \qcitenp{he2018pseudo}; \qcitenp{echenique2021}; \qcitenp{gul2025pseudo}; \qcitenp{nguyen2025}; \qcitenp{prendergast2022}; see \qcitenp{pycia-2023} for a review). These mechanisms typically target efficiency and envy-freeness, and balance is substantially different from envy-freeness (Appendix \ref{appendix:relationship}). I contribute an important new application, balanced task allocation, and two technical tools: a model in which agents purchase probability shares of future objects, and an equilibrium concept in which the previously fixed budgets vary endogenously. Unlike the constrained efficiency typical of constrained allocation problems (\qcitenp{budish2013}; \qcitenp{echenique2021}), DPM is efficient among all allocations. This suggests that future research could investigate conditions for when constraints do not impose limitations on efficiency.

\textbf{Random Markets.} The literature on random matching markets usually analyzes the average rank of matching mechanisms in abstract markets (\qcitenp{knuth1996}; \qcitenp{ashlagi2017}; \qcitenp{nikzad2022rank}). One notable departure is \qcite{aryal2026}; this paper employs random markets to establish asymptotic validity for a method of aggregating individual preferences to rank alternatives. My largest methodological contribution uses random markets for a conceptually similar role, although the contribution is substantively different. I develop a bound on my mechanism's expected productivity using data that summarizes moments of the random market. I derive the bound using simplifying assumptions, but large market results show that the bound can extend to general environments. Specifically, I extend \qcite{che2018}'s analysis of Pareto efficient matchings in large random markets to show that workers receive near maximum utility in every efficient and balanced allocation. They consider many-to-one deterministic matchings, whereas my model allows and nests fractional, multi-unit (one-to-many) allocations and uses a new proof technique.

\textbf{Fair Online Division.} The balanced task allocation problem is closer to fair online division, which allocates (in)divisible objects that arrive dynamically to a fixed set of agents (\qcitenp{dickerson2014}; \qcitenp{kurokawa2016}; \qcitenp{he2019}; \qcitenp{gkatzelis2021}; \qcitenp{bogomolnaia2022}). This literature typically targets variants of envy-freeness or maximin welfare, and balance is different from both. \qcite{benade2024fair} is the closest work to mine. Their setting is nearly identical, but they design algorithms for approximate efficiency and envy-freeness, and they prove that ex-post efficiency and ex-post envy-freeness up to one good are incompatible. The different objective leads to a different solution: I use the insight that budgets can be unequal in equilibrium to achieve balance. I also simplify and generalize their proof that efficiency for a fractional allocation implies ex-post efficiency for realized allocations. Aside from this, my consideration of a mechanism's expected effects on productivity and applying this model to task allocation are also novel in this strand.

\textbf{Rotation.} Two other works study Rotation, though neither directly addresses balance. \qcite{baron2024mechanism} design a mechanism to allocate child abuse investigations, where Rotation was the status quo. They maximize productivity subject to maintaining status-quo worker welfare, so their mechanism increases productivity with incidental or null effects on welfare. I take balance as the constraint and design a mechanism that improves worker welfare with incidental effects on productivity. Keeping balance has a cost: their mechanism is strategyproof, while no efficient and balanced mechanism can be (Proposition \ref{proposition:impossibility}). We also differ in applications (investigations with two types versus ongoing cases with arbitrary types) and frameworks (mechanism design with quasirandom identification versus a pseudomarket and bounding approach).

\qcite{dogan2025} design a mechanism for allocating cases to public defenders in Texas, where Rotation was used before a brief switch to a defendant-choice mechanism that suffered from congestion at popular lawyers. They study a static problem from the defendants' side: efficiency, fairness for defendants, and robustness in equilibrium. DPM addresses the other side of the market, allocating cases based on lawyers' preferences while equalizing their workloads.

\textit{Roadmap:} Section 2 presents the model. Section 3 introduces DPM and its efficiency, balance, and incentive properties. Section 4 develops the productivity bound, Section 5 evaluates it in simulations, and Section 6 concludes.

%% file: Paper/Preliminaries.tex
\section{Preliminaries}

\textbf{Notation.} The space of real numbers is $\mathbb{R}$. $m \in \mathbb{R}^{n \times k}$ denotes a matrix with $n$ rows and $k$ columns. For two vectors $v, v'\in \mathbb{R}^n$ with $n$ rows, the dot product is $v \cdot v' := \sum_{i = 1}^n v_i v'_i$. The expectation of a random variable $Y$ with respect to a distribution $F$ is $\E[F]*{Y}$. The support of a distribution $F$ is $Supp(F)$.

\textbf{Agents and Objects.} There is a finite market. The agents are $I = \{1 , 2, ..., |I|\}$ with a typical $i \in I$. The tasks (task types) are $X = \{1, 2, ..., |X|\}$ with a typical $x \in X$. The time horizon is $[T] = \{1, 2, ..., T\}$. The market unfolds over this finite horizon, and one task $x \in X$ arrives at each time $t$ according to the family of independent distributions $F = (F_1, F_2, ..., F_t, ..., F_{T-1}, F_T)$. I assume that $F_t(x) > 0$ for all $(x,t) \in X \times [T]$. Task weights are $w = (w^1, w^2, ..., w^{|X|})$ such that $w^x \in (0,1]$ for all $x \in X$. In my running example, equal weights correspond to balancing caseloads, and unequal weights represent a social work agency's estimate of case complexity for case $x$. I will denote the total expected weight as $W_F := \sum_{t = 1}^T \sum_{x \in X} F_t(x) w^x$.

\textbf{Interim Allocations.} $q_{i,t}^x$ represents the quantity of $x$ that $i$ consumes \textit{conditional on} $x$ arriving in time $t$. A single-period allocation is $q_{i,t} \in \mathbb{R}_{+}^{|X|}$. A multi-period allocation is a matrix $q_i = (q_{i,1}, q_{i,2}, ... q_{i,T})$. I will denote the consumption set for $i$ as $\mathcal{Q} := \mathbb{R}_{+}^{|X| \times T}$ such that $q_i \in \mathcal{Q}$. The market allocation is a multidimensional matrix $q = (q_1, q_2, ..., q_{|I|})$. $q$ is feasible if $\sum_{i \in I} q_{i,t}^x \le 1$ for all $(x,t) \in X \times [T]$. I will write $E_F[q_i] := \sum_{t = 1}^T \sum_{x \in X} F_t(x) q_{i,t}^{x}$ to denote the expected number of tasks that $i$ receives under $q_i$. The weight-expected tasks for $i$ is $E_{F,w}[q_i] := \sum_{t = 1}^T \sum_{x \in X} F_t(x) w^x q_{i,t}^{x}$.

\textbf{Ex-Post Allocations.} The set of ex-post allocations is:
\[\mathcal{A} = \left\{ a \in \mathcal{Q} : a_{i,t}^{x} \in \{0,1\} \hspace{3pt} \forall x \in X \text{ and } \sum_{i \in I} \sum_{x \in X} a_{i,t}^{x} \leq 1 \hspace{3pt} \forall t \in [T]\right\}\]
The set of ex-post allocations at $k$ is:
\[\mathcal{A}_k = \{ a \in \mathcal{Q} : a \in \mathcal{A} \text{ and } a_{i,t}^x = 0 \text{ if } t \ne k  \}\]
and the set of realizations is:
\[\tilde{A} = \left\{ \tilde a \in \mathbb{Z}^{|X| \times T}_+ : \sum_{x \in X} \tilde a_t^x = 1 \text{ for all } t\right\}\]
where $\tilde A_k$, the realizations at $k$, are defined analogously with $\tilde a_t^x = 0$ if $\tilde a \in \tilde A_k$ and $t \ne k$. An ex-post allocation $a$ assigns an integral task and, at most, one task at each time. I will refer to $a$ as feasible for realizations $\tilde a$ if $\sum_{i \in I} a \le \tilde a$. $\tilde a$ represents the tasks that arrive where $\tilde a_t^x = 1$ if $x$ arrives at $t$, and $a$ cannot allocate more tasks than those that arrive. The set of \textit{random ex-post allocations} is the set of probability distributions $\Delta(\mathcal{A})$ over $\mathcal{A}$. The weight of an ex-post allocation $a_i$ is $W(a_i) = \sum_{t = 1}^T \sum_{x \in X} w^x a_{i,t}^{x}$, and $W(a) = \sum_{i \in I} W(a_i)$. I will refer to $q$ as an interim allocation and $a$ as an ex-post allocation.

\textbf{Valuations.} An agent's valuation is a random vector $v_i \in \mathbb{R}^{|X|}$ where $v_i(x)$ is $i$'s valuation for $x$. $v$ has a commonly known joint distribution $\mathcal{V}$. I assume that, for all $i \in I$, (i) $v_i$ is continuously distributed, (ii) has $\sup \Supp{v_i(x)} = \bar{V}$ for some $\bar V \in \mathbb{R}$ for all $(i,x) \in I \times X$, (iii) has $\inf \Supp{v_i(x)} > 0$ for all $(i,x)\in I \times X$, and (iv) is i.i.d across $i$. (iv) means that I consider a random sample of workers from the firm's worker distribution\footnote{A firm's workers are generally selected from a population distribution. This allows me to restrict attention to the distribution of workers that the firm selects and analyze a generic, random sample from the firm-level distribution.}. The vector of valuations is $v = (v_1, v_2, ..., v_{|I|})$. I assume that $i$'s utility for an allocation is additive so that $i$'s expected (interim) utility is $u_i : \mathcal{Q} \to \mathbb{R}$ where $u_i(q_i) = \sum_{t = 1}^T \sum_{x \in X} F_t(x) q_{i,t}^x v_i(x)$. $i$'s ex-post utility from an ex-post allocation $a$ is $U_i(a_i) = \sum_{t = 1}^T \sum_{x \in X} a_{i,t}^x v_i(x)$. A market is a tuple describing the primitives: $\mathcal{M} = (I, X, [T], F, w, v)$. The set of markets is $\mathbb{M}$.

\textbf{Desiderata and Mechanisms.} An allocation $q \in \mathcal{Q}$ is \textit{interim balanced} if $E_{F,w}[q_i] = W_F/|I|$ for all $i \in I$. An ex-post allocation $a \in \mathcal{A}$ is \textit{ex-post balanced} if $|W(a_i) - W(a)/|I|| < \max_{x \in X} w^x$ for all $i \in I$\footnote{The constant factor allows some ex-post balanced allocation to exist. For example, if $T = 1$, then, for some arrival, every allocation necessarily has non-zero imbalance because only one worker receives a task.}. An allocation $q$ is \textit{interim efficient} if there does not exist any feasible $q'$ such that $u_i(q'_i) \geq u_i(q_i)$ for all $i \in I$ and $u_j(q'_j) > u_j(q_j)$ for some $j \in I$. $a$ is \textit{ex-post efficient for $\tilde a$} if $a$ is feasible for $\tilde a$ and there does not exist any $a'$ feasible for $\tilde a$ such that $U_i(a'_i) \geq U_i(a_i)$ for all $i \in I$ and $U_j(a'_j) > U_j(a_j)$ for some $j \in I$.

A random allocation mechanism is $\pi = (\pi_1, \pi_2, ..., \pi_T)$ where $\pi_t : \mathbb{R}^{|I|\times|X|}_{++} \times \{\mathcal{A}_{k}\}_{k < t} \times \{\tilde A_k\}_{k \le t} \rightarrow \Delta(\mathcal{A}_t)$ assigns a random ex-post allocation for reported preferences given previous allocations and realized arrivals. $\pi$ is feasible if it is feasible for each $\tilde a$. This implicitly requires that tasks must be allocated immediately or discarded. I will typically omit $\pi$'s dependence on $(a, \tilde a)$ and simply write $\pi(v)$. For each $v$, I identify $\pi(v)$ with the interim allocation $q \in \mathcal{Q}$ given by $q_{i,t}^x := \Pr(a_{i,t}^x = 1 \mid \tilde a_t^x = 1, \pi(v)) = \E*{a_{i,t}^x}/F_t(x)$, so that $u_i(\pi(v))$ and $E_{F,w}[\pi_i(v)]$ are well-defined.

A random allocation mechanism $\pi$ has an interim property if the distribution $\pi(v)$ has that property for all $v \in \Supp{\mathcal{V}}$, and it has an ex-post property if every ex-post allocation in its support has the property for all $v \in \Supp{\mathcal{V}}$. $\pi$ is \textit{strategyproof} if truthfully reporting is weakly dominant in the interim for each agent: for all agents $i \in I$, $u_i(\pi(v_1, ..., v_i, ..., v_{|I|})) \geq u_i(\pi(v_1, ..., \hat{v}_i, ..., v_{|I|}))$ for all $\hat{v}_i \in \mathbb{R}^{|X|}_{++}$ and all $v \in \Supp{\mathcal{V}}$.

\textbf{Knowledge and Timing.} I assume that the designer and all agents know the horizon $T$, the distribution $F$, and the weights $w$. Each agent privately observes her own valuation $v_i$. The joint distribution of $v$ is common knowledge for everyone. The goal of the firm or market designer is to construct a mechanism that elicits preferences $v$ and dynamically allocates tasks $a$ that satisfy desirable properties. The timing includes three stages:
\begin{enumerate}
    \item \textit{Ex-Ante:} Valuations $v$ have not yet realized.
    \item \textit{Interim:} Valuations $v$ realize. Each worker $i$ observes $v_i$, simultaneously reports $\hat v_i$ to $\pi$, and $\pi(\hat v)$ is observed by all workers.
    \item \textit{Ex-Post:} Each arrival $\tilde{a}_t^x$ realizes, and some $a_t \in \Supp{\pi_t(\hat{v})}$ realizes.
\end{enumerate}
I emphasize that the interim allocations $q$ are what are typically referred to in the literature as ex-ante allocations. However, this model has a ``true'' ex-ante stage and studies properties before valuations realize, so I disambiguate using this terminology.

\textbf{Discussion.} The main simplification of this model is to assume that utility for an allocation is additive expected utility. While this rules out many other interesting cases, I view this as a good starting point for studying balanced task allocation. Many results do not trivially extend without additive expected utility. Valuations are strictly positive. This can be viewed as workers having strong enough incentives to perform a marginal task or as the firm's choice of a preference reporting language. Last, the model also abstracts away from tasks that might have temporal start and end dates. Instead, here, balance is \textit{procedural} in that a balanced mechanism allocates equal weighted-tasks to workers. In a more complex model, workers might temporally experience imbalance in \textit{outcomes} if some workers are able to finish tasks more quickly than others.

%% file: Paper/DPM.tex
\section{The Dynamic Pseudomarket Mechanism}

I begin with Rotation then construct DPM. All proofs are relegated to Appendix \ref{appendix:theory}.

\subsection{Rotation}

The Rotation mechanism draws a uniformly random order of workers and allocates each arriving task to the next worker in the order. \\~\

\textit{Algorithm: Rotation} (Input: $\mathcal{M}$.)
\begin{enumerate}
    \item Draw a uniformly random order of $I$.
    \item Allocate an item $x$ arriving at time $t$ to the next worker $i$ in the order. \\~\
\end{enumerate}

I denote $\pi_R$ as the Rotation mechanism. The main advantage of Rotation is that it is ex-post balanced if weights are equal, and it is balanced in the long-run if arrivals are time-invariant. To study long-run properties, I introduce $k$-extension markets that replicate the length of the market $k$ times. The $k$-extension (for an integer $k$) of $\mathcal{M}$ is defined as $\mathcal{M}^k = (I, X, [kT], F^k, w, v)$ where $F^k = (F^1, F^2, ..., F^{k})$ satisfies $F^i = F$. $\pi^k(v)$ is $\pi$ applied to the $k$-extension $\mathcal{M}^k$.

$\pi$ is asymptotically ex-post balanced if, for all $\epsilon > 0$, $i \in I$:
\[\lim_{k \rightarrow \infty} \Pr\Bigg\{ \bigg| \frac{1}{k} W(a_i^k) - \frac{W_F}{|I|} \bigg| < \epsilon \Bigg\} = 1\]
where $a^k \sim \pi^k(v)$. This property means that, on average, workers receive allocations with the same weight in the long run.

\begin{remark}\label{remark:rotation}
    Rotation is ex-post balanced if $w^x = 1$ for all $x \in X$.  Rotation is asymptotically ex-post balanced if $F_t(x) = F_{t'}(x)$ for all $x \in X$ and $t \neq t'$.
\end{remark}

When weights are not equal, Rotation's random assignment is not necessarily ex-post balanced because some workers can coincidentally receive greater workload. When arrivals are not time-invariant, asymptotic ex-post balance can fail because the initial random sequence can result in a worker receiving tasks at intervals where tasks with higher weight happen to arrive\footnote{A weighted variant instead allocates each task to the worker with the least weighted allocation. This is similar to the survey's equalization of case complexity, and it is ex-post balanced (Appendix \ref{appendix:weighted-rotation}). I focus on the variant in the main text because it is widely used.}. Rotation's disadvantage is that it is not efficient. Worse, it is ex-post inefficient with probability one in the long run for any distribution of $v$ that satisfies the primitives' conditions. 

I denote $a^k_R \sim \pi_R^k(v)$ as an ex-post realization of $\pi_R^k(v)$. Then:

\begin{proposition}\label{proposition:inefficient}
    \textit{(Inefficiency)} $\lim_{k \to \infty} \Pr(a^k_R \text{ is ex-post efficient}) = 0$.
\end{proposition}

Essentially, Rotation is always inefficient in the long run. Rotation ignores preferences, and every worker eventually receives every task type. Pareto improving trades exist whenever two workers' relative valuations between two tasks differ because workers can trade several units of one task for one unit of another. However, no mechanism achieves both ex-post efficiency and ex-post balance.

\begin{proposition}\label{proposition:expost}
    \textit{(Impossibility)} An ex-post efficient and ex-post balanced mechanism does not exist.
\end{proposition}

Rotation is ex-post balanced only with equal weights (Remark \ref{remark:rotation}), and supervisors in the running example already occasionally depart from balance to satisfy preferences. Following this practice, DPM relaxes ex-post balance to asymptotic ex-post balance to achieve efficiency.

\subsection{Efficiency and Balance}

I design the Dynamic Pseudomarket Mechanism (DPM) to allocate tasks in an efficient and balanced manner. DPM relies on a novel equilibrium concept: balanced equilibria.

\textbf{Balanced Equilibria.} A single-period price vector $p_t \in \mathbb{R}^{|X|}$ describes a price for each task at time $t$. The price matrix is $p = (p_1, p_2, ..., p_T)$. A purchasing power ratio for $i$ is $r_i \in \mathbb{R}$, and the purchasing power ratio vector is $(r_1, r_2, ..., r_{|I|})$; I will also refer to $i$'s implied budget as $b_i := 1/r_i$. A balanced equilibrium is a tuple $(p^*, r^*, q^*)$ such that, for all $i \in I$:
\begin{enumerate}
    \item $r_i^* p^* \cdot q_i^* \leq 1$
    \item $u_i(q_i) > u_i(q_i^*)$ implies that $r_i^* p^* \cdot q_i > 1$ for each $q_i \in \mathcal{Q}$.
    \item $q^*$ is interim balanced.
\end{enumerate}
An equilibrium is \textit{market clearing} if $p^{*,x}_t > 0$ implies that $\sum_{i \in I} q_{i,t}^{*,x} = 1$ for all $(x,t) \in X \times [T]$.

A constrained-efficient interim balanced allocation necessarily exists, but it is not obvious that an \textit{unconditionally} interim efficient and interim balanced allocation exists. I prove that they do through proving that a market-clearing balanced equilibrium exists. Since every market-clearing balanced equilibrium is a market-clearing competitive equilibrium (taking the endogenous ratios as fixed), the first welfare theorem will also hold, and the existence of interim efficient and interim balanced allocations is guaranteed.

\begin{theorem}\label{thm:existence}
    \textit{(Existence)} There exists a balanced equilibrium satisfying market clearing.
\end{theorem}

The intuition is analogous to Robin-Hood mechanisms: by taking from the ``rich'' (allocations with large weight) and giving to the ``poor'' (allocations with small weight), the market travels along the Pareto frontier toward a balanced allocation. The market can do so by endogenously adjusting worker purchasing power through their ratios. An immediate application of Kakutani's fixed point theorem shows that there are market-clearing prices, but the technical challenge is finding ratios and prices that simultaneously balance allocations and clear the market. Such ratios exist because, if they did not, some worker would always over-demand a task in equilibrium, but her purchasing power can be arbitrarily small such that the task she over-demands must be infinitely cheap. Every other worker must also be over-demanding that task for a low enough price, but this is impossible because the prices could not be market-clearing.

\begin{theorem}\label{thm:first-welfare}
    \textit{(First Welfare Theorem)} Any allocation $q^*$ for a balanced equilibrium satisfying market clearing is interim efficient.
\end{theorem}

The proof is standard as noted above. Notably, a simple reduction of this model to one period with equal weights and arrivals also demonstrates that interim efficient and interim balanced one-to-one matchings of agents to objects exist. To my knowledge, both the general existence of interim efficient and interim balanced allocations and this matching-theoretic result are novel in the literature. \qcite{miralles2021} show that any interim efficient random matching can be supported as a pseudomarket equilibrium for some budgets, but that does not imply that an interim efficient and interim balanced allocation exists. Theorem \ref{thm:existence} proves that one does exist. Theorems \ref{thm:existence} and \ref{thm:first-welfare} provide a foundation for the \textit{Dynamic Pseudomarket Mechanism (DPM)}. DPM is the following algorithm: \\~\

\textit{Algorithm: Dynamic Pseudomarket Mechanism} (Input: $\mathcal{M}$)
\begin{enumerate}
    \item Compute a balanced equilibrium $(p^*, r^*, q^*)$ for the market $\mathcal{M}$.
    \item Allocate an item $x$ arriving at time $t$ to $i$ with probability $q_{i,t}^x$. \\~\
\end{enumerate}

Multiple balanced equilibria can exist. This algorithm defines a set of mechanisms, and a mechanism in the set corresponds to a selection rule for balanced equilibria. Generally, I will restrict attention to anonymous selections that are invariant to permutations of agent identities $I$. If a property holds for all anonymous selections, then I will write that the property holds for DPM.

As aforementioned, ex-post balance is not always compatible with other desiderata; asymptotic ex-post balance relaxes it and only requires that allocations are balanced on average in the long run. DPM satisfies this and has good finite-time properties. I introduce a mechanism's \textit{expected imbalance} $\pi^W$ to study its finite-time imbalance. Given an ex-post allocation $a_i \sim \pi(v)$, $\pi^W$ is:
\[\pi^W = \E[F]*{ \frac{1}{|I|} \sum_{i \in I} \Big| W(a_i) - \frac{W_F}{|I|} \Big| }\]
The following positive result establishes DPM's desirability:

\begin{theorem}\label{thm:efficient}
    \textit{(Efficiency and Balance)} DPM is interim efficient, ex-post efficient, interim balanced, and asymptotically ex-post balanced. In addition: $\pi_{DPM}^W \leq \sqrt{T/|I|}$.
\end{theorem}

DPM improves allocations interim and ex-post. Interim improvements are important because ex-post efficient mechanisms can have arbitrarily poor interim allocations. A dictatorship that uniformly, randomly forces one worker to perform all tasks is ex-post efficient, but every worker receives the same interim allocation which allows for large gains from trade. Ex-post efficiency is important for DPM to improve realized allocations; otherwise, workers could potentially never perceive any gains.

Although it may be surprising that DPM is ex-post efficient even with arrival uncertainty, this is possible because workers do not have demand constraints. Ex-post balanced allocation is similar to treating the problem like a matching market with multi-unit demand constraints that cannot be violated. When the mechanism allocates a task to $i$ at $t$, this occupies a slot in the worker's demand that the mechanism cannot regain if a better task for $i$ arrives, and this constrains the space of implementable interim allocations. Per Proposition \ref{proposition:expost}, ex-post balance subtracts the entire set of interim efficient allocations. Interim balance relaxes this and permits the mechanism to implement interim efficient allocations, and \textit{all} interim efficient allocations have a special combinatorial structure that guarantees ex-post efficiency. The proof is a generalization and simplification of \qcite{benade2024fair} to allow for non-identical distributions over time. An example illustrating the intuition is in Appendix \ref{appendix:example-expost}.

Under DPM, caseworkers would receive the same number of cases, on average, in the long-run as with Rotation under non-equal weights and time-invariant arrivals. Theorem \ref{thm:efficient} also bounds DPM's short-run imbalance sublinearly in the average number of tasks per worker. For example, the median U.S. caseworker handles 55 cases per year or 13.75 cases per quarter (\qcitenp{edwards2018}). If $T$ is one quarter, her expected deviation from balance is approximately $\sqrt{13.75} \approx 3.7$ cases at most.

\subsection{Strategic Incentives}

DPM's efficiency relies on truthful reports, but no efficient and balanced mechanism guarantees strategyproofness.

\begin{proposition}\label{proposition:impossibility}
    \textit{(Manipulability)} An ex-post efficient, asymptotically balanced, and strategyproof mechanism does not exist.
\end{proposition}

Of course, this immediately implies that DPM is not strategyproof, and this means that there is no mechanism that can efficiently improve workers' allocations, maintain balance, and incentivize truthtelling. All three properties are necessary---there exists an ex-post efficient and strategyproof mechanism, and there exists an asymptotically balanced and strategyproof mechanism. An example of the former is dictatorship: one worker is chosen to do all tasks at the beginning of time. An example of the latter is Weighted Rotation (Appendix \ref{appendix:weighted-rotation}). The problem is that ex-post efficiency and asymptotic ex-post balance together pin down the allocation. Efficiency gives tasks to the workers with the highest relative valuations for them, and balance prevents the mechanism from compensating workers with more or fewer tasks. A worker's reported relative valuations for $x$ over $y$ determine how much $x$ she receives, so she can gain by exaggerating her preference for $x$.

Nevertheless, DPM has probabilistic incentive properties. The presence of manipulability occurs because the worst-case preference distribution can generate large gains from manipulation. I instead consider the asymptotic probability that an agent has a profitable manipulation in i.i.d.\ markets. A task-replica market $\mathcal{M}^\ell = (I^\ell, X^\ell, T, F^\ell, w^\ell, v^\ell)$ is a set of agents $I = \{1, 2, ..., h(\ell)\}$ for some finite, strictly increasing function $h$, a set of objects $X^\ell := \{1, 2, ..., \ell\}$ with weights $w^{\ell,x}$, a fixed horizon $T$\footnote{This can also scale so that tasks-per-worker do not change in $\ell$. Results are the same. I focus on scaling $|X|$ and $|I|$ to show that the results do not depend on long horizons.}, arrivals $F^\ell$, and valuations $\{v_i^\ell\}_{i \in I}$ such that $\{v_i^\ell(x)\}$ are random variables satisfying the conditions in the Preliminaries and are also i.i.d.\ across $(i,x)$. i.i.d.\ across $i$ is without loss of generality for studying a random worker sample from the firm-level worker distribution. i.i.d.\ across $x$ is for simplicity in exposition. One can allow popularity differences and correlations under the much weaker assumption that all tasks have an arbitrarily small, bounded-below probability of being a worker's $K$-th favorite task conditional on her $K-1$ favorite tasks, together with regularity conditions (Appendix \ref{appendix:general}).

The mechanism $\pi$ applied to preferences $v^\ell$ in the market $\mathcal{M}^\ell$ is $\pi^\ell(v)$. I impose the following assumption:
\begin{assumption}\label{assumption:equal}
    \textit{(Weak Diversity)}
    \begin{enumerate}
        \item $(\ell/T) \sum_{t=1}^T F_t^\ell(x) \le \bar\rho$ for some $\bar\rho < \infty$, for all $x \in X^\ell$ and $\ell$.
        \item $\ell \le c h(\ell)$ for some $c > 0$, for all $\ell$.
        \item $w^{\ell,x} = 1$ for all $x \in X^\ell$ and $\ell$.
    \end{enumerate}
\end{assumption}
This is a new condition that I introduce for a large, random market. The first condition means that no task type dominates the arrival distribution as $\ell$ increases. The second condition requires that the amount of preference heterogeneity (task types) does not \textit{grow} faster than the number of workers, though it can have a \textit{level} $c > 1$ that is larger. For example, two caseworkers in an agency may only have preference heterogeneity in a limited number of dimensions: the child's age and gender. Twenty caseworkers in an agency may have heterogeneity in many more dimensions: geographic distance, disability status, language differences, and more. These dimensions must not grow faster than the size of the firm. The last condition restricts attention to balancing the (expected) number of tasks. This describes many agencies in the survey evidence from the running example. This collection of conditions is \textit{weak diversity} because, together, they ensure that task heterogeneity grows at a controlled rate.

A worker's normalized utility from her most profitable unilateral manipulation under DPM in a task-replica is:
\[\hat u_i^\ell := \sup_{\hat v_i \in \mathbb{R}^{|X|}_{++}} \frac{|I|}{T} \left[ u_i(\pi_{DPM}^\ell(\hat v_i, v_{-i})) - u_i(\pi_{DPM}^\ell(v)) \right] \]
Under these conditions:
\begin{proposition}\label{proposition:approximate}
    \textit{(Approximate Incentive Compatibility)} Under Assumption \ref{assumption:equal}, for all $\eta, \delta > 0$:
    \[\Pr(\hat u_i^\ell > \eta) < \delta\]
    for all $\ell$ large enough.
\end{proposition}

In large markets satisfying weak diversity, the probability that a worker can gain more than any fixed amount by manipulating vanishes. Truthfulness is approximately an equilibrium. The result does not rule out manipulation by a small mass of workers, and it relies on equal weights. Under unequal weights, a worker could obtain more tasks of lower weight by manipulating, and simulations confirm that manipulation gains can remain large (Section \ref{section:strategic}). The result follows from the proof of Theorem \ref{thm:general}: in every interim efficient and interim balanced allocation, workers receive almost all of their probability mass from their top $K$ tasks. A worker already attains near-maximum utility under truthful reporting. Section \ref{section:general} discusses this argument and its relationship to \qcite{che2018}\footnote{In Appendix \ref{appendix:asymptotic}, I discuss why relaxations of strategyproofness that rely on markets with a large number of agents would not hold. This includes \qcitenp{azevedo2019strategy}'s strategyproofness in-the-large and \qcite{he2018pseudo}'s asymptotic incentive compatibility. Increasing the number of workers \textit{and} tasks is necessary for this positive result.}.

%% file: Paper/Statistic.tex
\section{Connecting Efficiency to Productivity}

DPM's effect on productivity is not immediately clear. Workers might prefer tasks in which they are more productive, or they might prefer tasks for reasons orthogonal to productivity, such as exerting less effort. DPM also does not Pareto dominate Rotation, and balance shapes which allocation it selects. Consequently, preferences that are positively correlated with productivity do not immediately imply that DPM increases productivity. Estimating a mechanism's causal effects typically requires sophisticated econometrics or costly field experiments. However, I prove that DPM's effect depends primarily on the relationship between preferences and productivity and can be measured with cheap data. I propose a bound that formalizes this idea. It is directionally correct with two task types and asymptotically conservative with many. Simulations reaffirm the theoretical results. All proofs are relegated to Appendix \ref{appendix:reduced-form}.

\subsection{Setup}

Each worker $i$ has random valuation $v_i(x)$ for a task $x \in X$ and random productivity $Y_{i,t}^x = \beta_i^x + s_{i,t}^x$ for a task $x$ allocated at time $t$ with $Y_{i,t}^x \ge 0$ normalized. $\beta_i^x$ is $i$'s interim ability at a task $x$ known at the time of preference reporting. $s_{i,t}^x$ is a mean-zero random shock that realizes at each $t$ which can be arbitrarily correlated between- or within-worker\footnote{For example, a case (task) in foster care could have randomly missing information about the child's circumstances. A case with missing information is impossible for every worker, and a case with information is trivial for every worker. $s_{i,t}^x$ represents whether a case has missing information or not, and it is perfectly correlated across caseworkers. Non-mean zero shocks can be normalized: $Y_{i,t}^x = \tilde \beta_i^x + \tilde s_{i,t}^x$ where $\tilde \beta_i^x = \beta_i^x + \E*{s_{i,t}^x}$ and $\tilde s_{i,t}^x = s_{i,t}^x - \E*{s_{i,t}^x}$.}. I denote the marginal productivity for task $x$ as $\E[\beta,s]*{Y_{i,t}^x} := g^x$. The firm's expected productivity under $\pi$ is denoted $\mathcal{P}(\pi)$. I assume a linear specification for expected productivity:
\begin{assumption}\label{assumption:linear}
    \textit{(Linear Specification)}
    \[\mathcal{P}(\pi) = \E*{\frac{1}{T} \sum_{t = 1}^T \sum_{x \in X} \sum_{i \in I} a_{i,t}^x Y_{i,t}^x}\]
    where $a$ is an ex-post allocation of $\pi$.
\end{assumption}
This rules out non-linear behavioral responses to the allocation\footnote{Per the running example, larger caseloads might decrease overall productivity. The effect could be convex or concave: the penalty for increasing from ten to twenty cases could be smaller than the penalty for increasing from fifty to one-hundred or vice versa. In practice, most caseworkers might receive between ten and thirty cases, and the caseload effect could be mostly linear on that domain. This model requires that such effects can be approximated as linear in this way.},  but it accommodates any effect that is linear in the allocation such as more tasks adding workload burden and decreasing total productivity (Appendix \ref{appendix:linearity}). The assumption is also testable under the status quo (Section \ref{section:managerial}).

As in the Preliminaries, worker preferences $v_i(x)$ are random variables that can have an arbitrary relationship to productivity. I assume that $\{(v_i(x), \beta_i^x)\}$ is i.i.d.\ across $(i,x) \in I \times X$, that $\{s_{i,t}^x\}$ is independent of $(v, \beta)$ and identically distributed across $t$, and that $\log v_i(x)$ has a log-concave density function for all $(i,x) \in I \times X$. This includes many continuous distributions (gaussian, uniform, exponential, and more). I extend the bound to non-identical distributions and correlation across $x$ under weak additional assumptions (Appendices \ref{appendix:solution} and \ref{appendix:general}).

I maintain the following additional condition:
\begin{assumption}\label{assumption:non-decreasing}
    \textit{(Positive Monotonic Correlation)} $\CE[v_i(x)]{Y_{i,t}^x}$ is strictly increasing in $v_i(x)$ for all $i,t,x$
\end{assumption}
One can also maintain that it is strictly decreasing, and results are symmetric. The important regularity condition is that the conditional expectation is monotonic. This means that workers who prefer a task more intensely are (on average) worse or better at it across the entire preference distribution. This assumption is also testable (Section \ref{section:managerial}).

The timing still includes three stages. Here, I now abstract away from strategic incentives:
\begin{enumerate}
    \item \textit{Ex-Ante:} Valuations $v$ have not yet realized.
    \item \textit{Interim:} Valuations $v$ and interim abilities $\beta$ realize. Each worker $i$ observes $v_i$ and $\beta_i$, simultaneously reports $v_i$ to $\pi$, and $\pi(v)$ is observed by all workers.
    \item \textit{Ex-Post:} Each arrival $\tilde{a}_t^x$ realizes, some $a_t \in \Supp{\pi_t(v)}$ realizes, and $Y_{i,t}^x$ simultaneously realizes at each $t$ for all $(i,x)$.
\end{enumerate}
The firm receives $\sum_{i \in I} \sum_{t = 1}^T \sum_{x \in X}  a_{i,t}^x Y_{i,t}^x$ at the end of time $T$.

\subsection{Data}

I assume the firm has access to cross-sectional productivity and ordinal preferences data. A worker's ordinal preference, fixing $(x,y)$ is $V_i^{x,y} = \mathbbm{1}\{v_i(x) > v_i(y)\}$. An $(x,y)$ \textit{dataset} is a three-column matrix $\mathcal{Y}^{x,y} = \{\E[s]*{Y_{i,t}^x}, \E[s]*{Y_{i,t}^y}, V_i^{x,y}\}_{i \in I}$. Formally, given a single $(x,y)$ dataset, the five moments that the bound requires are $\Pr(V_i^{x,y} = 1) := G$, $\E[\beta,s]*{Y_{i,t}^x|V_i^{x,y} = 1} := \tilde g^x$, $\E[\beta,s]*{Y_{i,t}^y|V_i^{x,y} = 1} := \tilde g^{x,y}$, $\E[\beta,s]*{Y_{i,t}^y|V_i^{x,y} = 0} := \tilde g^y$, and $\E[\beta,s]*{Y_{i,t}^x|V_i^{x,y} = 0} := \tilde g^{y,x}$. These are the probability that the average worker prefers $x$, the average worker's expected productivity in $x$ ($y$) given a preference for $x$ ($y$), and the average worker's expected productivity in $y$ ($x$) given a preference for $x$ ($y$). To avoid degeneracy, I assume $G \in (0,1)$. 

I say that a function $f$ is \textit{estimable} if it is a known function of $G$, $\tilde g^x$, $\tilde g^{x,y}$, $\tilde g^y$, $\tilde g^{y,x}$, and primitives ($|I|$, $F$, and $w$) for an arbitrary set $\{\mathcal{Y}^{x,y}\}$ of $(x,y)$ datasets. A firm that uses Rotation, if it records outcomes, necessarily has productivity data because Rotation is equivalent to random assignment\footnote{The firm must strictly adhere to Rotation or another random assignment process. As mentioned in the running example, some agencies do, and some do not. This section assumes that Rotation is strictly followed.}. The expectation of a worker's outcome in a task type $x$ is her average productivity uncontaminated by selection bias. The remainder of the data that is necessary is $V_i^{x,y}$ for each $i \in I$. 

The firm can estimate these ordinal preferences with cross-sectional surveys. They could administer a single question survey to a small sample of workers: ``What is your preference ranking over task types $X$?'' In each sub-market restricted to two types $x$ and $y$, $\tilde g^x$ is the average productivity for $x$-cases among workers who have a preference for $x$ over $y$, and $\tilde g^{x,y}$ is the average productivity for $y$-cases among the same workers. This survey would complete every possible $(x,y)$ dataset. In principle, this should also have lower fixed and per-subject costs than a field experiment.

\subsection{Binary Case}

Obtaining a closed-form solution for DPM allocations in arbitrary markets is generally intractable. This precludes characterizing DPM's expected productivity. Instead, I pursue a strategy that allows me to derive a \textit{directional} test for DPM's expected productivity using an assumption of only two task types. I show that the bound remains a valid directional test even when there are more types.
\begin{assumption}\label{assumption:types}
    \textit{(Binary Types)} $|X| = 2$
\end{assumption}
This restriction allows the bound to only depend on an $(x,y)$ dataset.
\begin{theorem}\label{thm:expected}
    \textit{(DPM's Expected Productivity)} Under Assumptions \ref{assumption:linear}--\ref{assumption:types}, there exists an estimable bound $\underline{\mathcal{P}}$ such that:
    \[\mathcal{P}(\pi_{DPM}) \ge \underline{\mathcal{P}} \ge \mathcal{P}(\pi_{R})\]
\end{theorem}
This Theorem states that (i) when preferences and productivity are positively monotonically correlated, DPM increases productivity compared to Rotation and (ii) I derive a bound, $\underline{\mathcal{P}}$, that correctly predicts a (weakly) positive effect. The inequalities are symmetric under negative monotonic correlation. This means that the bound attenuates the true effect toward zero but is a valid directional prediction.

The proof is in Appendix \ref{appendix:solution}, but the intuition is straightforward. Under Assumption \ref{assumption:types}, in DPM allocations, workers receive a score based on their relative preference for $x$ and must beat a fixed threshold to receive $x$. Expected productivity for workers who receive $x$ is then the productivity of the ``top few'' workers who receive $x$. Technically, it is a sum of order statistics associated with the joint distribution of preferences and productivity. $\underline{\mathcal{P}}$ then restricts to only considering above-median workers; ``top few'' workers become equivalent to ``better than median'' workers. I show that the former group has weakly higher expected productivity than the latter under Assumption \ref{assumption:non-decreasing}.

$\underline{\mathcal{P}}$ has a clean form when $|X| = 2$, $\tilde g := \tilde g^x = \tilde g^y$, $\tilde g^* := \tilde g^{x,y} = \tilde g^{y,x}$, $G = 1/2$, $w^x = w^y$, and $F(x) = F(y)$:

\[\underline{\mathcal{P}} - \mathcal{P}(\pi_R) = \left(\tilde g - \tilde g^*\right) \cdot C(F, w, |I|, G)\]

for a weakly positive function $C$ that depends on primitives and $G$\footnote{This function is known. I omit the details to focus on the core intuition; see Appendix \ref{appendix:solution}'s Corollary \ref{corollary:bound} for the exact form.}. DPM's expected gain over Rotation depends on what $v_i(x) > v_i(y)$ signals. The difference $\tilde g - \tilde g^*$ measures how much the preference signals that \textit{the worker has a comparative advantage in $x$ over $y$}. When $\tilde g > \tilde g^*$, it is more likely that the worker is better at $x$ than $y$. When $\tilde g = \tilde g^*$, the preference is uninformative. When $\tilde g < \tilde g^*$, the preference is worse than uninformative in that it is more likely that the worker is better at $y$ than $x$. 

This captures the intuition that more positively correlated preferences and productivity should increase DPM's effect on productivity. While this is sensible, as noted, it is not trivial because DPM does not Pareto dominate Rotation. Balance also complicates this further. $C(\cdot)$ represents the effect of market primitives through choosing a particular balanced allocation. This result shows that it can distort expected productivity but can never reverse its sign. The next sub-section extends the comparison to markets with more task types.

\subsection{General Case}\label{section:general}

Comparing DPM's expected productivity in \textit{any} markets with fewer versus more task types leads to pathological cases where it is impossible for $\underline{\mathcal{P}}$ to be an accurate prediction. For example, if workers have very noisy preferences for additional task types ($|\tilde g - \tilde g^*|$ decreases on average), then DPM's gain over Rotation mechanically moves toward zero. As in the strategic results, I consider task-replica markets with weak diversity. A task-replica market with productivity is $\mathcal{M}^\ell = (I^\ell, X^\ell, T, F^\ell, w^\ell, v^\ell, Y^\ell)$ satisfying the setup conditions for each $\ell$. $|I| = h(\ell)$ for some strictly increasing function $h$, and $|X| = \ell$. $\pi$ applied to preferences $v^\ell$ in the market $\mathcal{M}^\ell$ is $\pi^\ell(v)$. I write $\pi^{2,\ell}(v)$ for $\pi$ in a worker-replica market $\mathcal{M}^\ell$ with $|I| = h(\ell)$ and $|X| = 2$. That is, the number of workers grows, but the number of task types is fixed to two. All other assumptions remain the same in these markets.

\begin{theorem}\label{thm:general}
    \textit{(General Robustness)} Under Assumptions \ref{assumption:equal}--\ref{assumption:non-decreasing}:
    \[\mathcal{P}(\pi^\ell_{DPM}) > \mathcal{P}(\pi^{2,\ell}_{DPM})\]
    for all $\ell$ large enough.
\end{theorem}

The inequality is symmetric under negative monotonic correlation. Fixing a large $|I|$, having many task types rather than two (up to some constant factor of $|I|$ by weak diversity) increases DPM's true expected productivity and so increases the wedge between it and $\underline{\mathcal{P}}$. Despite the fact that exactly predicting DPM's effect on productivity is generally intractable, the bound is asymptotically correct for the directional effect in large markets since it remains bounded above:
\[\mathcal{P}(\pi^\ell_{DPM}) > \mathcal{P}(\pi^{2,\ell}_{DPM}) \ge \underline{\mathcal{P}}\]
The linear specification, monotonicity, and large markets are required for this asymptotic guarantee. The linear specification is testable using a simple regression with Rotation data, and the monotonicity assumption is testable with preference data, although it is a mild regularity condition that appears reasonable. I summarize these tests in the managerial implications (Section \ref{section:managerial}). Section \ref{section:simulations} evaluates the bound in small markets and with unequal weights. 

\begin{figure}[t]
\centering
\caption{Pairs}
\label{figure:pairs}
\begin{tikzpicture}[
    worker/.style={circle, draw, minimum size=7mm},
    task/.style={rectangle, draw, minimum size=7mm},
    node distance=8mm
]
\node[worker] (i) {$i$};
\node[worker, below=of i] (j) {$j$};
\node[worker, below=of j] (k) {$k$};
\node[task, right=35mm of i] (x) {$x$};
\node[task, below=of x] (y) {$y$};
\node[task, below=of y] (z) {$z$};
\node[above=2mm of i, font=\footnotesize] {Workers};
\node[above=2mm of x, font=\footnotesize] {Tasks};
\draw (i) -- (x);
\draw[red, very thick] (i) -- (y);
\draw (j) -- (y);
\draw (j) -- (z);
\draw (k) -- (z);
\node[below=8mm] at ($(k)!0.5!(z)$) {Pairs: $\{(i,x),\ \textcolor{red}{(i,y)},\ (j,y),\ (j,z),\ (k,z)\}$};
\end{tikzpicture}

\vspace{0.5ex}
\footnotesize
\textbf{Note:} Each edge, a worker and a task she consumes with positive probability, is a pair.
\end{figure}

\subsubsection*{Proof Sketch}

The full proof is in Appendix \ref{appendix:general}. Here, I present a sketch. The result relies on a broader theoretical finding: under Assumption \ref{assumption:equal}, \textit{every} interim efficient and interim balanced allocation has the average worker's expected welfare converging to the maximum possible welfare. \qcite{che2018} originally showed that every ex-post efficient matching mechanism has welfare converging to the maximum. Their result applies in matching markets where each agent either receives an object or does not, objects have fixed capacity, objects can have priorities, and the matching is ex-post Pareto efficient. Here, (i) agents (workers) have fractional allocations, (ii) agents must obtain balanced allocations, (iii) tasks do not have priorities, and (iv) the allocation is \textit{interim efficient}. Regarding (ii): in the special case when $F_t(x)$ is equal for all $(x,t)$ and $w^x = 1$ for all $x$, interim balance requires that $\sum_{t=1}^T \sum_{x \in X} q_{i,t}^x = T|X|/|I|$. Each task type has $T$ units, so this is the classical multi-unit demand constraint where each worker demands $T|X|/|I|$ units, which scales with $T$. Since utilities are strictly positive and total supply equals total demand, the demand constraint binds with equality in every efficient matching. This implies that my results extend to this case, which still differs from \qcite{che2018} because their model assumes that each agent only receives one object at most.

The proof uses similar intuition but different techniques. In their setting, any Pareto efficient matching is the outcome of a Serial Dictatorship. This is false in the present model. However, it is still true that interim efficient allocations obey a \textit{trading cycle property} and an \textit{inefficient set property}. Because fractional allocations permit workers to own many tasks, both properties are stated for \textit{pairs}, where a pair is a worker and a task that she consumes with positive probability (Figure \ref{figure:pairs}). The trading cycle property states that, in any set of pairs, some worker consumes positive probability of her favorite task among the tasks in that set. If not, each pair could point to a pair holding its worker's favorite task, forming a cycle. Workers would exchange arbitrarily small probability along the cycle to Pareto improve, contradicting interim efficiency (Figure \ref{figure:cycle}).

\begin{figure}[t]
\centering
\caption{Trading Cycle Property}
\label{figure:cycle}
\begin{tikzpicture}[
    pair/.style={rectangle, draw, rounded corners, minimum width=14mm, minimum height=7mm},
    point/.style={-latex, red, thick}
]
\node[pair] (ix) {$(i,x)$};
\node[pair, below right=14mm and 10mm of ix] (jy) {$(j,y)$};
\node[pair, below left=14mm and 10mm of ix] (kz) {$(k,z)$};
\draw[point] (ix) -- node[above right, font=\footnotesize] {$i$ wants $y$} (jy);
\draw[point] (jy) -- node[below, font=\footnotesize] {$j$ wants $z$} (kz);
\draw[point] (kz) -- node[above left, font=\footnotesize] {$k$ wants $x$} (ix);
\end{tikzpicture}

\vspace{0.5ex}
\raggedright
\footnotesize
\textbf{Note:} No worker holds her favorite among tasks held by these pairs. Each pair points to the pair holding its worker's favorite task. Every worker is strictly better off by reallocating along the cycle, contradicting interim efficiency.
\end{figure}

The trading cycle property implies the \textit{inefficient set property}. Call a pair \textit{inefficient} if its task is not among the worker's top $K$ tasks. The trading cycle property orders the inefficient pairs so that no worker's top $K$ tasks are held by her pair or any pair after it. Splitting this \textit{inefficient set} into halves by mass, no worker in the top half has any of her top $K$ tasks held by the bottom half (Figure \ref{figure:inefficient}). These are analogous to \qcite{che2018}'s independent sets, and they complete their proof through invoking a random graph theoretic result to demonstrate convergence. The same strategy does not straightforwardly apply here due to fractional allocations complicating the graph-theoretic representation.

Instead, I use a counting argument. If inefficient pairs hold a large share of total mass, then both halves are large. Interim balance gives every worker the same mass, so the top half contains many distinct workers. Assumption \ref{assumption:equal} prevents any task type from arriving disproportionately often, so the bottom half holds many distinct task types. Hence, many workers must have top $K$ sets that avoid many task types. These sets are endogenous to the allocation, but they must be one of finitely many possible collections of workers and task types. For any fixed collection, preference heterogeneity makes this avoidance exponentially unlikely, and the probability shrinks faster than the number of collections grows. A union bound then shows that the expected share of a worker's allocation outside her top $K$ tasks is of order $(\log K)/K$ in large markets\footnote{Interim efficiency uses cardinal preferences over probability distributions, but the trading cycle property discards the cardinal information and only uses ordinal valuations. \qcite{ortega2026remarkable} demonstrates that the Probabilistic Serial mechanism, which is ordinally efficient over probability distributions---a weaker property than interim efficiency---has near-constant expected rank. I conjecture that the trading cycle argument could be tightened to yield a near-constant or constant expected rank bound for interim efficient allocations. This is an interesting open question.}. 

The proof concludes through using the fact that, for large enough $K$ independent of $|I|$ and $|X|$, workers asymptotically receive almost only their top $K$ tasks. i.i.d.\ valuations imply that a worker's $K$-th favorite task is arbitrarily near her bliss point for large enough $|X|$. Given positive monotonic correlation between preferences and productivity, each worker is allocated tasks in which she is very productive (in expectation).

\begin{figure}[t]
\centering
\caption{Inefficient Set}
\label{figure:inefficient}
\begin{tikzpicture}[x=7mm, y=-7mm]
\fill[gray!25] (5,0) rectangle (8,6);
\draw[gray!60] (0,0) grid (8,6);
\foreach \r in {1,...,6}
    \node[left, font=\footnotesize] at (0,\r-0.5) {$i_\r$};
\foreach \c in {1,...,8}
    \node[above, font=\footnotesize] at (\c-0.5,0) {$x_\c$};
\foreach \r/\c in {1/1, 1/2, 2/2, 2/3, 3/1, 3/4}
    \fill[red] (\c-0.5,\r-0.5) circle (1.5mm);
\foreach \r/\c in {1/4, 2/5, 3/5, 4/6, 5/7, 6/8}
    \fill (\c-0.5,\r-0.5) +(-1.5mm,-1.5mm) rectangle +(1.5mm,1.5mm);
\draw[red, dashed, thick] (-0.8,3) -- (8.3,3);
\node[right, font=\footnotesize] at (8.3,1.5) {Top half};
\node[right, font=\footnotesize] at (8.3,4.5) {Bottom half};
\fill[red] (0.5,7) circle (1.5mm);
\node[right, font=\footnotesize] at (0.8,7) {top $K$};
\fill (3.2,7) +(-1.5mm,-1.5mm) rectangle +(1.5mm,1.5mm);
\node[right, font=\footnotesize] at (3.5,7) {held};
\fill[gray!25] (5.6,6.75) rectangle (6.1,7.25);
\node[right, font=\footnotesize] at (6.2,7) {held by bottom half};
\end{tikzpicture}

\vspace{0.5ex}
\footnotesize
\textbf{Note:} Pairs are ordered by the trading cycle property, so no red dot above the line lies in a shaded column. At $K = 2$, none of $i_4, i_5, i_6$ hold any of $i_1, i_2, i_3$'s top two tasks.
\end{figure}

\subsection{Estimation}

How can one estimate the bound when $|X| > 2$? In my running example, the task types are not binary. Caseworkers have preferences over geographic location, child age, competency, previous experiences, and more. How can $\mathcal{P}(\pi_{DPM})$ be estimated in environments that do not satisfy the theoretical assumptions? If $|X| > 2$, I propose randomly sampling sub-markets that each contain two tasks to estimate $\underline{\mathcal{P}}$, then take a pair-weighted average over the estimates. I abuse notation and denote this average as $\underline{\mathcal{P}}$ (see Appendix \ref{appendix:estimation} for details). A second interpretation of Theorem \ref{thm:general} is that this approach is asymptotically valid because the true effect in large, two task markets is smaller than the true effect in large, many task markets. Section \ref{section:simulations} evaluates its accuracy in small samples.

%% file: Paper/Simulations.tex
\section{Simulations}\label{section:simulations}

The previous theoretical results leave a few questions unanswered. First, does the bound accurately predict DPM's expected productivity in small markets under varying primitives? Second, how tight are the asymptotic approximations for strategic incentives in large-worker, large-task markets? I describe the simulation setup, then I address each of these questions.

\subsection{Setup}

I simulate $\underline{\mathcal{P}}$ versus DPM's realized productivity under various market parameters. Each simulated market fixes a number of workers $|I|$ and task types $|X|$. Task weights are either equal or randomly drawn from a Dirichlet distribution. Arrivals are either time-invariant or randomly drawn from a Dirichlet distribution in each period. Each task $x \in X$ draws average productivity $g^x$ i.i.d.\ from a uniform distribution. Worker productivity is generated from a Gaussian copula with a single within-worker correlation parameter $\rho$:
\[
Z_i \sim \mathcal{N}(0,\Sigma), \qquad \Sigma_{x,y} = \begin{cases} 1 & x = y \\ \rho & x \neq y \end{cases}, \qquad Y_i^x = \mathbbm{1}\{Z_i^x \leq \Phi^{-1}(g^x)\}
\]
so that the marginal success probability is exactly $g^x$. Note that $\rho = 1$ does not necessarily imply that within-worker productivity is perfectly correlated since this is impossible for some levels of $g^x$ and $g^y$. 

Outcomes $Y_i^x$ are always Bernoulli. I simulate between-worker productivity shocks: with probability $z$, after preference reporting, a worker-task productivity is set to a common outcome drawn with probability $g^x$ (the same across all workers). Reported preferences equal productivity plus reporting error: $v_i(x) = Y_i^x + \xi_i^x$ (shifted by a market-level constant so that the lowest valuation is one). $\xi_i^x$ is either mean-zero Normal or Uniform noise with variance $\sigma^2$. I report noise through the preference reporting accuracy. A worker's preferences are accurate if her pairwise order for two tasks $x,y$ is the same as ordering by her productivity: $v_i(x) > v_i(y)$ implies $Y_i^x > Y_i^y$ and vice versa. I report the percentage of $(i,x,y)$ tuples with $Y_i^x \ne Y_i^y$ where the worker's preferences are accurate. To simulate negatively correlated preferences, I also set $v_i(x) = (1 - Y_i^x) + \xi_i^x$. This creates pairwise accuracy below one half.

For each environment, I generate twenty markets with varying marginals $g^x$. I sample worker outcomes, preferences, and allocations twenty times and estimate the primitives for the bound $\underline{\mathcal{P}}$ from the pooled Rotation allocations (Appendix \ref{appendix:estimation}). The total number of simulations for each environment (a point on Figure \ref{figure:prod-outcome}) is four-hundred. I compare the theoretical prediction against the market's realized productivity (averaged across market samples) computed by solving for the DPM equilibrium and Rotation allocation in the market. The y-axis value is the level of realized gain (black) versus the level of the bound $\underline{\mathcal{P}}$ (red). I use a custom algorithm to solve for DPM equilibria that converges in nearly all instances within fractions of a second (Appendix \ref{appendix:computation}).

\subsection{Small Market Performance}

\begin{figure}[t!]
    \begin{minipage}{0.8\linewidth}
    \centering
    \includegraphics[width=\linewidth]{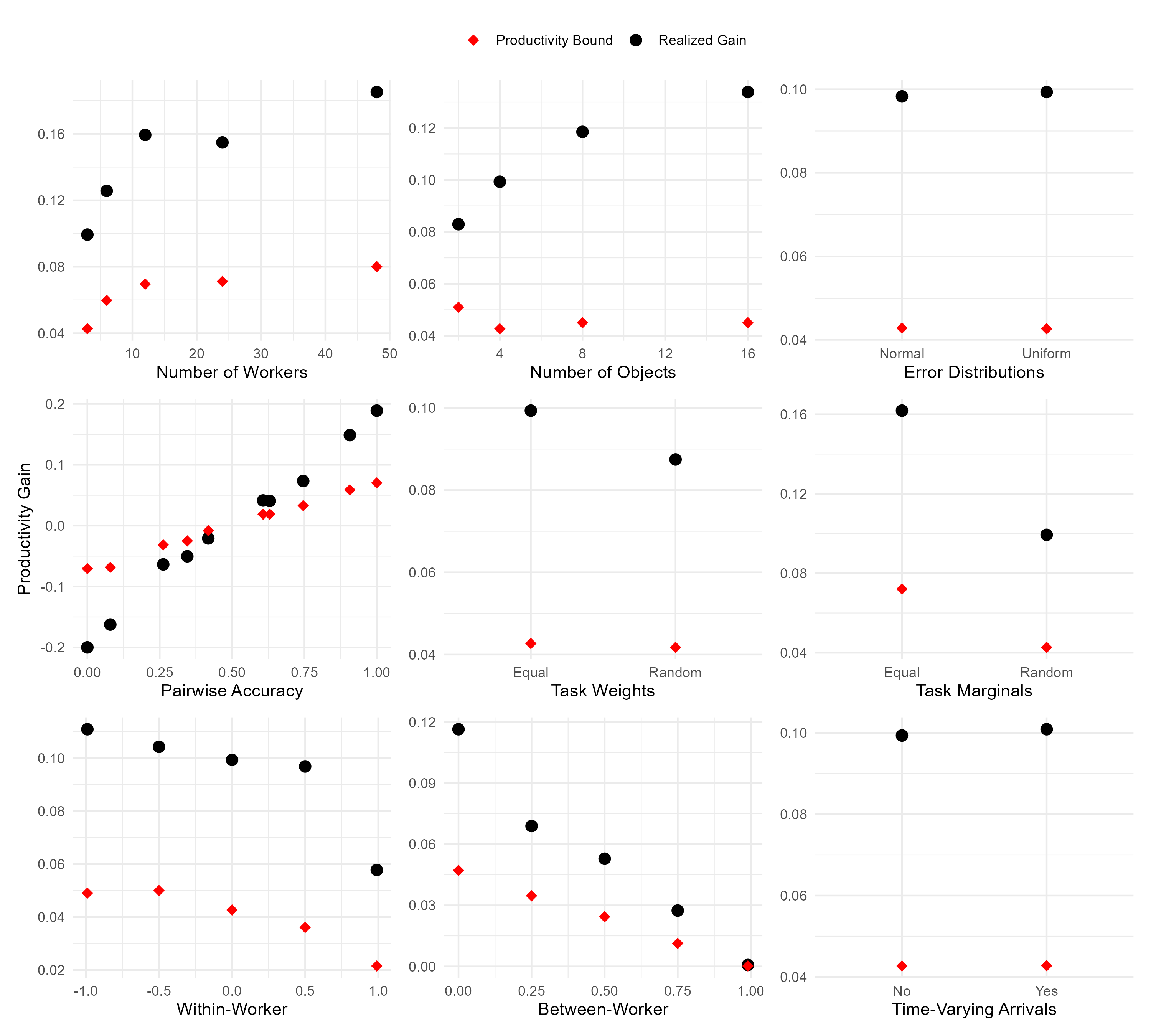}
    \caption{Productivity Effects}
    \label{figure:prod-outcome}

    \vspace{0.5ex}

    \raggedright
    \footnotesize
    \textbf{Note:} Default fixed parameters (when not varying) are $|I| = 3$, $|X| = 4$, Uniform preferences with $\sigma = 0.75$, equal weights and arrivals, within-worker correlation $\rho = 0$, between-worker correlation equal to 0.1, and time-invariant distributions.
  \end{minipage}
\end{figure}

Results are in Figure \ref{figure:prod-outcome}. These indicate that $\underline{\mathcal{P}}$ is robust when varying the number of workers, number of objects, preference reporting error, weights, task marginals, within-worker correlation, between-worker correlation, and time-varying arrivals. In most specifications, realized gains are between 10-20 pp. As theoretically predicted, the bound is systematically attenuated toward zero (Rotation productivity). Yet, it tracks the movement of realized gains indicating that, though it is conservative, it is a reliable signal of DPM's true effect. The results are symmetric when pairwise accuracy falls below 50\%. The varied primitives include within-worker correlation and random marginals where preferences are correlated and non-identical across tasks. Results are similar to the i.i.d.\ cases.

The one case where the bound poorly tracks realized gains is the one previewed in my intuition: while realized gains are increasing in $|X|$, the bound is constant. This is because randomly sampling pairs of tasks to estimate the bound estimates the gain from any worker's \textit{pairwise} comparative advantage, whereas DPM exploits combinatorial comparative advantages. Additionally, realized gains increase (though still plateauing) even when $|X| > |I|$, reflecting weak diversity's allowance of level differences.

Figure \ref{figure:power} demonstrates this method's predictive accuracy for small ($|I| \le 40$) samples. The x-axis is the number of workers, and the y-axis is the percentage of individual market simulations where the bound, estimated from that market alone, correctly predicts the direction of the productivity effect. The true effect is either positive or negative with probability one-half. I vary the preference reporting noise and the number of task types $|X|$. Even when workers' preference reports are near a coin flip (correct pairwise accuracy is 42\% or 58\%), the bound has about 90\% predictive accuracy at only $|I| \approx 20$ with six task types. Less noise (greater than 75\% pairwise reporting accuracy) requires either less than half the number of workers (five to ten) or half the number of task types (two) to achieve the same predictive power. Either more workers or more task types increase the number of pairwise preference comparisons which results in higher accuracy across all scenarios.

\begin{figure}[thpb]
    \begin{minipage}{0.8\linewidth}
    \centering
    \includegraphics[width=\linewidth]{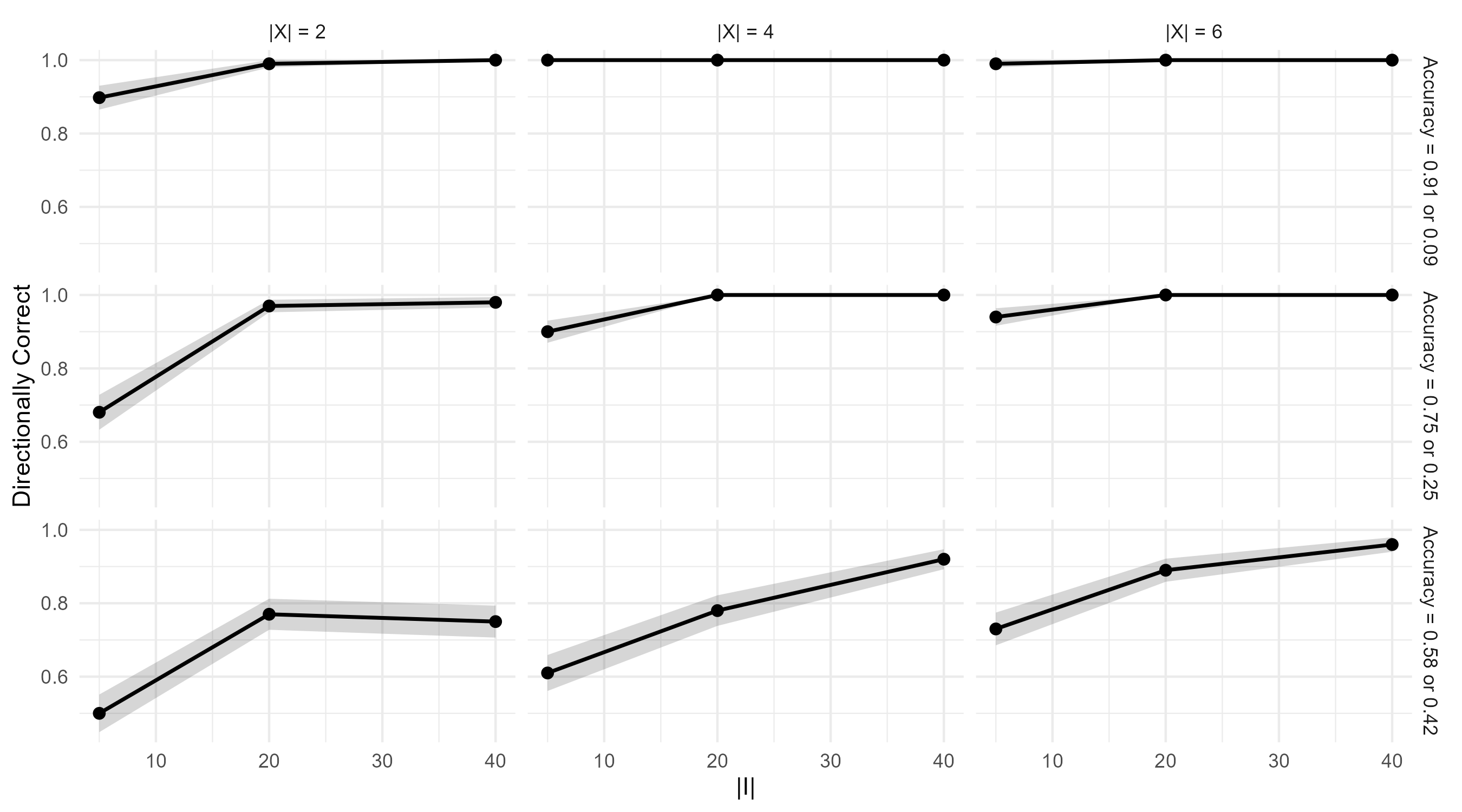}
    \caption{Predictive Accuracy}
    \label{figure:power}

    \vspace{0.5ex}

    \raggedright
    \footnotesize
    \textbf{Note:} Default fixed parameters are Uniform preferences, equal weights and arrivals, zero within-worker correlation, between-worker correlation equal to 0.1, time-invariant distributions, and constant marginals.
  \end{minipage}
\end{figure}

\subsection{Strategic Performance}\label{section:strategic}

I simulate workers' incentives to manipulate DPM under similar parameters as in Figure \ref{figure:prod-outcome}. For each scenario (a point on Figure \ref{figure:large}), I generate one-hundred markets and solve for the DPM equilibrium under truthful reports. In each market, I compute each worker's maximum utility from a balanced allocation $\bar u_i$ and each worker's minimum utility from a balanced allocation $\underline u_i$. I measure each worker's distance from her maximum utility $\bar u_i$ as:
\[
\tilde u_i = \frac{\bar u_i - u_i(\pi_{DPM})}{\bar u_i - \underline u_i}
\]
$\tilde u_i$ lies between zero and one and bounds a worker's gain from any possible manipulation as a percentage of the difference between a worker's best and worst possible welfare. A high $\tilde u_i$ does not imply that there exists any profitable manipulation that would achieve that level of utility. It is a worst case bound. A low $\tilde u_i$ implies that any possible manipulation yields relatively little gain. The y-axis value is the average $\tilde u_i$ across workers. The x-axis is $|I|$. I grow $|I|$ and $|X|$ together at three rates to test the necessity and sufficiency of weak diversity: $|X| = |I|^{1/2}$, $|X| = |I|$, and $|X| = |I|^{3/2}$ for $|I| \in \{4, 8, 16, 24, 36, 48, 60\}$ with equal or random weights. I simulate these markets with a single period because allocations are effectively identical across periods with time-invariant arrivals. 

Figure \ref{figure:utility} shows DPM's utility gain over Rotation in the same environments as Figure \ref{figure:prod-outcome}. I measure the welfare gain in units of $\tilde u_i$. As the theoretical results suggest, gains are uniformly positive and generally at least 25\% of the difference between a worker's worst and best possible balanced allocation. Pareto efficiency matters for average welfare gains over Rotation.

\begin{figure}[thpb]
    \begin{minipage}{0.8\linewidth}
    \centering
    \includegraphics[width=\linewidth]{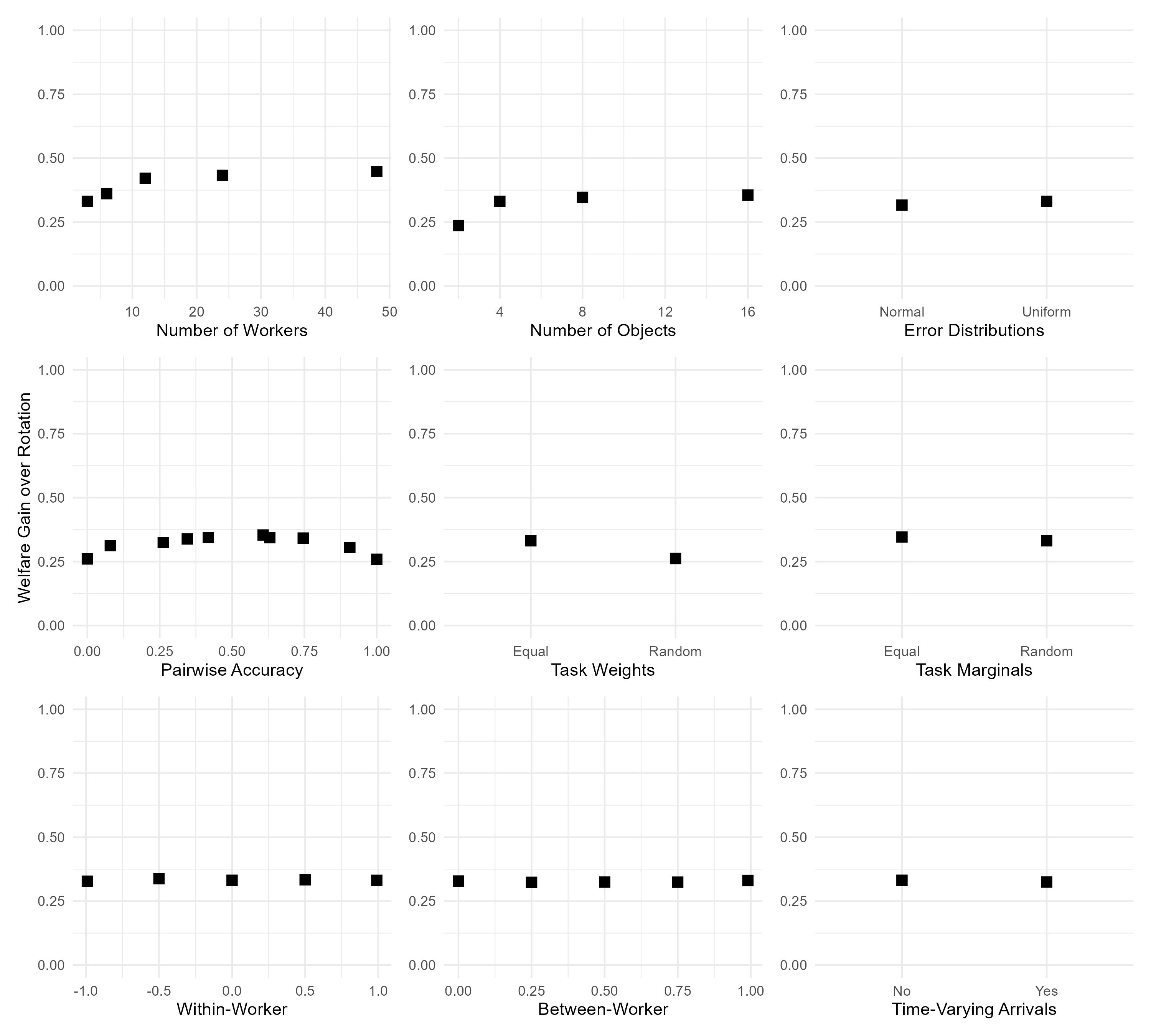}
    \caption{Welfare Effects}
    \label{figure:utility}

    \vspace{0.5ex}

    \raggedright
    \footnotesize
    \textbf{Note:} Default fixed parameters are as in Figure \ref{figure:prod-outcome}. The y-axis value is DPM's utility gain over Rotation, $(u_i(\pi_{DPM}) - u_i(\pi_R))/(\bar u_i - \underline u_i)$, averaged across workers.
  \end{minipage}
\end{figure}

Figure \ref{figure:manipulations} shows manipulations when varying the number of workers, number of objects, error distributions, task weights, task marginals, and within-worker correlation. These represent small market environments since either the number of workers or number of task types remain small. Two findings are of interest. First, even when only the number of workers grows, the gain from the best possible manipulation sharply decreases. Second, when both the number of workers and task types are small, the potential gain from manipulation is mostly constant and non-negligible across all market parameters. This suggests that strategic incentives are plausible in small markets.

\begin{figure}[thpb]
    \begin{minipage}{0.8\linewidth}
    \centering
    \includegraphics[width=\linewidth]{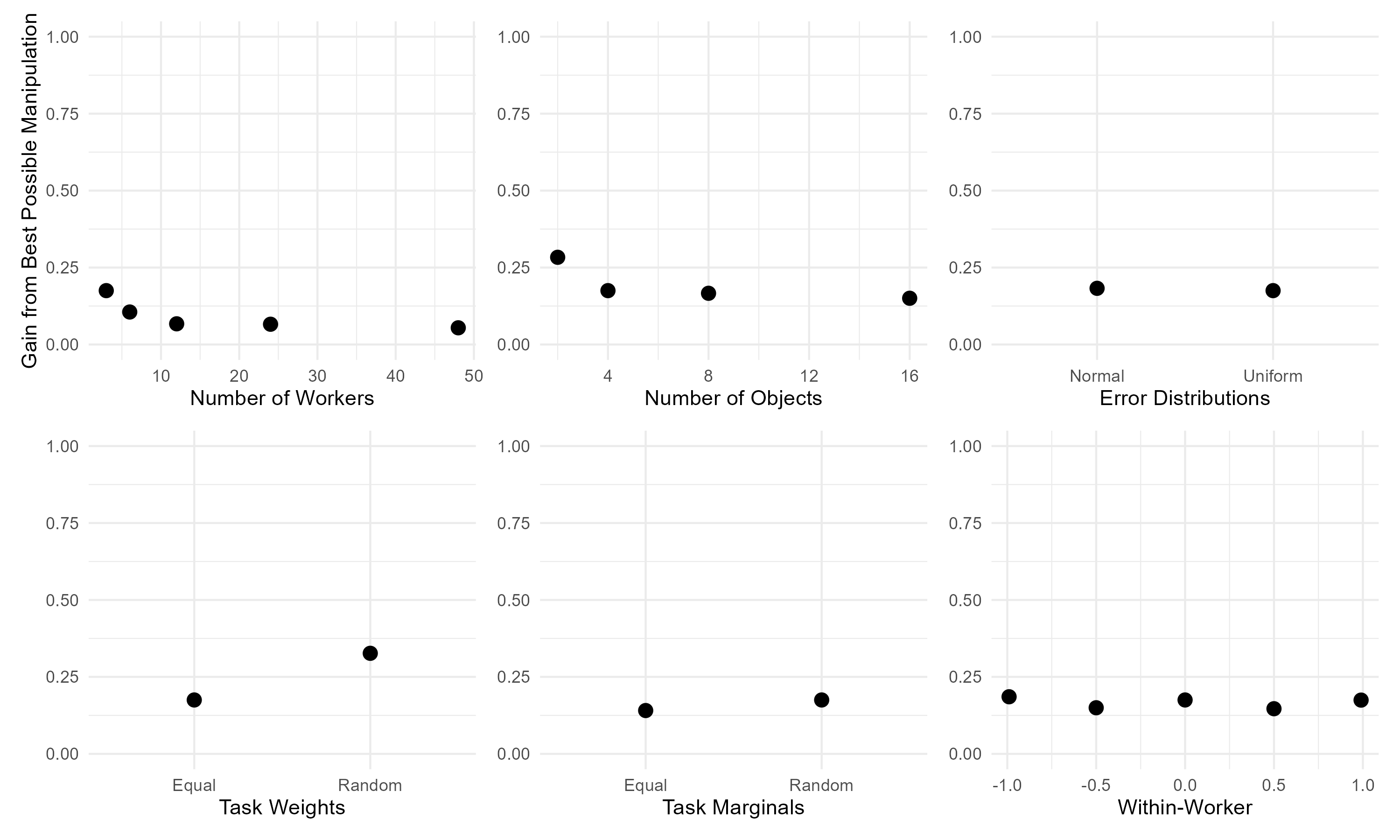}
    \caption{Manipulation in Small Markets}
    \label{figure:manipulations}

    \vspace{0.5ex}

    \raggedright
    \footnotesize
    \textbf{Note:} Default fixed parameters (when not varying) are Uniform preferences with $\sigma = 0.75$, within-worker correlation $\rho = 0$, between-worker correlation equal to 0.1, and time-invariant distributions.
  \end{minipage}
\end{figure}

Figure \ref{figure:large} shows results when the market grows. The theoretical predictions appear to be tight: the gain from the best possible manipulation declines to 7\% of the total gap when weak diversity holds and $|I| = 24, |X| \le 24$. The largest decline is from about 21\% at $|I| = 4$ to 5\% at $|I| = 60$ when task types grow slower than workers. Interestingly, even when $|X|$ grows faster than $|I|$, the gain still declines. The bottom row shows the share of workers whose possible gain exceeds 10\% of the total gap. With equal weights, the share above 10\% declines at every growth rate. This is most sharp when $|X|$ grows faster than $|I|$ (from near 50\% to under 1\%). This share's convergence to 0\% appears slower when $|X|$ has slower growth. This suggests the results may extend to faster growth rates under some conditions. Equal weights appear to be crucial for this result. Without them, the gain from the best possible manipulation could be over 50\% of the total gap, and the share of workers with a 10\% possible gain converges to 100\%. This does not decline as the market grows larger.

\begin{figure}[thpb]
    \begin{minipage}{0.8\linewidth}
    \centering
    \includegraphics[width=\linewidth]{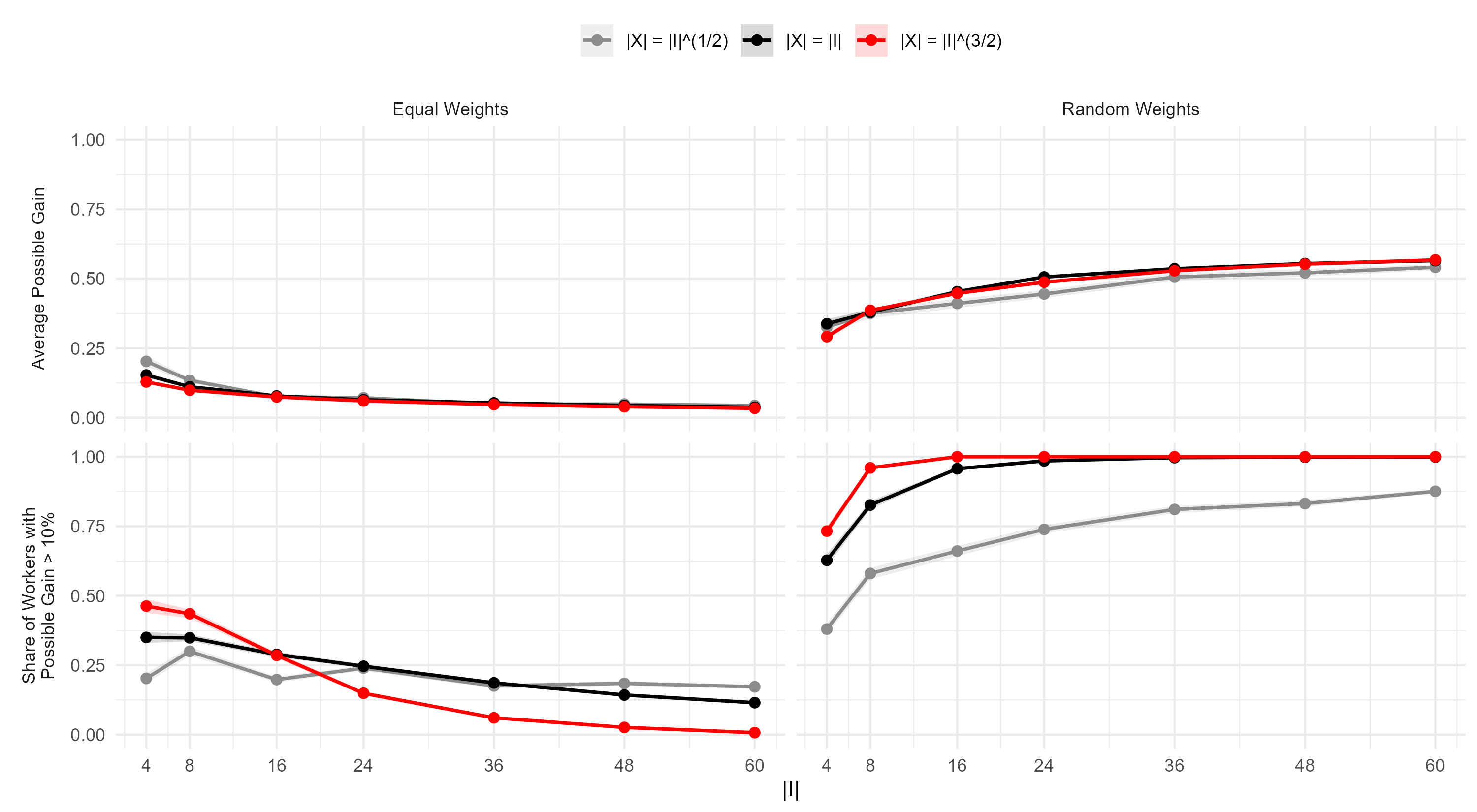}
    \caption{Manipulation in Large Markets}
    \label{figure:large}

    \vspace{0.5ex}

    \raggedright
    \footnotesize
    \textbf{Note:} Fixed parameters are Uniform preferences with $\sigma = 0.75$, equal arrivals, within-worker correlation $\rho = 0$, between-worker correlation equal to 0, and a single period.
  \end{minipage}
\end{figure}

\subsection{Managerial Implications}\label{section:managerial}

In this paper, I construct the Dynamic Pseudomarket Mechanism. It improves worker allocations (efficiency) without sacrificing balance. This implies that firms can implement a mechanism that performs better than Rotation, and my bounding approach predicts its effects on productivity. Taken together, my results delineate when DPM and the bound are useful in practice. 

First, the firm size should not be too small. The bound predicts the direction of DPM's effect with high accuracy at about twenty workers, and manipulation gains fall quickly as the numbers of workers and task types grow. Several hundred U.S. counties handle at least one-thousand cases per year (\qcitenp{afcars}). At fifty-five cases per caseworker, these counties could employ at least eighteen caseworkers, and those with more cases would employ far more. A crucial part of \textit{size} is also the number of task types: agencies should allow workers to report preferences over granular, rather than sparse, features of cases. Small agencies could still benefit from DPM, but they would need to be especially aware that strategic incentives could degrade the bound's predictive accuracy and that the number of features workers can report preferences over might need to be limited.

Second, the agency should expect some, albeit only short-run, imbalance. As discussed, the median caseworker would have a worst-case expected deviation of $3.7$ from $13.75$ cases per quarter. Nevertheless, DPM is balanced in the long-run. Agencies that employ Rotation are likely accustomed to some variance in workload with equal or unequal weights because cases can vary in difficulty even when the agency assigns equal weights. However, the asymptotic and incentive results only apply to equal weights. Similar to small agencies, agencies that balance by case complexity could still use DPM, but its potential effects on productivity are more unpredictable.

Third, the agency should have the resources to deploy a one-question ordinal preference survey. An agency that records outcomes under Rotation can immediately use the bound following this. The same outcome data can test the linear specification. Rotation randomly allocates tasks, and a worker's case composition varies randomly. Assumption \ref{assumption:linear} implies that productivity should not depend on composition after controlling for worker-by-task effects. One can regress Rotation worker productivity on worker-task fixed effects and composition covariates to assess this. If desired, firms can test the monotonicity assumption by surveying for cardinal preferences instead of ordinal preferences. Visual inspection of productivity's conditional expectation can verify or falsify a monotonic shape.

Practically, this could substantially outperform experimentation in cost-efficacy. Field experiments tend to have large fixed costs. In the running example, a randomized field experiment exposes vulnerable children to potential risk that further increases expected costs. In contrast, this approach could be operable without such risks and through cheap preference survey instruments. While field implementation would require more research and testing, this is a promising avenue to evaluate a new mechanism in high-stakes settings like foster care.

%% file: Paper/Conclusion.tex
\section{Conclusion}

Many firms allocate tasks subject to balance, yet balanced task allocation has received little attention from market design. Following \qcite{sonmez2023}'s minimalist approach, I take balance as the institution's constraint and design a mechanism, the Dynamic Pseudomarket Mechanism, that improves worker allocations while respecting balance. Comparatively, Rotation is inefficient with probability approaching one in the long run. Balance must be relaxed to hold in expectation and in the long run because no ex-post efficient mechanism is ex-post balanced, but DPM's short-run deviations are bounded. The analysis shows that interim efficient and interim balanced allocations always exist, that interim efficiency guarantees ex-post efficiency, and that in large random markets every such allocation gives workers almost all of their tasks from their top $K$ task types.

An important outcome for many firms is productivity. In my running example, a mechanism that makes caseworkers more satisfied but reduces service quality could harm children. I provide a bound, estimable from status quo data and a one-question preference survey, that lies between Rotation's and DPM's expected productivity. Its sign is correct, and its magnitude is conservative. In large markets, it is more conservative with many task types. In simulations, the bound tracks DPM's realized gains across environments outside the model's assumptions and predicts the direction of DPM's effect about 90\% of the time with twenty workers, even when preference reports are very noisy.

These results matter for high-stakes fields like child welfare. Better allocations could raise caseworker satisfaction and reduce turnover, but caseworkers' preferences may be only loosely related to child outcomes. The bound lets an agency assess DPM's effect on productivity before implementation without the cost and risk of a randomized field experiment.

Some caveats remain. Rotation requires no preference reports. DPM requires a reporting language and an explanation of its mechanics. DPM is not strategyproof, so workers may perceive it as unfair if they observe that strategizing improves allocations, even though allocations remain balanced. The large-market results limit this concern under weak diversity, but simulations show that strategic incentives can be large when task weights differ substantially. Therefore, introducing DPM into the field requires care.

Nevertheless, the scope for using preferences to improve task allocation seems large. This paper restricts attention to additive utility and arrival distributions known in advance. Do the results extend to preferences with complementarities, or to workers with demand constraints? How do they change if outcomes like productivity are also non-linear in the allocation? More fundamentally, what forces lead a firm to require balance, and can balance arise endogenously within profit-maximizing firms? These questions are ripe for exploration.

%% file: Paper/Examples.tex
\section{Omitted Examples}

Additional examples to illustrate intuition for the main results are here.

\subsection{Ex-Post Efficiency}\label{appendix:example-expost}

In principle, one might imagine that an interim efficient allocation $q$ can fail to be implementable as a feasible, ex-post efficient allocation. The dynamic uncertainty that arrivals cause could prevent the central planner from realizing an efficient allocation; a task $x$ allocated at time $t$ to $i$ might be preferred by a worker $j$, and a task $y$ allocated at $t+1$ to $j$ might be preferred by $i$. In this case, DPM could be interim efficient yet fail to produce ex-post efficient allocations. However, this intuition, surprisingly, is \textit{false}. $i$ and $j$ should trade all of their probability shares for the respective objects so that $i$ consumes $x$ with zero probability and $j$ consumes $y$ with zero probability. The very fact that the ex-post allocation is inefficient means that the interim allocation must also be inefficient. The example below expands on this.

\begin{example}\label{example:ex-post}
    The market $\mathcal{M}$ is:
    \[I = \{1,2\} \quad X = \{x, y\} \quad T = 2\]
    where $F(x) = 1 - \epsilon$ and $F(y) = \epsilon$ are time-independent. Preferences are:
    \begin{center}
        \begin{tikzpicture}
            \matrix[payoffmatrix] (m) {
                & $x$ & $y$ \\
                $v_1(\cdot)$ & $1$ & $2$  \\
                $v_2(\cdot)$ & $1000$ & $1$ \\
            };
        \end{tikzpicture}
    \end{center}
    The primary concern in this example is that, if $\epsilon$ is small enough (for example, $\epsilon = 1/4$), then both workers have the same ordinal preferences over objects in the \textit{interim}:
    \[(1 - \epsilon) v_1(x) > \epsilon v_1(y) \quad \text{ and } (1 - \epsilon) v_2(x) > \epsilon v_2(y)\]
    Consequently, one might imagine the following interim allocation $q$: give all of $x$ to worker $1$ and all of $y$ to worker $2$ (in both time periods). It is plausible to think that this is interim efficient because worker $1$ interim prefers one unit of $x$ over one unit of $y$, whereas worker $2$ can only be made better off if $x$ is taken from worker $1$ and given to worker $2$. In this case, the interim allocation would have an ex-post inefficient allocation in its support: whenever $x$ and $y$ arrive over the finite time horizon, $1$ gets $x$ and $2$ gets $y$. Ex-post, the two workers would trade objects.
    
    Indeed, this would be true if the objects were indivisible in the interim. However, because the probability shares are divisible, the workers can trade incrementally. To see this, note that utilities under the proposed allocation are:
    \[u_1(q_1) = 2(1 - \epsilon) \quad \text{ and } \quad u_2(q_2) = 2\epsilon\]
    If, instead, one considers an allocation $q^*$ where worker $1$ receives $\alpha$ probability shares of $x$ and $y$ with probability one, whereas worker $2$ receives $1 - \alpha$ of $x$ and zero of $y$:
    \[u_1(q_1^*) = 2\big[ (1 - \epsilon)\alpha + 2\epsilon \big] \quad \text{ and } \quad u_2(q_2^*) = 2\big[ 1000(1 - \epsilon)(1 - \alpha) \big]\]
    At $\epsilon = 1/4$ and $\alpha = 1/3$, it can be verified that $u_1(q_1^*) = u_1(q_1) = 3/2$ and $u_2(q_2^*) = 1000 > u_2(q_2) = 1/2$. Therefore, $q$ is not interim efficient. Moreover, it must be the case that worker $2$ receives $y$ with zero probability, which fully pre-empts any ex-post Pareto improving trades.

    More technically, interim efficiency is stronger than ex-post efficiency. Any interim efficient allocation is also efficient under ex-post utility. By the Birkhoff--von Neumann Theorem, one can decompose the interim efficient allocation into a randomization over ex-post allocations. If any ex-post allocation was inefficient, then one could replace it with an ex-post efficient allocation, yielding also an improvement on the interim allocation. The difficulty of the proof is to show that this replacement maintains a feasible interim allocation, and the conclusion follows as, if such a replacement existed, this would contradict interim efficiency.
\end{example}

Theorem \ref{thm:efficient} generally only holds under unconstrained demand and additive expected utility. Further complications such as demand constraints or combinatorial preferences could lead to different results. An example demand constraint is workers that can only be assigned one object. The resulting interim allocation implemented by DPM would not be equivalent to the interim allocation that DPM computes; dynamic uncertainty could then lead to discrepancies in the distributions and prevent efficiency. One instance of the second assumption might be if workers discount time at heterogeneous rates for different allocations (e.g. a worker is willing to wait for her ideal allocation but not for an allocation that is marginally different). An ex-post Pareto improving trade may not be a feasible interim Pareto improving trade if the trade involves swapping objects intertemporally. As a result, an interim efficient allocation could contain allocations that are ex-post inefficient.

\subsection{Worker Welfare and Productivity}

A worker $i$'s productivity for a task $x$ is $Y_i^x$. The productivity of an allocation for $i$ is:
\[\mathcal{P}_i(a_i) = \sum_{t = 1}^T \sum_{x \in X} Y_i^x a_{i,t}^x\]

\begin{example}\label{example:effort}
    The market $\mathcal{M}$ is:
    \[I = \{1,2\} \quad X = \{x, y\} \quad T \text{ even periods} \quad F(x) = F(y) = \frac{1}{2}\]
    Preferences are:
    \[v_i(x) = Y_i^x - 1/2\]
    where the worker incurs a cost of $1/2$ to complete a task. $w^x = w^y = 1/2$. $W_F = T/2$, and $W_F/|I| = T/4$. Balance requires $q_i^x + q_i^y = T/2$ for all $i \in I$. Workers have heterogeneous productivity: $Y_1^x = Y_2^y = 2$ and $Y_1^y = Y_2^x = 1$. With these parameters, worker utility is simplified to:
    \[v_i(x) = Y_i^x - \frac{1}{2}\]
    Under Rotation, it is immediate that workers 1 and 2 receive an equal number of each task in expectation. Then:
    \[\E[F,q_i]*{Y_i^x + Y_i^y} = \frac{1}{2} Y_i^x + \frac{1}{2} Y_i^y = \frac{3}{2}\]
    This implies that each worker's average productivity per task is $3/2$.

    Under DPM, every allocation is interim efficient and interim balanced. It is immediate that for any allocation such that $q^{*,x}_1 + q^{*,y}_1= T/2$, $q^{*,x}_1 = T/2$ and $q^{*,y}_1 = 0$ maximizes worker $1$'s utility, and vice versa for worker $2$. This can be the only interim efficient and interim balanced allocation. So:
    \[\E[F,q_1^*]*{Y_1^x + Y_1^y} = \E[F,q_2^*]*{Y_2^x + Y_2^y} = 2\]
    Therefore, average productivity under DPM is $2$ for both workers, an increase of one-third over Rotation's $3/2$.
\end{example}

Here, DPM necessarily increases worker welfare. When worker preferences correlate with productivity, DPM also increases worker welfare and productivity. On the other hand, it is immediate that if $v_i(x) = (3 - Y_i^x) - 1/2$, then the opposite conclusion holds: DPM still necessarily increases worker welfare, but preferences are negatively correlated with productivity, so it decreases productivity.

%% file: Paper/Relationship.tex
\section{Relationship to Existing Mechanisms}\label{appendix:relationship}

DPM shares most with the classic Hylland-Zeckhauser pseudomarket (\qcitenp{hylland1979}). The main difference is that DPM is balanced. The HZ pseudomarket computes competitive equilibria from equal budgets rather than allowing budgets to vary; consequently, it is envy-free rather than balanced. Balance is a key constraint in the applications that I describe, and, rather than being a version of envy-freeness, it is a fundamentally separate normative concept. Moreover, Pareto efficiency, envy-freeness, and balance are at odds. 

Consider this example: there are two agents, Alice and Bob, and two tasks, $x$ and $y$. Each task has one copy. Task $x$'s weight is three-times that of $y$. The allocation must then satisfy $3q_a^x + q_a^y = 3q_b^x + q_b^y = 2$, $q_a^x + q_b^x = 1$, and $q_a^y + q_b^y = 1$. Alice prefers $x$ over $y$, while Bob prefers $y$ over $x$. A first pass at an efficient and balanced allocation might be assigning $2/3$ of $x$ to Alice and $1/3$ to Bob, while Bob receives $1$ of $y$. Conjecture the utilities $v_a(x) = 1$ and $v_a(y) = \epsilon$ for Alice and $v_b(x) = \epsilon$ and $v_b(y) = 1$ for Bob for an arbitrary $\epsilon < 1$. The constraints imply that:
\[3q_a^x + q_a^y = 2 \implies q_a^x \leq \frac{2}{3}\]
and
\[q_a^x + q_b^x = 1 \implies q_b^x \geq \frac{1}{3}\]
However, if $q_b^x > 1/3$, then:
\[3q_b^x + q_b^y = 2 \implies q_b^y < 1\] 
Furthermore:
\[q_a^y + q_b^y = 1 \implies q_a^y > 0\]
so that the first constraint implies that $q_a^x < 2/3$. Yet, this implies that some of $x$ could be reallocated from Bob to Alice and some of $y$ from Alice to Bob to make both strictly better off, contradicting efficiency. Therefore, it must be that $q_b^x = 1/3$. The previous constraints then imply that the hypothesized allocation is the unique efficient and balanced allocation. However, notice that $u_a(q_b) = 1/3 + \epsilon > 2/3 = u_a(q_a)$ for any $\epsilon \in (1/3, 1)$, implying that this allocation cannot be envy-free.

The key in this counterexample is that---if weights represent task costs---then $x$ is too costly to fully allocate to Alice, who might prefer it because she is relatively more productive at it than $y$. Bob must help Alice with $x$ even though it would be efficient and envy-free to allocate all of $x$ to Alice and all of $y$ to Bob\footnote{In the above example, there is a quirk in the interpretation because I normalize utilities to be positive. It should seem that the efficient and envy-free allocation is \textit{less} preferable to Alice than the unique efficient and balanced allocation. Re-normalizing utilities to be negative implies that Bob envies Alice, and then the efficient and envy-free allocation is less preferable for Alice.}. Alice's envy is a reflection of the normative judgment of costs that the firm makes, implying that, in the workplace context, such envy is unlikely to be viewed as justified. It follows that an envy-free allocation will not necessarily be \textit{more fair} than a balanced allocation.

%% file: Paper/WeightedRotation.tex
\section{Weighted Rotation}\label{appendix:weighted-rotation}

Weighted Rotation allocates each arriving task to the worker with the least weighted allocation. \\~\

\textit{Algorithm: Weighted Rotation} (Input: $\mathcal{M}$. Initialize $a_{i,t}^x = 0$ for all $i,t,x$.)
\begin{enumerate}
    \item Allocate an item $x$ arriving at time $t$ to $i^*$ satisfying $i^* \in \argmin_{i \in I} W(a_i)$. \\~\
\end{enumerate}

Ties are broken by an arbitrary fixed rule that does not depend on arrivals nor preferences.

\begin{remark}\label{remark:weighted-rotation}
    Weighted Rotation is ex-post balanced.
\end{remark}

\begin{center}
    \textbf{Proof of Remark \ref{remark:weighted-rotation}}
\end{center}
I will write:
\[W_t(a_i) = \sum_{k = 1}^t \sum_{x \in X} w^x a_{i,k}^x \quad \text{ and } \quad W_t(a) = \sum_{i \in I} \sum_{k = 1}^t \sum_{x \in X} w^x a_{i,k}^x \]
and $k^* := \max_{x \in X} w^x$. I proceed with induction. For any $t \in [T]$:
\[ 0 \le \max_{i \in I} W_t(a_i) - \min_{j \in I} W_t(a_j) \le k^* \quad \text{(Inductive Statement)}\]
\textit{Base Case:} $t = 1$.

The quantity is vacuously weakly positive. Necessarily, $a_{i,t}^x = 1$ for exactly one $i,t,x$. Therefore, $W_t(a_j) \in [0, k^*]$ for all $j \in I$, and the statement trivially holds.

\textit{Inductive Step:} I assume the statement holds for all $k \le t$ and prove it holds for $t+1$. By definition of Weighted Rotation, the $i$ allocated a task at $t + 1$ satisfies:
\[i \in \argmin_{j \in I} W_t(a_j)\]
so that $W_t(a_i) \le W_t(a_j)$ for all $j \ne i$. Thus:
\begin{align*}
    W_{t+1}(a_i) &= W_t(a_i) + \sum_{x \in X} w^x a_{i,t+1}^x \\
    &\leq W_t(a_j) + \sum_{x \in X} w^x a_{i,t+1}^x \\
    &\leq W_t(a_j) + k^* \\
    &= W_{t+1}(a_j) + k^*
\end{align*}

There are two cases. If $i = \argmax_{j \in I} W_{t+1}(a_j)$, then the above directly implies:
\[0 \le \max_{i \in I} W_{t+1}(a_i) - \min_{j \in I} W_{t+1}(a_j) \le k^*\]
If $i \neq \argmax_{j \in I} W_{t+1}(a_j)$, then $\max_{j \in I} W_{t+1}(a_j) = \max_{j \in I} W_t(a_j)$ and $\min_{j' \in I} W_{t+1}(a_{j'}) \ge \min_{j' \in I} W_t(a_{j'})$ necessarily, so:
\[0 \le \max_{j \in I} W_{t+1}(a_j) - \min_{j' \in I} W_{t+1}(a_{j'}) \le \max_{j \in I} W_{t}(a_j) - \min_{j' \in I} W_{t}(a_{j'}) \le k^*\]
proving the inductive statement. It follows that for any $i,j,t$:
\[\left| W_t(a_i) - W_t(a_j) \right| \le k^*\]
Toward a contradiction, suppose that Weighted Rotation is not ex-post balanced. Let $t^*$ be the first $t$ such that for some $i \in I$:
\[\left| W_{t^*}(a_i) - \frac{W_{t^*}(a)}{|I|} \right| \geq k^* = \max_{x \in X} w^x\]
I claim there must exist some $j \in I$ such that $|W_{t^*}(a_i) - W_{t^*}(a_j)| > k^*$. If this were not true, then, for all $j \in I$, $|W_{t^*}(a_i) - W_{t^*}(a_j)| \leq k^*$. Hence:
\[\frac{W_{t^*}(a)}{|I|} \in \left[ W_{t^*}(a_i) - \frac{(|I| - 1)}{|I|} \cdot k^*, \quad W_{t^*}(a_i) + \frac{(|I| - 1)}{|I|}\cdot k^* \right]\]
implying:
\[\left| W_{t^*}(a_i) - \frac{W_{t^*}(a)}{|I|} \right| \in \left[ 0, \frac{(|I| - 1)}{|I|} \cdot k^* \right]\]
which is strictly less than $k^*$. Therefore, the claim must hold if Weighted Rotation is not ex-post balanced, contradicting the inductive argument. $\blacksquare$

%% file: Paper/Appendix.tex
\section{Theory Appendix}\label{appendix:theory}

\begin{center}
    \textbf{Proof of Remark \ref{remark:rotation}}
\end{center}

Let $a$ be a realization of $\pi_R(v)$. I write $\tilde W(a_i) := \sum_{t = 1}^T \sum_{x \in X} a_{i,t}^x$ for the number of tasks that $i$ receives under $a$.

\textbf{Step 1.} $|\tilde W(a_i) - T/|I|| < 1$ for all $i \in I$.

Each worker receives exactly one task every $|I|$ periods. This implies that $\tilde W(a_i) \in \{\lfloor T/|I| \rfloor, \lceil T/|I| \rceil\}$. Hence, $|\tilde W(a_i) - \tilde W(a_j)| \le 1$ for all $i, j \in I$. Then:
\begin{align*}
    \left| \tilde W(a_i) - \frac{T}{|I|} \right| &= \left| \frac{1}{|I|} \sum_{j \in I} \left( \tilde W(a_i) - \tilde W(a_j) \right) \right| \\
    &\le \frac{1}{|I|} \sum_{j \ne i} \left| \tilde W(a_i) - \tilde W(a_j) \right| \\
    &\le \frac{|I| - 1}{|I|} < 1
\end{align*}
The equality holds because $\sum_{j \in I} \tilde W(a_j) = T$ since Rotation allocates all tasks. The first inequality holds by the triangle inequality.

\textbf{Step 2.} If $w^x = 1$ for all $x \in X$, then Rotation is ex-post balanced.

With equal weights, $W(a_i) = \tilde W(a_i)$ and $W(a) = T$. Then:
\[\left| W(a_i) - \frac{W(a)}{|I|} \right| = \left| \tilde W(a_i) - \frac{T}{|I|} \right| < 1 = \max_{x \in X} w^x\]
where the inequality holds by Step 1.

\textbf{Step 3.} If $F_t(x) = F_{t'}(x)$ for all $x \in X$ and $t \ne t'$, then Rotation is asymptotically ex-post balanced.

Let $a^k$ be a realization of $\pi_R^k(v)$, and fix the worker order. One has:
\[\E*{W(a_i^k)} = \sum_{t = 1}^{kT} \sum_{x \in X} w^x \E*{a_{i,t}^{k,x}} = \tilde W(a_i^k) \sum_{x \in X} w^x F_1(x) = \tilde W(a_i^k) \frac{W_F}{T}\]
where the second equality follows conditional on a fixed order because $i$ receives every task arriving in $\tilde W(a_i^k)$ fixed periods. Moreover:
\[\left| \E*{W(a_i^k)} - \frac{kW_F}{|I|} \right| = \frac{W_F}{T} \left| \tilde W(a_i^k) - \frac{kT}{|I|} \right| < 1\]
where the inequality follows from Step 1 and $W_F/T = \sum_{x \in X} w^x F_1(x) \le 1$. Additionally:
\[Var(W(a_i^k)) = \sum_{t = 1}^{kT} Var\left( \sum_{x \in X} w^x a_{i,t}^{k,x} \right) \le \frac{\tilde W(a_i^k)}{4}\]
where the equality follows because arrivals are independent across $t$, and the inequality follows because $w^x \in (0,1]$ for all $x \in X$ and $a_{i,t}^{k,x} \in \{0,1\}$. Then, for any $\epsilon > 0$ and $k\epsilon > 1$:
\begin{align*}
    \Pr\Bigg\{ \bigg| \frac{1}{k} W(a_i^k) - \frac{W_F}{|I|} \bigg| \ge \epsilon \Bigg\} &\le \Pr\Bigg\{ \bigg| W(a_i^k) - \E*{W(a_i^k)} \bigg| + \bigg| \E*{W(a_i^k)} - \frac{kW_F}{|I|} \bigg| \ge k\epsilon \Bigg\} \\
    &\le \Pr\Bigg\{ \bigg| W(a_i^k) - \E*{W(a_i^k)} \bigg| \ge k\epsilon - 1 \Bigg\} \\
    &\le \frac{Var(W(a_i^k))}{(k\epsilon - 1)^2} \\
    &\le \frac{kT/|I| + 1}{4(k\epsilon - 1)^2} \to 0
\end{align*}
The first inequality holds by the triangle inequality. The third holds by Chebyshev's inequality, and the fourth holds by the variance bound above and Step 1. The bound does not depend on the worker order. It holds unconditionally. $\blacksquare$

\begin{center}
    \textbf{Proof of Proposition \ref{proposition:inefficient}}
\end{center}

Fix $i \ne j$ in $I$ and $x \ne y$ in $X$, and let $V_i := v_i(x)/v_i(y)$ and $V_j := v_j(x)/v_j(y)$. Fix $k$, let $\tilde a^k$ be the realized arrivals in $\mathcal{M}^k$, and let $a^k$ be a realization of $\pi_R^k(v)$. I write $\bar a_i^{k,z} := \sum_{t=1}^{kT} a_{i,t}^{k,z}$ and $\underline f := \min_{z,t} F_t(z)$.

\textbf{Step 1.} If $V_j < c_1/c_0 < V_i$ for positive integers $c_0, c_1$, $\bar a_j^{k,x} \ge c_0$, and $\bar a_i^{k,y} \ge c_1$, then $a^k$ is ex-post inefficient.

Consider an alternative allocation $a^*$ where $\bar a^{*,x}_i = \bar a_i^{k,x} + c_0$, $\bar a^{*,x}_j = \bar a_j^{k,x} - c_0$, $\bar a^{*,y}_i = \bar a_i^{k,y} - c_1$, and $\bar a^{*,y}_j = \bar a_j^{k,y} + c_1$. This is feasible for $\tilde a^k$. Moreover:
\[U_i(a_i^*) - U_i(a_i^k) = c_0 v_i(x) - c_1 v_i(y) > 0 \quad \text{ and } \quad U_j(a_j^*) - U_j(a_j^k) = c_1 v_j(y) - c_0 v_j(x) > 0\]
since $c_1/c_0 < V_i$ and $c_1/c_0 > V_j$. The same holds with the roles of $i$ and $j$ exchanged.

For the next step, I define $C_k := \lfloor kT/|I| \rfloor$.

\textbf{Step 2.} For each $z \in X$ and positive integer $N < C_k \underline f$:
\[\Pr\big(\bar a_i^{k,z} < N\big) \le \frac{C_k}{4(C_k \underline f - N)^2}\]
and the same holds for $j$.

$i$ receives the tasks arriving in at least $C_k$ fixed periods. I write $a_{i,(s)}^{k,z}$ for whether or not $i$ receives $z$ in the $s$-th of these periods. Since the order is drawn independently of arrivals and arrivals are independent across $t$, the $a_{i,(s)}^{k,z}$ are independent Bernoulli variables with $\Pr(a_{i,(s)}^{k,z} = 1) \ge \underline f$. Then:
\begin{align*}
    \Pr\left(\bar a_i^{k,z} < N\right) &\le \Pr\left( \sum_{s = 1}^{C_k} a_{i,(s)}^{k,z} < N \right) \\
    &\le \Pr\left( \sum_{s = 1}^{C_k} a_{i,(s)}^{k,z} - \E*{\sum_{s = 1}^{C_k} a_{i,(s)}^{k,z}} < N - C_k \underline f \right) \\
    &\le \frac{C_k}{4(C_k \underline f - N)^2}
\end{align*}
The first inequality holds because $\bar a_i^{k,z} \ge \sum_{s=1}^{C_k} a_{i,(s)}^{k,z}$. The second holds because:
\[\E*{\sum_{s=1}^{C_k} a_{i,(s)}^{k,z}} \ge C_k \underline f\]
The third holds by Chebyshev's inequality.

\textbf{Step 3.} The Proposition holds.

Since $v_i$ is continuously distributed and valuations are i.i.d.\ across workers:
\[\Pr(V_i = V_j) = 0\]
Fix $\epsilon > 0$. Positive integers $c_0, c_1$ with $c_1/c_0$ strictly between $V_i$ and $V_j$ exist whenever $V_i \ne V_j$. One can choose $N$ such that, with probability at least $1 - \epsilon/2$, such $c_0, c_1 \le N$ exist. Choose $K$ such that, for all $k \ge K$, $C_k \underline f > N$ and:
\[\frac{C_k}{(C_k \underline f - N)^2} < \frac{\epsilon}{2}\]
By Step 1, $a^k$ is ex-post inefficient if such $c_0, c_1 \le N$ exist and $\bar a_i^{k,x}, \bar a_i^{k,y}, \bar a_j^{k,x}, \bar a_j^{k,y} \ge N$. By the union bound and Step 2, for all $k \ge K$:
\[\Pr\big(a^k \text{ is ex-post efficient}\big) \le \frac{4 C_k}{4(C_k \underline f - N)^2} + \frac{\epsilon}{2} < \frac{\epsilon}{2} + \frac{\epsilon}{2} = \epsilon\]
$\blacksquare$

\begin{center}
    \textbf{Proof of Proposition \ref{proposition:expost}}
\end{center}

The market $\mathcal{M}$ is:
\[I = \{1,2\} \quad X = \{x, y, z\} \quad T = 2\]
The arrival distribution is $F_t(x) = F_t(y) = F_t(z) = 1/3$ for $t \in \{1,2\}$. The weights are: $w^x = 1/3$, $w^y = 1/3$, and $w^z = 1/3$. Preferences are:
\begin{center}
    \begin{tikzpicture}
        \matrix[payoffmatrix] (m) {
            & $x$ & $y$ & $z$ \\
            $v_1(\cdot)$ & $10$ & $1$ & $5$ \\
            $v_2(\cdot)$ & $1$ & $10$ & $5$ \\
        };
    \end{tikzpicture}
\end{center}
Trivially, the preference distribution can be constructed to contain the valuation in this proof. Let $a := \pi(\mathcal{M})$ for some arbitrary random mechanism $\pi$ that satisfies ex-post efficiency and ex-post balance. I will write $a_i = (x,o)$ if $a_{i,1}^x = 1$ and $a_{i,2}^{x'} = 0$ for all $x' \in X$.

Any mechanism necessarily cannot distinguish between realizations of arrivals at $t = 1$, and any ex-post efficient mechanism cannot discard an arrival. Hence, $\pi_1$ cannot condition on the realized arrival at $t = 2$ and must allocate the arrival at $t = 1$ at that time. Suppose that $a_{1,1}^z = 1$ and that the realization is $(z,x)$. Ex-post balance implies that $a_1 = (z, o)$ and $a_2 = (o,x)$. Otherwise, if $a_1 = (z, x)$, then $W(a) = 2/3$, and $|W(a_2) - W(a)/2| = 1/3$ which is not strictly less than $\max_{x \in X} w^x = 1/3$. However, $a$ is ex-post inefficient because agents would prefer to swap allocations. 

Instead, suppose that $a_{2,1}^z = 1$ and that the realization is $(z,y)$. Symmetrically, ex-post balance gives $a_1 = (o,y)$ and $a_2 = (z,o)$ which is also ex-post inefficient. Therefore, if the first arrival is $z$, then there is always some realization such that $\pi$ is ex-post inefficient if it is ex-post balanced, a contradiction. $\blacksquare$

\begin{center}
    \textbf{Proof of Theorem \ref{thm:existence}}
\end{center} 

For prices $p \in \mathbb{R}_+^{|X| \times T}$ and a ratio $r_i \in \mathbb{R}_+$, $i$'s budget correspondence is
\[B_i(p,r_i) = \{ q_i \in \mathcal{Q} : r_i p \cdot q_i \leq 1 \}\]
I will truncate the budget correspondence to a compact set. Denote the truncated consumption set as $\tilde{\mathcal{Q}} = \{ q_i \in \mathcal{Q} : q_{i,t}^x \leq K \hspace{5pt} \forall x \in X, t \in [T] \}$ for some $K > 1$, then the truncated budget set is
\[\tilde{B}_i(p, r_i) = B_i(p, r_i) \cap \tilde{\mathcal{Q}}\]
$i$'s demand correspondence is:
\[d_i(p, r_i) = \argmax_{q_i \in \tilde{B}_i(p, r_i)} u_i(q_i)\]

\textbf{Step 1.} $d_i(p, r_i)$ is upper hemicontinuous.

\begin{itemize}
    \item $\tilde{B}_i(p, r_i)$ is compact. This follows from the fact that it is bounded and closed.
    \item $\tilde{B}_i(p, r_i)$ is non-empty. The zero consumption vector is always in it.
    \item $d_i(p, r_i)$ is convex-valued. This follows from the fact that expected utility $u_i$ is continuous in $q_i$, and the set of maximizers for a linear function over a convex domain is convex.
    \item $d_i(p, r_i)$ is non-empty and compact by the fact that the set of maximizers of a continuous function on a compact domain is non-empty and compact.
\end{itemize}

To apply Berge's maximum theorem, it remains to show that the truncated budget set is upper and lower hemicontinuous. I begin with the former.

\begin{itemize}
    \item $\tilde{B}_i(p, r_i)$ is upper hemicontinuous.
\end{itemize}

Because $\tilde{B}_i(p, r_i)$ is compact-valued, it is upper hemicontinuous if it has a closed graph. Consider a sequence $(p^n, r_i^n, q_i^n) \rightarrow (p, r_i, q_i)$ such that $q^n_i \in \tilde{B}_i(p^n, r_i^n)$. This implies that $r_i^n p^n \cdot q_i^n \leq 1$. By limit properties, $r_i^n p_i^n \cdot q_i^n = \sum_{t = 1}^T \sum_{x \in X} r_i^n p_t^{n,x} q_{i,t}^{n,x} \rightarrow \sum_{t = 1}^T \sum_{x \in X} r_i p_t^x q_{i,t}^x = r_i p \cdot q_i$. Moreover:
\[r_i^n p_i^n \cdot q_i^n \leq 1 \hspace{5pt} \forall n\]
implies that $r_i p \cdot q_i \leq 1$, i.e., $q_i \in B_i(p, r_i)$. Clearly, $q_i^n \in \tilde{B}_i(p^n, r_i^n) \implies q^n_i \in \tilde{\mathcal{Q}}$. If $q_i \notin \tilde{\mathcal{Q}}$, then I can take a sufficiently large $n$ such that $q_i^n \notin \tilde{\mathcal{Q}}$, a contradiction. Therefore, $q_i \in \tilde{\mathcal{Q}}$ so that $q_i \in \tilde{B}_i(p, r_i)$.

\begin{itemize}
    \item $\tilde{B}_i(p, r_i)$ is lower hemicontinuous.
\end{itemize}

$\tilde{B}_i(p, r_i)$ is lower hemicontinuous at $(p, r_i)$ if, for every open set $V$ such that $V \cap \tilde{B}_i(p, r_i) \neq \emptyset$, there exists an open ball $U$ around $(p, r_i)$ such that $\tilde{B}_i(p', r'_i) \cap V \neq \emptyset$ for all $(p', r'_i) \in U$.

Since wealth is strictly positive, I can take some $q_i \in V \cap \tilde{B}_i(p, r_i)$ satisfying $r_i p \cdot q_i < 1$. Let
\[c = 1 - r_i p \cdot q_i \]
$c$ is strictly positive. Consider some stacked vector $(p', r'_i) \in U$ where $U = \{ (p', r'_i) : ||(p', r'_i) - (p, r_i)||_\infty < \epsilon \}$ for some $\epsilon > 0$. Denote $M \equiv ||(p', r'_i) - (p, r_i)||_\infty$ as the largest magnitude component of the vector difference. I have that:
\begin{align*}
    r'_i p' \cdot q_i = \sum_{t = 1}^T \sum_{x \in X} r'_i p'^{,x}_t q_{i,t}^x &\leq \sum_{t = 1}^T \sum_{x \in X} (r_i + M)(p^x_t + M) q_{i,t}^x \\
    &\leq \sum_{t = 1}^T \sum_{x \in X} r_i p^x_t q_{i,t}^x + r_i M q_{i,t}^x + p_t^x M q_{i,t}^x + M^2 q_{i,t}^x \\
    &< \sum_{t = 1}^T \sum_{x \in X} r_i p^x_t q_{i,t}^x + r_i M q_{i,t}^x + p_t^x M q_{i,t}^x + M q_{i,t}^x \\
    &< r_i p \cdot q_i + \epsilon D
\end{align*}
where $\sum_{t = 1}^T \sum_{x \in X} (r_i + p_t^x + 1) q_{i,t}^x \leq (r_i + 1)KT|X| + 1/r_i \equiv D$ follows by the budget constraint and algebra. The first inequality above follows by definition of the $\infty$-norm since $r'_i - r_i \leq M$ and $p'^{,x}_t - p^x_t \leq M$. The third inequality holds for any $\epsilon < 1$ since this will imply that $M^2 \leq M < \epsilon$. Then, I can choose $\epsilon = c/D$, and the above requirement is trivially satisfied since $r_i \geq 0$, $c \leq 1$, $K > 1$, $T \geq 1$, and $|X| \geq 1$. Then, I have:
\[r'_i p'\cdot q_i < 1\]
Therefore, $q_i \in B_i(p', r'_i)$. Of course, $q_i \in \tilde{\mathcal{Q}}$ by assumption that $q_i \in \tilde{B}_i(p, r_i)$ so $q_i \in \tilde{B}_i(p', r_i')$. This implies $\tilde{B}_i(p', r'_i) \cap V \neq \emptyset$ for all $(p', r'_i) \in U$.

\textbf{Step 2.} The excess demand correspondence has a fixed point. 

Let $\bar{p} = 2|I|$, $\underline{r} = 1$, and $\bar{r} := 2|I|T\bar V/(\underline{f} \underline{V})$, where $\bar V := \sup_{i,x} \Supp{v_i(x)}$, $\underline{V} := \inf_{i,x} \Supp{v_i(x)}$, and $\underline{f} := \min_{x,t} F_t(x)$. Market excess demand for $(x,t) \in X \times [T]$ given an allocation $q$ is
\[Z^x_t(q) = \sum_{i \in I} q_{i,t}^x - 1\]
Excess demand for $i$ is
\[Z_i(q_i) = E_{F,w}[q_i] - \frac{W_F}{|I|}\]
I can then define the price- and ratio-adjustment maps:
\[A^x_t(p, q) = \min\Bigg\{ \max\{ p_t^x + Z_t^x(q), 0 \}, \bar{p} \Bigg\} \quad \text{and} \quad A_i(r_i, q_i) = \min\Bigg\{ \max\{ r_i + Z_i(q_i), \underline{r} \}, \bar{r} \Bigg\}\]
and the stacked map:
\[\Phi(p,r,q) = \Big(\{A^x_t(p, q)\}_{(x,t) \in X \times [T]}, \{A_i(r_i, q_i)\}_{i \in I}, \{d_i(p, r_i)\}_{i \in I}\Big)\]
$\Phi$ maps $D := [0,\bar p]^{|X| \times T} \times [\underline r, \bar r]^{|I|} \times \tilde{\mathcal{Q}}^{|I|}$ into itself, and $D$ is non-empty, convex, and compact. The adjustment maps are single-valued and continuous because they are compositions of linear maps with $\min$ and $\max$. By Step 1, each $d_i$ is upper hemicontinuous with non-empty, convex, and compact values. Hence, $\Phi$ is upper hemicontinuous with non-empty, convex, and compact values. By Kakutani's fixed-point theorem, there is a fixed point $(p^*, r^*, q^*) \in \Phi(p^*, r^*, q^*)$ such that $q_i^* \in d_i(p^*, r_i^*)$ for all $i \in I$, and
\[p^{*,x}_t = \min\Bigg\{\max\{ p_t^{*,x} + Z_t^x(q^*), 0\}, \bar{p} \Bigg\} \quad \text{and} \quad r_i^* = \min\Bigg\{\max\{ r_i^* + Z_i(q_i^*), \underline{r} \}, \bar{r} \Bigg\}\]

\textbf{Step 3.} The fixed point $(p^*, r^*, q^*)$ forms a balanced equilibrium.

\textit{Claim 1:} $p^{*,x}_t > 0$ for all $(x,t) \in X \times [T]$. Toward a contradiction, suppose that $p^{*,x}_t = 0$ for some $(x,t)$. Since demand is strictly increasing in all $(x,t)$ for all $i$, $q^{*,x}_{i,t} = K > 1$ which implies that $A_t^x(p^*, q^*) > p^{*,x}_t = 0$, contradicting that this is a fixed point. 

\textit{Claim 2:} $p^{*,x}_t < \bar{p}$ for all $(x,t) \in X \times [T]$. Suppose not, then $p^{*,x}_t = \bar{p}$ for some $(x,t)$. By the budget constraint:
\[q_{i,t}^{*,x} \leq \frac{1}{r_i \bar{p}} \implies \sum_{i \in I} q_{i,t}^{*,x} \leq \sum_{i \in I} \frac{1}{r_i \bar{p}} \leq \frac{|I|}{\bar{p}} < 1\]
where the second inequality on the right-hand side follows because $r_i \geq 1$ for all $i$. This implies that $A_t^x(p^*, q^*) < p^{*,x}_t = \bar{p}$, contradicting that this is a fixed point.

\textit{Claim 3:} $\sum_{i \in I} q_{i,t}^{*,x} = 1$ for all $(x,t) \in X \times [T]$. By the above two claims, I have:
\[p_t^{*,x} = p_t^{*,x} + Z_t^x(q^*) \implies Z_t^x(q^*) = 0 \implies \sum_{i \in I} q_{i,t}^{*,x} = 1\]
Claim 3 implies
\[\sum_{i \in I} E_{F,w}[q_i^*] = \sum_{i \in I} \sum_{t = 1}^T \sum_{x \in X} F_t(x) w^x q_{i,t}^{*,x} = \sum_{t = 1}^T \sum_{x \in X} F_t(x) w^x = W_F\]
so that there is at least one $i$ such that $E_{F,w}[q_i^*] \geq W_F/|I|$.

\textit{Claim 4:} $Z_i(q_i^*) \leq 0$ for all $i \in I$. Suppose not, then $Z_i(q_i^*) > 0$ for some $i$. By assumption that $r^*$ is a fixed point, I must have that $r_i^* = \bar{r}$. Moreover, $Z_i(q_i^*) > 0$ implies that $E_{F,w}[q_i^*] > W_F/|I|$. This also implies that there exists some $j$ such that $E_{F,w}[q_j^*] < W_F/|I|$. Otherwise:
\[\sum_{k \in I} E_{F,w}[q_k^*] > W_F\]
contradicting Claim 3. In addition, by assumption that $r$ is a fixed point, I must have that $r_j^* = 1$. Let $(y,n) = \argmax_{(x,t) \in X \times [T]} q_{i,t}^{*,x}$. By the budget constraint, I have:
\[p^* \cdot q_i^* \leq \frac{1}{\bar{r}} \implies 0 < p^{*,y}_n  \leq \frac{1}{\bar{r} q_{i,n}^{*,y}} \leq \frac{|I|}{\bar{r}}\]
where the last inequality on the right-hand side follows because $E_{F,w}[q_i^*] > W_F/|I|$ implies that $q_{i,n}^{*,y} > 1/|I|$, otherwise:
\[E_{F,w}[q_i^*] \leq \sum_{t = 1}^T \sum_{x \in X} \frac{F_t(x) w^x}{|I|} = \frac{W_F}{|I|}\]
a contradiction. By Claim 3, $q_{j,t}^{*,x} \le 1 < K$ for all $(x,t) \in X \times [T]$, so $j$ must spend her entire budget, $p^* \cdot q_j^* = 1/r_j^* = 1$. In addition, $j$ consumes only goods that maximize $F_t(x)v_j(x)/p_t^{*,x}$, since otherwise she could reallocate expenditure and strictly increase utility. Denote this maximum by $\lambda_j$. Then:
\begin{align*}
    u_j(q_j^*) &= \sum_{t=1}^T \sum_{x \in X} F_t(x) v_j(x) q_{j,t}^{*,x} \\
    &= \sum_{t=1}^T \sum_{x : q_{j,t}^{*,x} > 0} \frac{F_t(x) v_j(x)}{p_t^{*,x}} p_t^{*,x} q_{j,t}^{*,x} = \lambda_j\, p^* \cdot q_j^* = \lambda_j \\
    &\ge \frac{F_n(y)v_j(y)}{p_n^{*,y}} \ge \frac{\underline{f} \underline{V} \bar r}{|I|} \\
    &= 2T\bar V
\end{align*}
However, $q_{j,t}^{*,x} \le 1$ gives $u_j(q_j^*) \le \sum_{t=1}^T\sum_{x \in X} F_t(x)\bar V = T\bar V$, a contradiction. Therefore, $Z_i(q_i^*) \leq 0$ for all $i \in I$.

\textit{Claim 5:} $E_{F,w}[q_i^*] = W_F/|I|$ for all $i \in I$. By Claim 4, I have that $E_{F,w}[q_i^*] \leq W_F/|I|$ for all $i \in I$. By Claim 3, I have that $\sum_{i \in I} E_{F,w}[q_i^*] = W_F$ which is only possible if $E_{F,w}[q_i^*] = W_F/|I|$ for all $i \in I$.

\textit{Claim 6:} If $u_i(q_i) > u_i(q_i^*)$, then $r_i^* p^* \cdot q_i > 1$ for each $q_i \in \mathcal{Q}$. If $q_i \in \tilde{\mathcal{Q}}$, this follows immediately. Suppose that $q_i \notin \tilde{\mathcal{Q}}$ and $r_i^* p^* \cdot q_i \leq 1$. There exists some $\alpha \in (0,1)$ such that $q'_i = \alpha q_i + (1 - \alpha) q_i^* \in \tilde{\mathcal{Q}}$ because $q_{i,t}^{*,x} \leq 1$ for all $(x,t) \in X \times [T]$. Clearly $q'_i \in B_i(p^*, r_i^*)$, and by linearity of expected utility, $u_i(q'_i) > u_i(q_i^*)$, contradicting that $q_i^* \in d_i(p^*, r_i^*)$. Therefore, $r_i^* p^* \cdot q_i > 1$.

By Claims 1-6, I have that $(p^*, r^*, q^*)$ is a balanced equilibrium satisfying market clearing. $\blacksquare$

\begin{center}
    \textbf{Proof of Theorem \ref{thm:first-welfare}}
\end{center} 

Toward a contradiction, suppose that $(p^*, r^*, q^*)$ is a balanced equilibrium and that there is some feasible allocation $q$ such that, for all $i \in I$, $u_i(q_i) \geq u_i(q_i^*)$ and $u_j(q_j) > u_j(q_j^*)$ for some $j \in I$. 

First, note that for any $q'_i \in \mathcal{Q}$, if $u_i(q'_i) \geq u_i(q_i^*)$, then $p^* \cdot q'_i \geq p^* \cdot q_i^*$. Otherwise, I can find a $\alpha \in (0,1)$ such that $p^* \cdot (q'_i + \alpha q'_i) \leq 1/r_i$, and, trivially, $u_i(q'_i + \alpha q'_i) > u_i(q_i^*)$, contradicting that $(p^*, r^*, q^*)$ is an equilibrium.

Second, by definition of a balanced equilibrium, if $u_i(q_i) > u_i(q_i^*)$, then $p^* \cdot q_i > 1/r_i^* \geq p^* \cdot q_i^*$. Therefore, summing over all agents:
\[\sum_{i \in I} p^* \cdot q_i > \sum_{i \in I} p^* \cdot q_i^*\]
so that there must exist some $(x,t) \in X \times [T]$ such that:
\[\sum_{i \in I} q_{i,t}^{x} > \sum_{i \in I} q_{i,t}^{*,x} = 1\]
where the equality holds by market clearing, contradicting the assumption that $q$ is feasible. Therefore, no such allocation $q$ can exist, and the allocation $q^*$ is Pareto efficient. $\blacksquare$

\begin{center}
    \textbf{Proof of Theorem \ref{thm:efficient}}
\end{center} 

For the first part of the Theorem (steps 1 and 2), I generalize the proof in \qcite{benade2024fair} to accommodate non-identical distributions over time. Along the way, I exploit my notation to simplify the proof at some steps.

Denote $(p^*, r^*, q^*)$ as a market clearing and balanced equilibrium for $\mathcal{M}$.

\textbf{Step 1:} DPM is interim efficient under ex-post utility. 

For an interim allocation $q_i$, write $U_i(q_i) := \sum_{t=1}^T \sum_{x \in X} q_{i,t}^x v_i(x)$. I claim that there does not exist any feasible $q \in \mathcal{Q}$ such that $U_i(q_i) \geq U_i(q_i^*)$ for all $i \in I$ and $U_j(q_j) > U_j(q_j^*)$ for some $j \in I$. 

Toward a contradiction, suppose that such a $q$ exists. Let:
\[\delta_{i,t}^x = \frac{q_{i,t}^x - q_{i,t}^{*,x}}{F_t(x)}\]
This implies:
\[u_i(q_i^* + c\delta_i) = u_i(q_i^*) + c(U_i(q_i) - U_i(q_i^*)) \geq u_i(q_i^*)\]
and strictly for some $j$. I show that there exists some allocation $q'_i := q_i^* + c\delta_i$ that is feasible
for some $c > 0$. I conjecture $c < \min_{x,t} F_t(x)$. By assumption of feasibility, $\sum_{i \in I} q_{i,t}^x \leq 1$. Since an interim efficient allocation assigns all objects, $\sum_{i \in I} q_{i,t}^{*,x} = 1$ for all $(x,t) \in X \times [T]$. Then:
\[\sum_{i \in I} q'^{,x}_{i,t} = \sum_{i \in I} q_{i,t}^{*,x} + \frac{c}{F_t(x)} \sum_{i \in I} \left( q_{i,t}^x - q_{i,t}^{*,x} \right) \leq \sum_{i \in I} q_{i,t}^{*,x}\]
because $\sum_{i \in I} q_{i,t}^x - q_{i,t}^{*,x} \leq 0$. Finally, $q'^{,x}_{i,t} \in [0,1]$ because if $1 > q^{*,x}_{i,t} > 0$, one can set $c$ such that $q_{i,t}^{*,x} + c\delta_{i,t}^x \in [0,1]$. Otherwise, if $q^{*,x}_{i,t} = 1$, then $\delta_{i,t}^x \leq 0$, and if $q^{*,x}_{i,t} = 0$, then $\delta_{i,t}^x \geq 0$. In both cases, one can still set $c$ small enough such that $q_{i,t}^{*,x} + c\delta_{i,t}^x \in [0,1]$.

However, this implies that the feasible allocation $q'$ Pareto dominates $q^*$ under interim utilities $u$, contradicting DPM's interim efficiency.

\textbf{Step 2:} DPM is ex-post efficient. 

DPM outputs an ex-post allocation $a^* \in \mathcal{A}$. I claim that there does not exist any $a \in \mathcal{A}$ such that (a) $U_i(a_i) \geq U_i(a^*_i)$ for all $i \in I$ and $U_j(a_j) > U_j(a^*_j)$ for some $j \in I$ and (b) $\sum_{i \in I} a_{i,t}^x \leq |(x,t)|$, where $|(x,t)| = \sum_{i \in I} a^{*,x}_{i,t}$. Toward a contradiction, suppose that such an $a$ exists. 

Let $\Delta = a - a^*$. By assumption, for any $\lambda > 0$, $U_i(\lambda\Delta_i) \geq 0$ for all $i \in I$, and $U_j(\lambda\Delta_j) > 0$ for some $j \in I$. I will prove that $q^* + \lambda\Delta$ is a feasible interim allocation for some $\lambda > 0$. 

By the ex-post feasibility constraint (b):
\[\sum_{i \in I} \Delta_{i,t}^x = \sum_{i \in I} a_{i,t}^x - a_{i,t}^{*,x} \leq 0 \]
which implies that $\sum_{i \in I} q^{*,x}_{i,t} + \lambda \Delta_{i,t}^x \leq \sum_{i \in I} q^{*,x}_{i,t} = 1$ for all $(x,t) \in X \times [T]$ and $\lambda > 0$. Therefore, it only remains to show that $q^{*,x}_{i,t} + \lambda \Delta_{i,t}^x \in [0,1]$ for all $i \in I$ and $(x,t) \in X \times [T]$.

Immediately, $q^{*,x}_{i,t} + \lambda \Delta_{i,t}^x \in [0,1]$ if $\Delta_{i,t}^x = 0$. So, I will consider two cases where $\Delta_{i,t}^x \neq 0$: (i) $a^{*,x}_{i,t} = 1$ and $a^{x}_{i,t} = 0$ or (ii) $a^{*,x}_{i,t} = 0$ and $a^{x}_{i,t} = 1$.

Since $a^{*,x}_{i,t} = 1$ implies that $x$ arrives at $t$ (similarly, feasibility and $a_{i,t}^x = 1$ would imply that $x$ arrives at $t$), in case (i):
\[a^{*,x}_{i,t} = 1 \implies q^{*,x}_{i,t} > 0\]
Therefore, I can find a small enough $\lambda > 0$ such that:
\[q^{*,x}_{i,t} + \lambda (a^x_{i,t} - a^{*,x}_{i,t}) = q^{*,x}_{i,t} - \lambda > 0\]
In case (ii):
\[a^{*,x}_{i,t} = 0 \implies q^{*,x}_{i,t} < 1\]
so that, again, I can find a small enough $\lambda > 0$ such that:
\[q^{*,x}_{i,t} + \lambda (a^x_{i,t} - a^{*,x}_{i,t}) = q^{*,x}_{i,t} + \lambda < 1\]
guaranteeing that $q^* + \lambda \Delta$ is feasible. Finally, I have that $U_i(q_i^* + \lambda \Delta_i) \geq U_i(q_i^*)$ for all $i \in I$, and $U_j(q_j^* + \lambda \Delta_j) > U_j(q_j^*)$, contradicting the fact that $q^*$ is interim efficient under $U$. Therefore, $a^*$ must be ex-post efficient. 

\textbf{Step 3.} DPM's expected imbalance is $O(\sqrt{T/|I|})$.

Let $\pi$ be a realization (ex-post allocation) of $\pi_{DPM}$. Define:
\[Z_{i,t} = \sum_{x \in X} w^x \pi_{i,t}^x\]
be $i$'s total weighted ex-post allocation at $t$. Necessarily, $Z_{i,t} \in \{0\} \cup \{w^x : x \in X\}$. Then:
\[\sum_{t = 1}^T \distexpectation[F]{Z_{i,t}} = \sum_{t = 1}^T \sum_{x \in X} F_t(x) w^x q_{i,t}^{*,x} = \distexpectation[F,w]{q_i^*} = \frac{W_F}{|I|}\]
by interim balance. Define:
\[D_{i,t} = Z_{i,t} - \distexpectation[F]{Z_{i,t}} \implies W(\pi_i) - \frac{W_F}{|I|} = \sum_{t = 1}^T \Big[ Z_{i,t} - \distexpectation[F]{Z_{i,t}} \Big] = \sum_{t = 1}^T D_{i,t}\]
This implies:
\begin{align*}
    \E[F]*{\Big| W(\pi_i) - \frac{W_F}{|I|} \Big|} = \E[F]*{\Big| \sum_{t = 1}^T D_{i,t} \Big|} &\leq \sqrt{ \E[F]*{\Big( \sum_{t = 1}^T D_{i,t} \Big)^2} }\\
    &= \sqrt{Var\bigg( \sum_{t = 1}^T D_{i,t} \bigg)}\\
    &= \sqrt{ \sum_{t = 1}^T Var(Z_{i,t}) }\\
    &\leq \sqrt{ \sum_{t = 1}^T \E[F]*{Z_{i,t}^2} }
\end{align*}
The first line follows from Jensen's (let $g(x) = x^2$. $g(\expectation{|D_{i,t}|}) \leq \expectation{g(|D_{i,t}|)} = \expectation{g(D_{i,t})}$. Taking the square root: $\expectation{|D_{i,t}|} \leq \sqrt{\expectation{D_{i,t}^2}}$). The third equality follows by independence of arrivals and DPM across $t$. The final inequality holds by definition of variance. Further:
\[\sum_{t = 1}^T \distexpectation[F]{Z_{i,t}^2} = \sum_{t = 1}^T \sum_{x \in X} F_t(x) (w^x)^2 q_{i,t}^{*,x} \leq \sum_{t = 1}^T \sum_{x \in X} F_t(x) w^x q_{i,t}^{*,x} = \sum_{t = 1}^T \distexpectation[F]{Z_{i,t}} = \frac{W_F}{|I|}\]
because at most one $\pi_{i,t}^x$ is non-zero and this one term is exactly one. Moreover, since $w^x \leq 1$ for all $x$, I have $(w^x)^2 \leq w^x$. Therefore:
\[\frac{1}{|I|} \sum_{i \in I} \E[F]*{\Big| W(\pi_i) - \frac{W_F}{|I|} \Big|} \leq \sqrt{\frac{W_F}{|I|}} \leq \sqrt{\frac{T}{|I|}}\]
where the last inequality follows since $\sum_{x \in X} F_t(x) w^x \leq 1$ for all $t$.

\textbf{Step 4:} DPM is asymptotically ex-post balanced.

Let $q^k$ be the balanced equilibrium allocation that DPM computes for $\mathcal{M}^k$, so that $E_{F,w}[q_i^k] = kW_F/|I|$, and let $a$ be a realization of $\pi_{DPM}^k(v)$. The argument in Step 3 applied to $\mathcal{M}^k$ gives $\E[F]*{W(a_i)} = kW_F/|I|$ and $Var(W(a_i)) \le kW_F/|I|$. By Chebyshev's inequality, for any $\epsilon > 0$:
\[\Pr\Bigg\{ \bigg| \frac{1}{k} \sum_{t = 1}^{kT} \sum_{x \in X} w^x a_{i,t}^x - \frac{W_F}{|I|} \bigg| \ge \epsilon \Bigg\} \le \frac{Var(W(a_i))}{k^2\epsilon^2} \le \frac{W_F}{|I|k\epsilon^2} \to 0\]
for each $i \in I$. $\blacksquare$

\begin{center}
    \textbf{Proof of Proposition \ref{proposition:impossibility}}
\end{center}

The market $\mathcal{M}$ is:
\[I = \{i,j\} \quad X = \{x, y\} \quad T = 1 \quad F(x) = F(y) = \frac{1}{2} \quad w^x = w^y = 1\]
with $(v_i(x), v_i(y)) = (3,1)$ and $(v_j(x), v_j(y)) = (2,1)$. $W_F/|I| = 1/2$. Trivially, the preference distribution can be constructed to contain the valuations in this proof. Suppose that $\pi$ is ex-post efficient, asymptotically balanced, and strategyproof. Let $\hat v_j = (4,1)$. Fix $k$, let $\tilde a^k$ be the realized arrivals in $\mathcal{M}^k$, and let $a^k$ and $b^k$ be realizations of $\pi^k(v)$ and $\pi^k(\hat v_j, v_i)$, respectively. I write $\bar a_i^{k,z} := \sum_{t=1}^{k} a_{i,t}^{k,z}$ and $\bar a_i^k := \bar a_i^{k,x} + \bar a_i^{k,y}$ for $i \in I$, and define $\bar b_i^{k,z}$ and $\bar b_i^k$ analogously. I also write $|x^k| := \sum_{t=1}^{k} \tilde a_t^{k,x}$. Since $w^x = w^y = 1$, $W(a_i^k) = \bar a_i^k$ and $W(b_i^k) = \bar b_i^k$.

\textbf{Step 1.} $a^k$ and $b^k$ allocate every arrival. Otherwise, assigning a discarded task to $i$ is feasible for $\tilde a^k$ and strictly improves $i$. This implies that $\bar a_i^{k,x} + \bar a_j^{k,x} = |x^k|$ and $\bar a_i^k + \bar a_j^k = k$, and the same for $b^k$.

\textbf{Step 2.} $\bar a_j^{k,x} \le 1$ or $\bar a_i^{k,y} \le 4$. Similarly, $\bar b_i^{k,x} \le 1$ or $\bar b_j^{k,y} \le 6$.

Toward a contradiction, suppose that $\bar a_j^{k,x} \ge 2$ and $\bar a_i^{k,y} \ge 5$. Consider an alternative allocation $a^*$ that reassigns two of $j$'s copies of $x$ to $i$ and five of $i$'s copies of $y$ to $j$ at the same times. This is feasible for $\tilde a^k$. Moreover:
\[U_i(a_i^*) - U_i(a_i^k) = 2 \cdot 3 - 5 \cdot 1 > 0 \quad \text{ and } \quad U_j(a_j^*) - U_j(a_j^k) = 5 \cdot 1 - 2 \cdot 2 > 0\]
contradicting ex-post efficiency. Under $(\hat v_j, v_i)$, the same argument applies to reassigning two of $i$'s copies of $x$ to $j$ and seven of $j$'s copies of $y$ to $i$: $i$ gains $7 \cdot 1 - 2 \cdot 3 > 0$ and $j$ gains $2 \cdot 4 - 7 \cdot 1 > 0$ under the reports.

\textbf{Step 3.} I claim that:
\[\bar a_j^{k,x} \le \left||x^k| - \frac{k}{2}\right| + \left|\bar a_j^k - \frac{k}{2}\right| + 4 \quad \text{ and } \quad \bar b_i^{k,x} \le \left||x^k| - \frac{k}{2}\right| + \left|\bar b_j^k - \frac{k}{2}\right| + 6\]

If $\bar a_j^{k,x} \ge 2$, then $\bar a_i^{k,y} \le 4$ by Step 2, so $\bar a_i^{k,x} = \bar a_i^k - \bar a_i^{k,y} \ge \bar a_i^k - 4$. By Step 1:
\begin{align*}
    \bar a_j^{k,x} = |x^k| - \bar a_i^{k,x} &\le |x^k| - \bar a_i^k + 4 \\
    &= \left(|x^k| - \frac{k}{2}\right) + \left(\bar a_j^k - \frac{k}{2}\right) + 4 \\
    &\le \left||x^k| - \frac{k}{2}\right| + \left|\bar a_j^k - \frac{k}{2}\right| + 4
\end{align*}
where the last inequality is the triangle inequality. Otherwise, $\bar a_j^{k,x} \le 1$, and the bound holds trivially. Similarly, if $\bar b_i^{k,x} \ge 2$, then $\bar b_j^{k,y} \le 6$ by Step 2, so $\bar b_j^{k,x} \ge \bar b_j^k - 6$, and, by Step 1:
\begin{align*}
    \bar b_i^{k,x} = |x^k| - \bar b_j^{k,x} &\le |x^k| - \bar b_j^k + 6 \\
    &= \left(|x^k| - \frac{k}{2}\right) - \left(\bar b_j^k - \frac{k}{2}\right) + 6 \\
    &\le \left||x^k| - \frac{k}{2}\right| + \left|\bar b_j^k - \frac{k}{2}\right| + 6
\end{align*}
Otherwise, $\bar b_i^{k,x} \le 1$, and the second bound holds trivially.

\textbf{Step 4.} I claim that:
\[\frac{1}{k}\Big[u_j(\pi^k(\hat v_j, v_i)) - u_j(\pi^k(v))\Big] \ge \frac{1}{2} - \frac{1}{\sqrt k} - 2\E*{\left|\frac{\bar a_j^k}{k} - \frac{1}{2}\right|} - 2\E*{\left|\frac{\bar b_j^k}{k} - \frac{1}{2}\right|} - \frac{10}{k}\]

Definitionally:
\[u_j(\pi^k(v)) = \sum_{t = 1}^k \sum_{z \in X} F(z) \pi_{j,t}^{k,z} v_j(z) = \sum_{t = 1}^k \sum_{z \in X} \E*{a_{j,t}^{k,z}} v_j(z) = \E*{U_j(a_j^k)}\]
and, likewise, $u_j(\pi^k(\hat v_j, v_i)) = \E*{U_j(b_j^k)}$ under $j$'s true valuation. Since $v_j = (2,1)$, $U_j(a_j^k) = 2\bar a_j^{k,x} + \bar a_j^{k,y} = \bar a_j^k + \bar a_j^{k,x}$. By Step 3:
\begin{align*}
    U_j(a_j^k) = \bar a_j^k + \bar a_j^{k,x} &\le \bar a_j^k + \left||x^k| - \frac{k}{2}\right| + \left|\bar a_j^k - \frac{k}{2}\right| + 4 \\
    &= \frac{k}{2} + \left(\bar a_j^k - \frac{k}{2}\right) + \left||x^k| - \frac{k}{2}\right| + \left|\bar a_j^k - \frac{k}{2}\right| + 4 \\
    &\le \frac{k}{2} + \left||x^k| - \frac{k}{2}\right| + 2\left|\bar a_j^k - \frac{k}{2}\right| + 4
\end{align*}
Similarly, by Steps 1 and 3:
\begin{align*}
    U_j(b_j^k) = \bar b_j^k + |x^k| - \bar b_i^{k,x} &\ge \bar b_j^k + |x^k| - \left||x^k| - \frac{k}{2}\right| - \left|\bar b_j^k - \frac{k}{2}\right| - 6 \\
    &= \frac{k}{2} + \left(\bar b_j^k - \frac{k}{2}\right) + |x^k| - \left||x^k| - \frac{k}{2}\right| - \left|\bar b_j^k - \frac{k}{2}\right| - 6 \\
    &\ge \frac{k}{2} + |x^k| - \left||x^k| - \frac{k}{2}\right| - 2\left|\bar b_j^k - \frac{k}{2}\right| - 6
\end{align*}
Arrivals are independent across $t$ with $F(x) = 1/2$, so $|x^k|$ is Binomial with $\E*{|x^k|} = k/2$ and $Var(|x^k|) = k/4$. By Jensen's inequality:
\[\E*{\left||x^k| - \frac{k}{2}\right|} \le \sqrt{\E*{\left(|x^k| - \frac{k}{2}\right)^2}} = \sqrt{Var(|x^k|)} = \frac{\sqrt k}{2}\]
Taking expectations of the two bounds above and applying this inequality:
\begin{align*}
    u_j(\pi^k(v)) = \E*{U_j(a_j^k)} &\le \frac{k}{2} + \frac{\sqrt k}{2} + 2\E*{\left|\bar a_j^k - \frac{k}{2}\right|} + 4 \\
    u_j(\pi^k(\hat v_j, v_i)) = \E*{U_j(b_j^k)} &\ge \frac{k}{2} + \frac{k}{2} - \frac{\sqrt k}{2} - 2\E*{\left|\bar b_j^k - \frac{k}{2}\right|} - 6
\end{align*}
Subtracting the first inequality from the second:
\[u_j(\pi^k(\hat v_j, v_i)) - u_j(\pi^k(v)) \ge \frac{k}{2} - \sqrt k - 2\E*{\left|\bar a_j^k - \frac{k}{2}\right|} - 2\E*{\left|\bar b_j^k - \frac{k}{2}\right|} - 10\]
Dividing both sides by $k$:
\[\frac{1}{k}\Big[u_j(\pi^k(\hat v_j, v_i)) - u_j(\pi^k(v))\Big] \ge \frac{1}{2} - \frac{1}{\sqrt k} - 2\E*{\left|\frac{\bar a_j^k}{k} - \frac{1}{2}\right|} - 2\E*{\left|\frac{\bar b_j^k}{k} - \frac{1}{2}\right|} - \frac{10}{k}\]

\textbf{Step 5.} The right-hand side of Step 4 is strictly positive for $k$ large enough.

Since $\pi$ is asymptotically balanced for every report profile, choose $K \ge 100$ such that, for all $k \ge K$:
\[\Pr\left(\left|\frac{\bar a_j^k}{k} - \frac{1}{2}\right| \ge \frac{1}{32}\right) < \frac{1}{16} \quad \text{ and } \quad \Pr\left(\left|\frac{\bar b_j^k}{k} - \frac{1}{2}\right| \ge \frac{1}{32}\right) < \frac{1}{16}\]
Since $|\bar a_j^k/k - 1/2| \le 1/2$, splitting the expectation on this event gives:
\[\E*{\left|\frac{\bar a_j^k}{k} - \frac{1}{2}\right|} \le \frac{1}{32} + \frac{1}{2}\Pr\left(\left|\frac{\bar a_j^k}{k} - \frac{1}{2}\right| \ge \frac{1}{32}\right) < \frac{1}{16}\]
and the same holds for $\bar b_j^k$. Since $1/\sqrt k \le 1/10$ and $10/k \le 1/10$ for $k \ge 100$, for all $k \ge K$:
\[\frac{1}{k}\Big[u_j(\pi^k(\hat v_j, v_i)) - u_j(\pi^k(v))\Big] > \frac{1}{2} - \frac{1}{10} - \frac{1}{8} - \frac{1}{8} - \frac{1}{10} = \frac{1}{20} > 0\]
contradicting strategyproofness. $\blacksquare$

\begin{center}
    \textbf{Proof of Proposition \ref{proposition:approximate}}
\end{center}

Fix $\ell$ and suppress it. Let $S_i^K$, $\theta_i^x$, $\omega_i$, $\alpha$, and $\mu$ be defined as in the proof of Theorem \ref{thm:general} for the balanced equilibrium selected under truthful reports $v$. Under any report $\hat v_i$, DPM selects a balanced equilibrium for $(\hat v_i, v_{-i})$, so $\sum_{t=1}^T \sum_{x \in X} F_t(x) q_{i,t}^x = T/|I|$ and:
\[u_i(\pi^\ell(\hat v_i, v_{-i})) \le \frac{T}{|I|}\max_{y \in X} v_i(y)\]
Under truthful reporting, every $x \in S_i^K$ satisfies $v_i(x) \ge \max_{y \notin S_i^K} v_i(y)$ and $\sum_{x \in S_i^K} |I|\theta_i^x = 1 - \omega_i$, so:
\[\frac{|I|}{T}u_i(\pi^\ell(v)) = \sum_{x \in X} |I|\theta_i^x v_i(x) \ge (1 - \omega_i)\max_{y \notin S_i^K} v_i(y)\]
Then:
\[\hat u_i^\ell \le \max_{y \in X} v_i(y) - (1 - \omega_i)\max_{y \notin S_i^K} v_i(y) \le \Big(\max_{y \in X} v_i(y) - \max_{y \notin S_i^K} v_i(y)\Big) + \bar V\omega_i\]
So, $\hat u_i^\ell > \eta$ requires either $\omega_i > \eta/(2\bar V)$ or $\max_{y \in X} v_i(y) - \max_{y \notin S_i^K} v_i(y) \ge \eta/2$. Fix $\eta, \delta > 0$. As in Claim 5, choose $K$ such that $2\alpha\log K/K < \eta\delta/(8\bar V)$, $\mu \le 1$, and $\log K \ge 1 + \log\bar\rho$. By Claim 5, for all $\ell$ large enough, $\E*{\omega_i} < \eta\delta/(4\bar V)$. By Markov's inequality:
\[\Pr\left(\omega_i > \frac{\eta}{2\bar V}\right) \le \frac{2\bar V}{\eta}\E*{\omega_i} < \frac{\delta}{2}\]
Since $\max_{y \in X} v_i(y) \le \bar V$, by Assumption \ref{assumption:top}, for all $\ell$ large enough:
\[\Pr\left(\max_{y \in X} v_i(y) - \max_{y \notin S_i^K} v_i(y) \ge \frac{\eta}{2}\right) \le \Pr\left(\max_{y \notin S_i^K} v_i(y) \le \bar V - \frac{\eta}{2}\right) < \frac{\delta}{2}\]
By the union bound, $\Pr(\hat u_i^\ell > \eta) < \delta$ for all $\ell$ large enough. $\blacksquare$

%% file: Paper/ReducedForm.tex
\section{Productivity Bound}\label{appendix:reduced-form}

I derive a bound on DPM's expected productivity in this Appendix.
\subsection{Setup}

Assumption \ref{assumption:types} holds everywhere except in the proof for Theorem \ref{thm:general}. For brevity, I write $f := \frac{1}{T}\sum_{t=1}^T F_t(x)$, so $\frac{1}{T}\sum_{t=1}^T F_t(y) = 1-f$. I also write $Y_i^x := \CE[\beta_i^x]{Y_{i,t}^x}$ and $|I| := n$.

\subsection{Characterization of the DPM Allocation}

I show that the structure of equilibria in the $|X| = 2$ case has a tractable characterization. I introduce a few pieces of notation for this. Define:
\[N^x := \frac{W_F}{n w^x} \qquad N^y := \frac{W_F}{n w^y} \qquad \kappa := \frac{fT}{N^x} = \frac{n f w^x}{f w^x + (1-f) w^y}\]
$N^x$ ($N^y$) is the per-worker expected consumption of $x$ ($y$) under full rationing. $\kappa \in (0, n)$ measures the expected supply of type $x$ per rationed unit. Let $\theta := \kappa - \lfloor\kappa\rfloor \in [0,1)$, and define the rank weights:
\[\eta^x_k := \mathbbm{1}\{k \le \lfloor\kappa\rfloor\} + \theta\,\mathbbm{1}\{k = \lfloor\kappa\rfloor + 1\} \qquad \eta^y_k := 1 - \eta^x_k\]
$\eta^x_k$ is non-increasing in $k$, $\eta^y_k$ is non-decreasing in $k$, and $\sum_{k=1}^n \eta^x_k = \kappa$. For a worker $i$, I write the expected consumption of $z \in \{x,y\}$ as:
\[\tilde q_i^z := \sum_{t = 1}^T F_t(z) q_{i,t}^z\]
Fix realized valuations $v_i(\cdot)$ for all $i \in I$. $i$'s relative valuation for $x$ over $y$ is $V_i = v_i(x)/v_i(y)$. Since $v_i(x)$ is continuously distributed, I can index workers in descending order of their relative valuations with probability one: $V_{(1)} > V_{(2)} > \cdots > V_{(n)}$, where subscript $(k)$ denotes the worker with the $k$-th largest ratio.

\begin{proposition}\label{proposition:threshold}
    (Threshold Characterization) Every market-clearing balanced equilibrium allocation $q^*$ satisfies, with probability one:
    \[\big(\tilde q_{(k)}^x, \tilde q_{(k)}^y\big) = \big(\eta^x_k N^x, \eta^y_k N^y\big)\]
\end{proposition}

\begin{proof}
    Let $(p^*, r^*, q^*)$ be any market-clearing balanced equilibrium with feasible allocation. I write $e_{(z,t)} \in \mathcal{A}$ to denote the allocation with $e_{(z,t)} = 1$ and $e_{(z',k)} = 0$ for all $(z',k) \neq (z,t)$. Note the following facts that arise from the definition of market-clearing balanced equilibria:
    \begin{enumerate}
        \item (Full Allocation) $\sum_{i \in I} q_{i,t}^{*,z} = 1$ for all $(z,t) \in X \times [T]$.
        \item (Positivity) $r_i^* > 0$ for all $i \in I$ and $p_t^{*,z} > 0$ for all $(z,t) \in X \times [T]$.
        \item (Expenditure) $r_i^* p^* \cdot q_i^* = 1$ for all $i \in I$.
    \end{enumerate}

    \textbf{Step 1:} Prices per expected unit are constant over time: $p_t^{*,z}/F_t(z) = P^z$ for all $t$, for each $z \in X$.

    Toward a contradiction, suppose $p_t^{*,z}/F_t(z) > p_k^{*,z}/F_k(z)$ for some $z$ and $t \neq k$. By full allocation, $q_{i,t}^{*,z} > 0$ for some $i$. Consider
    \[q_i := q_i^* - \delta e_{(z,t)} + \delta \frac{p_t^{*,z}}{p_k^{*,z}} e_{(z,k)}\]
    for some $\delta \in (0, q_{i,t}^{*,z})$. Then $p^* \cdot q_i = p^* \cdot q_i^*$, so $q_i$ is affordable, yet:
    \[u_i(q_i) - u_i(q_i^*) = v_i(z)\, \delta \bigg(F_k(z)\frac{p_t^{*,z}}{p_k^{*,z}} - F_t(z)\bigg) > 0\]
    Then, $i$ would demand $q_i$ in equilibrium, a contradiction. Therefore, the step holds.

    \textbf{Step 2:} Define the cutoff $c := P^x/P^y > 0$. Then $V_i > c$ implies $\tilde q_i^{*,y} = 0$, and $V_i < c$ implies $\tilde q_i^{*,x} = 0$.

    $V_i > c$ implies:
    \[\frac{v_i(x)}{P^x} > \frac{v_i(y)}{P^y}\]
    Suppose $\tilde q_i^{*,y} > 0$, so that $q_{i,t'}^{*,y} > 0$ for some $t'$. Consider transferring expenditure: $q_i := q_i^* - \delta e_{(y,t')} + \delta (p_{t'}^{*,y}/p_1^{*,x}) e_{(x,1)}$ for small $\delta > 0$. This is affordable and, by Step 1, yields:
    \[u_i(q_i) - u_i(q_i^*) = \delta p_{t'}^{*,y} \bigg( \frac{v_i(x)}{P^x} - \frac{v_i(y)}{P^y} \bigg) > 0\]
    contradicting demand. So $\tilde q_i^{*,y} = 0$. The case $V_i < c$ is symmetric.

    \textbf{Step 3:} The Proposition holds.

    Interim balance $E_{F,w}[q_i^*] = W_F/n$ implies:
    \begin{equation}\label{eq:balance-counts}
        w^x \tilde q_i^x + w^y \tilde q_i^y = \frac{W_F}{n} \iff \frac{w^x \tilde q^x_i n}{W_F} + \frac{w^y \tilde q_i^y n}{W_F} = 1 \iff \frac{\tilde q_i^x}{N^x} + \frac{\tilde q_i^y}{N^y} = 1
    \end{equation}
    A worker with $\tilde q_i^y = 0$ therefore has $\tilde q_i^x = N^x$, and a worker with $\tilde q_i^x = 0$ has $\tilde q_i^y = N^y$. Partition $I$ into $A = \{i : V_i > c\}$, $B = \{i : V_i < c\}$, and $C = \{i : V_i = c\}$. By Step 2, every $i \in A$ has $(\tilde q_i^x, \tilde q_i^y) = (N^x, 0)$ and every $j \in B$ has $(\tilde q_j^x, \tilde q_j^y) = (0, N^y)$. By full allocation:
    \[\sum_{i \in I} \tilde q_i^x = \sum_{t=1}^T F_t(x) \sum_{i \in I} q_{i,t}^{*,x} = fT\]
    Dividing by $N^x$:
    \begin{equation}\label{eq:kappa-accounting}
        |A| + \sum_{i \in C} \frac{\tilde q_i^x}{N^x} = \frac{fT}{N^x} = \kappa
    \end{equation}
    Since $V_i$ is continuously distributed, $|C| \le 1$ with probability one, and $\tilde q_j^x/N^x \in [0,1]$ for $j \in C$ by \eqref{eq:balance-counts}. If $\theta > 0$, \eqref{eq:kappa-accounting} requires $|C| = 1$, $|A| = \lfloor\kappa\rfloor$, and $\tilde q_j^x = \theta N^x$ for the worker $j \in C$, who has rank $\lfloor\kappa\rfloor + 1$. If $\theta = 0$, either $|A| = \kappa$, or the worker in $C$ has $\tilde q_j^x \in \{0, N^x\}$ and rank $\kappa + 1$ or $\kappa$, respectively. In every case, \eqref{eq:balance-counts} gives $\big(\tilde q_{(k)}^x, \tilde q_{(k)}^y\big) = (\eta^x_k N^x, \eta^y_k N^y)$.
\end{proof}

\subsection{Expected Productivity}

Rotation's expected productivity is simple:

\begin{remark}
    Under Assumption \ref{assumption:linear}, $\mathcal{P}(\pi_R) = f g^x + (1 - f) g^y$.
\end{remark}

This follows because Rotation's interim allocation is equal to $1/|I|$ for all $(i,t,x)$ and does not depend on $(v, \beta)$. Using the characterization in Proposition \ref{proposition:threshold}, I provide an exact expression for DPM's expected productivity.

\begin{theorem}\label{theorem:reduced-form}
    (DPM's Expected Productivity) Under Assumption \ref{assumption:linear}:
    \[\mathcal{P}(\pi_{DPM}) = \frac{N^x}{T} \sum_{k = 1}^n \eta^x_k \E*{Y_{(k)}^x} + \frac{N^y}{T} \sum_{k = 1}^n \eta^y_k \E*{Y_{(k)}^y}\]
    where $Y_{(k)}^z$ is the productivity of the worker with the $k$-th largest $V_i$ in task $z$.
\end{theorem}

\begin{proof}
    Note that $v$ determines $q^* := \pi_{DPM}(v)$. The shocks $s_{i,t}^z$ realize ex-post and are independent of $(v, \beta)$, mean-zero, and identically distributed across $t$, so $\CE[v,\beta]{Y_{i,t}^z} = Y_i^z$. Then:
    \begin{align*}
        \mathcal{P}(\pi_{DPM}) &= \E*{\frac{1}{T} \sum_{t = 1}^T \sum_{z \in X} \sum_{i \in I} a_{i,t}^z Y_{i,t}^z} \\
        &= \E*{\frac{1}{T} \sum_{t = 1}^T \sum_{z \in X} \sum_{i \in I} \CE[v,\beta,s]{a_{i,t}^z} Y_{i,t}^z} \\
        &= \E*{\frac{1}{T} \sum_{t = 1}^T \sum_{z \in X} \sum_{i \in I} F_t(z) q_{i,t}^{*,z} Y_{i,t}^z} \\
        &= \frac{1}{T}\sum_{z \in X} \E*{\sum_{i \in I} \tilde q_i^z Y_i^z} \\
        &= \frac{1}{T}\sum_{z \in X} \E*{\sum_{k = 1}^n \tilde q_{(k)}^z Y_{(k)}^z} \\
        &= \frac{N^x}{T} \sum_{k = 1}^n \eta^x_k \E*{Y_{(k)}^x} + \frac{N^y}{T} \sum_{k = 1}^n \eta^y_k \E*{Y_{(k)}^y}
    \end{align*}
    The second equality holds by the law of iterated expectations. The third holds because, given $v$, the ex-post allocation depends only on arrivals and DPM's interim allocation. These are independent of $(\beta, s)$, and $\CE[v,\beta,s]{a_{i,t}^z} = \CE[v]{a_{i,t}^z} = F_t(z) q_{i,t}^{*,z}$ by the definition of $q$. The fourth holds by the law of iterated expectations and $\CE[v,\beta]{Y_{i,t}^z} = Y_i^z$. The fifth re-indexes workers by rank, and the last uses Proposition \ref{proposition:threshold}.
\end{proof}

It is useful to define an auxiliary variable that represents $i$'s reported comparative advantage:
\[\tilde{V}_i = \mathbbm{1}\{ v_i(x) > v_i(y) \}\]
where $\tilde{V}_{(k)}$ is the $\tilde V_i$ for the $i$ with the $k$-th largest $V_i$. Let:
\[\tilde g^x_{(k)} := \CE[\tilde V_{(k)} = 1]{Y_{(k)}^x} \quad \text{and} \quad \tilde g^{x,y}_{(k)} := \CE[\tilde V_{(k)} = 1]{Y_{(k)}^y}\]
as well as:
\[\tilde g^y_{(k)} := \CE[\tilde V_{(k)} = 0]{Y_{(k)}^y} \quad \text{and} \quad \tilde g^{y,x}_{(k)} := \CE[\tilde V_{(k)} = 0]{Y_{(k)}^x}\]
Note that $\Omega_{n,g}(a)$ is the standard CDF for the Binomial distribution with $n$ trials and success probability $g$. $G = \Pr(\tilde V_i = 1)$ and $1 - G = \Pr(\tilde V_i = 0)$ are the probabilities that a worker has a reported comparative advantage for $x$ or $y$, respectively. The above formula can be simplified:

\begin{corollary}\label{corollary:parametric}
    (Parametric Form)
    \begin{align*}
        \mathcal{P}(\pi_{DPM}) &= \frac{N^x}{T} \sum_{k = 1}^n \eta^x_k \bigg[ \tilde g_{(k)}^x(1 - \Omega_{n,G}(k - 1)) + \tilde g_{(k)}^{y,x}(1 - \Omega_{n,1-G}(n - k)) \bigg] \\
        &+ \frac{N^y}{T} \sum_{k = 1}^n \eta^y_k \bigg[ \tilde g_{(k)}^y(1 - \Omega_{n,1-G}(n - k)) + \tilde g_{(k)}^{x,y}(1 - \Omega_{n,G}(k - 1)) \bigg]
    \end{align*}
\end{corollary}

\begin{proof}
    The probability that $\tilde V_{(k)}$ is equal to $a$ is $h_{(k)}^{a} = \Pr(\tilde V_{(k)} = a)$. Additionally:
\[M^a = \sum_{i \in I} \mathbbm{1}\{ \tilde V_i = a \} \]
Then:
\[h_{(k)}^1 = \Pr(M^1 \geq k) \quad \text{ and } \quad h_{(k)}^0 = \Pr(M^0 \geq n - k + 1)\]
$M^1$ and $M^0$ are Binomial with $n$ trials and success probabilities $G$ and $1 - G$, respectively. Hence:
\[h_{(k)}^1 = 1 - \Omega_{n,G}(k - 1) \quad \text{ and } \quad h_{(k)}^0 = 1 - \Omega_{n,1 - G}(n - k) \]
By Theorem \ref{theorem:reduced-form} and the law of total expectation:
\begin{align*}
    \mathcal{P}(\pi_{DPM}) &= \frac{N^x}{T} \sum_{k = 1}^n \eta^x_k \bigg[ \CE[\tilde V_{(k)} = 1]{Y_{(k)}^x} h^1_{(k)} + \CE[\tilde V_{(k)} = 0]{Y_{(k)}^x} h^0_{(k)} \bigg] \\
    &+ \frac{N^y}{T} \sum_{k = 1}^n \eta^y_k \bigg[ \CE[\tilde V_{(k)} = 0]{Y_{(k)}^y} h^0_{(k)} + \CE[\tilde V_{(k)} = 1]{Y_{(k)}^y} h^1_{(k)} \bigg]
\end{align*}
Substituting the definitions of $\tilde g_{(k)}$ and $h_{(k)}$ proves the Corollary. 
\end{proof}

\subsection{Closed Form Solution}\label{appendix:solution}

The bound is
\[\underline{\mathcal{P}} := \frac{N^x}{T} \sum_{k = 1}^n \eta^x_k \bigg[ \tilde g^x h_{(k)}^1 + \tilde g^{y,x} h_{(k)}^0 \bigg] + \frac{N^y}{T} \sum_{k = 1}^n \eta^y_k \bigg[ \tilde g^y h_{(k)}^0 + \tilde g^{x,y} h_{(k)}^1 \bigg]\]

To establish Theorem \ref{thm:expected}, I prove one Lemma and two Corollaries following from the previous results. First, I note that Theorem \ref{thm:expected} extends to correlation in $\{(v_i(x), \beta_i^x)\}$ across $x$ if the following is assumed:

\begin{assumption}\label{assumption:relative}
    \textit{(Relative Monotonic Correlation)} $\CE[V_i]{Y_{i,t}^x}$ is non-decreasing and $\CE[V_i]{Y_{i,t}^y}$ is non-increasing in $V_i$.
\end{assumption}

This holds if utility is normalized to sum to the same fixed value for each worker, and it also holds for many common distributions where an analog of the monotone likelihood ratio property holds (\qcitenp{saumard2018efron}). Notably, it is also implied by i.i.d.\ preferences and ability across $(i,x)$ under the log-log concave density assumption.

\begin{lemma}\label{lemma:efron}
    Under Assumption \ref{assumption:non-decreasing}, if $\{(v_i(x), \beta_i^x)\}$ is i.i.d.\ and $\log v_i(x)$ has a log-concave density function for all $(i,x) \in I \times X$, then Assumption \ref{assumption:relative} holds.
\end{lemma}

\begin{proof}
    Note that:
    \[\log V_i = \log v_i(x) - \log v_i(y)\]
    where $\log v_i(x) \ind \log v_i(y)$. One has:
    \[\CE[V_i]{Y_i^x} = \CE[\log V_i]{Y_i^x} = \CE[\log V_i]{\CE[v_i(x)]{Y_i^x}}\]
    The first equality holds because $\log$ is a measurable, injective function. The second holds by the law of iterated expectation and $Y_i^x \ind v_i(y)$. By Assumption \ref{assumption:non-decreasing}, $\CE[v_i(x)]{Y_i^x}$ is non-decreasing in $v_i(x)$. For independent $X$ and $Z$ with log-concave densities and any non-decreasing $g$, $\CE[X + Z]{g(X)}$ is non-decreasing in $X + Z$ (\qcitenp{efron1965}). Log concavity is preserved for $-\log v_i(y)$, so this implies that $\CE[V_i]{Y_i^x}$ is non-decreasing in $V_i$. The argument is the same for $\CE[V_i]{Y_i^y}$.
\end{proof}

The following is true:

\begin{corollary}\label{corollary:conservative}
    (Attenuation) Under Assumptions \ref{assumption:linear} and \ref{assumption:non-decreasing}:
    \[
        \mathcal{P}(\pi_{DPM}) \geq \underline{\mathcal{P}}
    \]
\end{corollary}

\begin{proof}
A worker $i \in I$ has rank:
\[R_i := 1 + \sum_{j \ne i} \mathbbm{1}\{V_j > V_i\}\]
then, letting $\gamma_i^x := \eta_{R_i}^x$
\[\sum_{k=1}^n \eta^x_k Y_{(k)}^x = \sum_{i \in I} \gamma_i^x Y_i^x\]
and $\sum_{k=1}^n \eta^x_k \mathbbm{1}\{\tilde V_{(k)} = 1\} = \sum_{i \in I} \gamma^x_i \tilde V_i$. By Theorem \ref{theorem:reduced-form}:
\begin{align*}
    \mathcal{P}(\pi_{DPM}) - \underline{\mathcal{P}} &= \frac{N^x}{T}\sum_{i \in I}\E*{\gamma^x_i\Big(Y_i^x - \tilde g^x \tilde V_i - \tilde g^{y,x}(1 - \tilde V_i)\Big)} \\
    &+ \frac{N^y}{T}\sum_{i \in I}\E*{\gamma^y_i\Big(Y_i^y - \tilde g^{x,y} \tilde V_i - \tilde g^{y}(1 - \tilde V_i)\Big)}
\end{align*}
by re-indexing and because $h_{(k)}^1 = \Pr(\tilde{V}_{(k)} = 1) = \E*{\tilde V_{(k)}}$, and $h_{(k)}^0 = 1- \E*{\tilde V_{(k)}}$. 

Fix $i \in I$. Since $\tilde V_i \in \{0,1\}$, the $i$-th term of the first sum splits as:
\[\E*{\gamma^x_i\Big(Y_i^x - \tilde g^x \tilde V_i - \tilde g^{y,x}(1 - \tilde V_i)\Big)} = \E*{\gamma^x_i\big(Y_i^x - \tilde g^x\big)\tilde V_i} + \E*{\gamma^x_i\big(Y_i^x - \tilde g^{y,x}\big)(1 - \tilde V_i)}\]
Note $\tilde V_i = \mathbbm{1}\{V_i > 1\}$ and $\tilde g^x = \CE[V_i > 1]{\CE[V_i]{Y_i^x}}$. Since $\gamma_i^x \ind Y_i^x$ given $V_i$, the law of iterated expectations gives:
\begin{align*}
    \E*{\gamma^x_i\big(Y_i^x - \tilde g^x\big)\tilde V_i} &= \Pr(V_i > 1)\,\CELR[V_i > 1]{\CE[V_i]{\gamma^x_i}\Big(\CE[V_i]{Y_i^x} - \tilde g^x\Big)} \\
    &= \Pr(V_i > 1)\,\mathrm{Cov}\Big(\CE[V_i]{\gamma^x_i}, \CE[V_i]{Y_i^x} \,\Big|\, V_i > 1\Big) \ge 0
\end{align*}
where the second equality holds because $\CE[V_i]{Y_i^x} - \tilde g^x$ has mean zero given $V_i > 1$, and the inequality is Chebyshev's integral inequality. $\CE[V_i]{\gamma^x_i}$ is non-decreasing in $V_i$ because $R_i$ is non-increasing in $V_i$ and $\eta^x_k$ is non-increasing in $k$. $\CE[V_i]{Y_i^x}$ is non-decreasing in $V_i$ by Assumption \ref{assumption:relative}.

The same steps on $\{V_i < 1\}$, with $\tilde g^{y,x} = \CE[V_i < 1]{\CE[V_i]{Y_i^x}}$, show that the second term is non-negative. Summing over $i$, the first sum is non-negative. The argument is symmetric for the $i$-th term of the second sum. This proves the Corollary.
\end{proof}

Moreover, the bound is always above Rotation productivity.

\begin{corollary}\label{corollary:bound}
    Under Assumptions \ref{assumption:linear}, \ref{assumption:non-decreasing}, and \ref{assumption:types}, $\underline{\mathcal{P}} \ge \mathcal{P}(\pi_{R})$.
\end{corollary}

\begin{proof}
    One can write:
    \begin{align*}
        \tilde g^x h_{(k)}^1 + \tilde g^{y,x} h_{(k)}^0 &= ( \tilde g^x - \tilde g^{y,x}) h_{(k)}^1 + \tilde g^{y,x} \\
        &= g^x + (h^1_{(k)} - G)(\tilde g^x - \tilde g^{y,x})
    \end{align*}
    where the first equality follows by $h^0_{(k)} = 1 - h^1_{(k)}$, and the second by the law of total expectation: $g^x = G \tilde g^x + (1 - G) \tilde g^{y,x}$. Similarly:
    \[\tilde g^y h_{(k)}^0 + \tilde g^{x,y} h_{(k)}^1 = g^y - (h_{(k)}^1 - G)(\tilde g^y - \tilde g^{x,y})\]
    Definitionally: 
    \[\frac{N^x}{T}\sum_{k=1}^n \eta^x_k = \frac{N^x\kappa}{T} = f \quad \text{and} \quad \frac{N^y}{T}\sum_{k=1}^n \eta^y_k = \frac{N^y(n - \kappa)}{T} = 1 - f\] Substituting:
    \begin{align*}
        \underline{\mathcal{P}} &= f g^x + (1 - f) g^y + \frac{N^x(\tilde g^x - \tilde g^{y,x})}{T} \sum_{k = 1}^n \eta^x_k (h^1_{(k)} - G) - \frac{N^y(\tilde g^y - \tilde g^{x,y})}{T} \sum_{k = 1}^n \eta^y_k (h^1_{(k)} - G)\\
        &= f g^x + (1 - f) g^y + \left[ \frac{N^x(\tilde g^x - \tilde g^{y,x}) + N^y(\tilde g^y - \tilde g^{x,y})}{T} \right] \sum_{k = 1}^n \eta^x_k (h^1_{(k)} - G)
    \end{align*}
    where the last equality follows since $\eta^y_k = 1 - \eta^x_k$ and $\sum_{k = 1}^n (h^1_{(k)} - G) = 0$. Then:
    \[\underline{\mathcal{P}} - \mathcal{P}(\pi_R) = \left[ \frac{N^x(\tilde g^x - \tilde g^{y,x}) + N^y(\tilde g^y - \tilde g^{x,y})}{T} \right] \sum_{k = 1}^n \eta^x_k (h^1_{(k)} - G)\]
    By Assumption \ref{assumption:relative}, $\tilde g^x \ge \tilde g^{y,x}$ and $\tilde g^y \ge \tilde g^{x,y}$. The bracketed term is non-negative. Since $h^1_{(k)} = \Pr(M^1 \ge k)$ is strictly decreasing in $k$ when $G \in (0,1)$, and $\eta^x_k$ is non-increasing, Chebyshev's sum inequality holds:
    \[\sum_{k = 1}^n \eta^x_k h^1_{(k)} \ge \frac{1}{n} \sum_{k = 1}^n h^1_{(k)} \sum_{k = 1}^n \eta^x_k = G\kappa = G\sum_{k=1}^n \eta^x_k\]
    so $\sum_{k = 1}^n \eta^x_k (h^1_{(k)} - G) \ge 0$. Therefore, $\underline{\mathcal{P}} \ge \mathcal{P}(\pi_R)$. The symmetric argument can be applied for decreasing monotonicity. This proves the Corollary.
\end{proof}

These results establish Theorem \ref{thm:expected}:

\begin{center}
    \textbf{Proof of Theorem \ref{thm:expected}}
\end{center}

By Corollary \ref{corollary:bound}:
\[\underline{\mathcal{P}} \ge \mathcal{P}(\pi_R)\]
By Corollary \ref{corollary:conservative}:
\[\mathcal{P}(\pi_{DPM}) \ge \underline{\mathcal{P}}\]
Therefore:
\[ \mathcal{P}(\pi_{DPM}) \ge\underline{\mathcal{P}} \ge \mathcal{P}(\pi_R)\]
$\blacksquare$

\subsection{General Asymptotics}\label{appendix:general}

Here, I extend the results to show that the bound remains conservative in large task-replica markets. As with Theorem \ref{thm:expected}, I prove Theorem \ref{thm:general} under weaker assumptions that allow valuations to be correlated and non-identically distributed across task types. Valuations remain i.i.d.\ across workers. Let $S_i^k$ be the set of $i$'s $k$ most preferred task types, and let $D_i(X') := \argmax_{x \in X'} v_i(x)$ be $i$'s favorite task in $X' \subseteq X$. The sufficient condition is that all tasks have a bounded, non-zero probability of being a worker's $K$-th favorite task conditional on her $K-1$ favorite tasks. Intuitively: there is enough preference heterogeneity that knowing a worker's top $K$ favorite tasks does not rule out any $x$ being her next favorite task with some uniformly bounded-below probability. This rules out cases such as extreme popularity or correlation differences where $v_i(x) > v_i(y)$ for all realizations.

\begin{assumption}\label{assumption:choice}
    \textit{(Preference Heterogeneity)} There exists $C \ge 1$ such that, for all $\ell$, $i \in I$, $k < |X|$, $X' \subseteq X$ with $|X'| = k$ and $\Pr(S_i^k = X') > 0$, and $x \in X \setminus X'$:
    \[\Pr\big(D_i(X \setminus X') = x \mid S_i^k = X'\big) \ge \frac{1}{C|X|}\]
    where $S_i^0 := \emptyset$.
\end{assumption}

Additionally, I assume the following regularity conditions:

\begin{assumption}\label{assumption:top}
    \textit{(Regularity)}
    \begin{enumerate}
        \item $\CE[v]{Y_{i,t}^x} = \psi(v_i(x))$ for some $\psi$ that is the same for all $i,t,x$.
        \item For every $K$ and $\delta > 0$, $\Pr\big(\max_{y \notin S_i^K} v_i(y) \le \bar V - \delta\big) \to 0$ as $\ell \to \infty$.
    \end{enumerate}
\end{assumption}

The first part requires that only preferences for $x$ directly influence productivity for $x$ and that the production function is homogeneous across tasks. The second part requires that some draw from an unbounded population of tasks is near $i$'s bliss point. Both assumptions are implied by the setup in the main text:

\begin{lemma}\label{lemma:iid}
    If $\{(v_i(x), \beta_i^x)\}$ is i.i.d.\ across $(i,x) \in I \times X$ and $\{s_{i,t}^x\}$ is mean-zero, independent of $(v, \beta)$, and identically distributed across $t$, then Assumptions \ref{assumption:choice} and \ref{assumption:top} hold.
\end{lemma}

\begin{proof}
    Given $S_i^k = X'$, the valuations $\{v_i(x)\}_{x \in X \setminus X'}$ are exchangeable and continuously distributed, so $D_i(X \setminus X')$ is uniformly distributed over $X \setminus X'$:
    \[\Pr\big(D_i(X \setminus X') = x \mid S_i^k = X'\big) = \frac{1}{|X| - |X'|} \ge \frac{1}{|X|}\]
    For the first part of Assumption \ref{assumption:top}, $(v_i(x), \beta_i^x) \ind \{(v_j(y), \beta_j^y)\}_{(j,y) \neq (i,x)}$ and $s_{i,t}^x \ind (v,\beta)$. $(v_i(x), \beta_i^x)$ is also identically distributed, and $s_{i,t}^x$ is mean-zero. This gives: 
    \[\CE[v]{Y_{i,t}^x} = \CE[v_i(x)]{Y_{i,t}^x} = \CE[v_i(x)]{\beta_i^x}\] 
    For the second part, $\max_{y \notin S_i^K} v_i(y)$ is the $(K+1)$-th largest of $\ell$ i.i.d.\ draws with $\sup \Supp{v_i(y)} = \bar V$, so $\Pr\big(\max_{y \notin S_i^K} v_i(y) \le \bar V - \delta\big) \to 0$ as $\ell \to \infty$.
\end{proof}

\begin{center}
    \textbf{Proof of Theorem \ref{thm:general}}
\end{center}

Fix a task-replica market $\mathcal{M}^\ell$ satisfying Assumptions \ref{assumption:equal}, \ref{assumption:choice}, and \ref{assumption:top}. I suppress $\ell$ and write $|X| = \ell$ and $|I| = h(\ell)$. Note that $\sum_{x \in X} \frac{1}{T}\sum_{t=1}^T F_t(x) = 1$ implies $\bar\rho \ge 1$. Suppose that $(p^*, r^*, q^*)$ is a market-clearing balanced equilibrium. Theorem \ref{thm:first-welfare} implies that $q^*$ is interim efficient. The mass of $x$ that $i$ consumes is:
\[\theta_i^x := \frac{1}{T}\sum_{t=1}^T F_t(x) q_{i,t}^{*,x}\]
Note that $W_F = \sum_{t=1}^T \sum_{x \in X} F_t(x) = T$ under unit weights. Interim balance gives $\sum_{x \in X} \theta_i^x = 1/|I|$ for all $i \in I$, and feasibility gives $\sum_{i \in I} \theta_i^x \le \frac{1}{T}\sum_{t=1}^T F_t(x)$ for all $x \in X$. A \textit{pair} is some $(i,x) \in I \times X$ such that $\theta_i^x > 0$ where:
\[\mathbf{J} := \{(i,x) : \theta_i^x > 0\}\]
is the set of all pairs. Let $\mathcal{J} \subseteq \mathbf{J}$ be a set of pairs. Then:
\[X^*(\mathcal{J}) = \{x \in X : (i,x) \in \mathcal{J} \text{ for some } i \in I\}\]
is the set of task types that pairs in $\mathcal{J}$ consume.

\textbf{Claim 1.} For every nonempty set of pairs $\mathcal{J}$, there exists some $(i,x) \in \mathcal{J}$ such that $x = D_i(X^*(\mathcal{J}))$.

Toward a contradiction, suppose that there exists some $\mathcal{J}$ such that there does not exist any $(i,x) \in \mathcal{J}$ with $x = D_i(X^*(\mathcal{J}))$. Consider a function $\gamma(\mathcal{J}, \cdot)$ defined by $\gamma(\mathcal{J}, (i,y)) = \{(j,x) \in \mathcal{J} : x = D_i(X^*(\mathcal{J}))\}$. Abusing notation, let this be any arbitrary single-valued selection of this function. If $\gamma(\mathcal{J}, (i,y)) = (i,x)$, then $(i,x) \in \mathcal{J}$ and $x = D_i(X^*(\mathcal{J}))$, contradicting the assumption. So, the owner of $\gamma(\mathcal{J}, (i,y))$ is not $i$ for each $(i,y) \in \mathcal{J}$. This implies there exists some cycle of distinct pairs:
\[\left( (i_1, x_0), (i_2, x_1), ..., (i_{\bar k}, x_{\bar k - 1}) \right)\]
such that $(i_{k+1}, x_k) = \gamma(\mathcal{J}, (i_k, x_{k-1}))$ for $k < \bar k$, and $\gamma(\mathcal{J}, (i_{\bar k}, x_{\bar k - 1})) = (i_1, x_0)$, where $x_0 := x_{\bar k}$ and $i_{\bar k + 1} := i_1$. Then $x_k = D_{i_k}(X^*(\mathcal{J}))$ for each $k$. Fix $t_k$ with $q_{i_{k+1},t_k}^{*,x_k} > 0$. For $\epsilon \in (0, \min_k F_{t_k}(x_k) q_{i_{k+1},t_k}^{*,x_k}]$, let $q$ satisfy $q_{i_k,t_k}^{x_k} = q_{i_k,t_k}^{*,x_k} + \epsilon/F_{t_k}(x_k)$ and $q_{i_{k+1},t_k}^{x_k} = q_{i_{k+1},t_k}^{*,x_k} - \epsilon/F_{t_k}(x_k)$ for each $k$, and $q = q^*$ otherwise. Then $q \ge 0$ and $\sum_i q_{i,t}^x = \sum_i q_{i,t}^{*,x}$ for all $(x,t)$, so $q$ is feasible. Each appearance of a worker $i_k$ in the cycle changes her utility by:
\[\epsilon \big[ v_{i_k}(x_k) - v_{i_k}(x_{k-1}) \big] > 0\]
since $x_{k-1} \in X^*(\mathcal{J}) \setminus \{x_k\}$ and $D_i(\cdot)$ is single-valued w.p. one. A worker who appears more than once in the cycle gains the sum of these terms, which is also strictly positive. So $q$ Pareto dominates $q^*$, a contradiction. The claim holds.

Let $B := \{(i,x) : \theta_i^x > 0, \, x \notin S_i^K\}$.

\textbf{Claim 2.} There is an ordering $(i_1, x_1), ..., (i_{|B|}, x_{|B|})$ of $B$ such that $S_{i_a}^K \cap X^*(B_a) = \emptyset$ for each $a$, where $B_a := \{(i_a, x_a), ..., (i_{|B|}, x_{|B|})\}$.

Let $B_1 := B$. Given $B_a \ne \emptyset$, Claim 1 allows one to define $(i_a, x_a) \in B_a$ with $x_a = D_{i_a}(X^*(B_a))$. Let $B_{a+1} := B_a \setminus \{(i_a, x_a)\}$. Since $x_a \notin S_{i_a}^K$, $v_{i_a}(y) > v_{i_a}(x_a)$ for all $y \in S_{i_a}^K$. If some $y \in S_{i_a}^K$ were in $X^*(B_a)$, then $x_a \ne D_{i_a}(X^*(B_a))$, a contradiction. The claim holds.

\textbf{Claim 3.} For every set of pairs $\mathcal{J}$:
\[\sum_{(i,x) \in \mathcal{J}} \theta_i^x \le \frac{\bar\rho}{|X|} |X^*(\mathcal{J})|\]

Each $(i,x) \in \mathcal{J}$ has $x \in X^*(\mathcal{J})$. By feasibility and Assumption \ref{assumption:equal}:
\[\sum_{(i,x) \in \mathcal{J}} \theta_i^x \le \sum_{x \in X^*(\mathcal{J})} \sum_{i \in I} \theta_i^x \le \sum_{x \in X^*(\mathcal{J})} \frac{1}{T}\sum_{t=1}^T F_t(x) \le \frac{\bar\rho}{|X|}|X^*(\mathcal{J})|\]
The claim holds.

\textbf{Claim 4.} Let $M := \sum_{(i,x) \in B} \theta_i^x$, $\mu \in (0,1]$, $m := \lceil \mu|I| \rceil$, and $s := \lceil \mu|X|/\bar\rho \rceil$. Then:
\[\Pr\left(M \ge 2\mu + \frac{1}{|I|}\right) \le \exp\left\{m\left[\log\frac{e}{\mu} - \frac{K\mu}{2C\bar\rho}\right] + s\left[\log\frac{e\bar\rho}{\mu} - \frac{K\mu}{2Cc}\right]\right\}\]

Suppose $M \ge 2\mu + 1/|I|$. Let $a^*$ be the smallest index such that $\sum_{a \le a^*} \theta_{i_a}^{x_a} \ge \mu$. Since $\theta_i^x \le \sum_{y \in X} \theta_i^y = 1/|I|$ for every pair, $\sum_{a \le a^*} \theta_{i_a}^{x_a} < \mu + 1/|I|$, so:
\[\sum_{(i,x) \in B_{a^*+1}} \theta_i^x = M - \sum_{a \le a^*} \theta_{i_a}^{x_a} \ge 2\mu + \frac{1}{|I|} - \mu - \frac{1}{|I|} \ge \mu\]
Claim 3 gives $|X^*(B_{a^*+1})| \ge \mu|X|/\bar\rho$, so $|X^*(B_{a^*+1})| \ge s$. For $a \le a^*$, $B_a \supseteq B_{a^*+1}$ and $X^*(B_{a^*+1}) \subseteq X^*(B_a)$. Claim 2 gives $S_{i_a}^K \cap X^*(B_a) = \emptyset$, which implies $S_{i_a}^K \cap X^*(B_{a^*+1}) = \emptyset$ for all $a \le a^*$. 

Let $I^* := \{i_1, ..., i_{a^*}\}$ be the set of distinct workers among the first $a^*$ pairs. Each of these pairs belongs to a worker in $I^*$, and each worker's pairs have total mass $1/|I|$ by interim balance, so:
\[\mu \le \sum_{a \le a^*} \theta_{i_a}^{x_a} \le \sum_{i \in I^*} \sum_{x \in X} \theta_i^x = \frac{|I^*|}{|I|}\]
so $|I^*| \ge \mu|I|$. Since $|I^*|$ is an integer, $|I^*| \ge m$.

So, the event $M \ge 2\mu + 1/|I|$ implies that there exist $W \subseteq I$ with $|W| = m$ and $Z \subseteq X$ with $|Z| = s$ such that $S_i^K \cap Z = \emptyset$ for all $i \in W$. Hence:
\[\left\{M \ge 2\mu + \frac{1}{|I|}\right\} \subseteq \bigcup_{\substack{W \subseteq I : |W| = m \\ Z \subseteq X : |Z| = s}} \big\{S_i^K \cap Z = \emptyset \text{ for all } i \in W\big\}\]
Fix any deterministic $(W,Z)$. If $K > |X| - s$, the probability below is zero. Otherwise:
\begin{align*}
    \Pr\big(S_i^K \cap Z = \emptyset \text{ for all } i \in W\big) &= \prod_{i \in W} \Pr\big(S_i^K \cap Z = \emptyset\big) \\
    &= \prod_{i \in W} \prod_{k=0}^{K-1} \Pr\big(S_i^{k+1} \cap Z = \emptyset \mid S_i^k \cap Z = \emptyset\big) \\
    &\le \left(1 - \frac{s}{C|X|}\right)^{Km} \\
    &\le \exp\left\{-\frac{Kms}{C|X|}\right\}
\end{align*}
The first equality holds because valuations are independent across workers. The second holds because $S_i^{k+1} = S_i^k \cup \{D_i(X \setminus S_i^k)\}$, so the events $\mathbbm{1}\{S_i^k \cap Z = \emptyset\}$ are decreasing in $k$. The first inequality holds by the bound below. The second holds because $1 - a \le e^{-a}$.

For the bound, let $\mathcal{X}_k := \{X' \subseteq X \setminus Z : |X'| = k\}$. Then:
\begin{align*}
    \Pr\big(S_i^{k+1} \cap Z = \emptyset \mid S_i^k \cap Z = \emptyset\big) &= \sum_{X' \in \mathcal{X}_k} \Pr\big(S_i^k = X' \mid S_i^k \cap Z = \emptyset\big) \\
    &\qquad \times \Pr\big(D_i(X \setminus X') \notin Z \mid S_i^k = X'\big) \\
    &= \sum_{X' \in \mathcal{X}_k} \Pr\big(S_i^k = X' \mid S_i^k \cap Z = \emptyset\big) \\
    &\qquad \times \Big(1 - \sum_{z \in Z} \Pr\big(D_i(X \setminus X') = z \mid S_i^k = X'\big)\Big) \\
    &\le 1 - \frac{s}{C|X|}
\end{align*}
The first equality conditions on the realization of $S_i^k$. The second holds because $Z \subseteq X \setminus X'$ for each $X' \in \mathcal{X}_k$. The inequality holds by Assumption \ref{assumption:choice} and because the weights $\Pr(S_i^k = X' \mid S_i^k \cap Z = \emptyset)$ sum to one over $\mathcal{X}_k$.

Finally, there are $\binom{|I|}{m}\binom{|X|}{s}$ pairs $(W,Z)$, so, by the union bound:
\begin{align*}
    \Pr\left(M \ge 2\mu + \frac{1}{|I|}\right) &\le \binom{|I|}{m}\binom{|X|}{s}\exp\left\{-\frac{Kms}{C|X|}\right\} \\
    &\le \left(\frac{e}{\mu}\right)^m \left(\frac{e\bar\rho}{\mu}\right)^s \exp\left\{-\frac{Kms}{C|X|}\right\} \\
    &\le \exp\left\{m\left[\log\frac{e}{\mu} - \frac{K\mu}{2C\bar\rho}\right] + s\left[\log\frac{e\bar\rho}{\mu} - \frac{K\mu}{2Cc}\right]\right\}
\end{align*}
The second inequality holds because $\binom{n}{k} \le (en/k)^k$, with $m \ge \mu|I|$ and $s \ge \mu|X|/\bar\rho$. The third holds because:
\[\frac{Kms}{C|X|} = \frac{Km}{2}\cdot\frac{s}{C|X|} + \frac{Ks}{2}\cdot\frac{m}{C|X|} \ge \frac{Km\mu}{2C\bar\rho} + \frac{Ks\mu}{2Cc}\]
using $s/|X| \ge \mu/\bar\rho$ and, since $|X| \le c|I|$, $m/|X| \ge \mu/c$. This proves the claim.

Let $\alpha := 6C(\bar\rho + c)$ and $\mu := \alpha\log K/K$. Suppose that $\log K \ge 1 + \log\bar\rho$ and $\mu \le 1$. Since $K$ grows faster than $\log K$, both are guaranteed for large enough $K$. $\omega_i := |I|\sum_{x \notin S_i^K} \theta_i^x$ is the share of $i$'s mass outside of her $K$ favorite task types.

\textbf{Claim 5.} For every $i \in I$:
\[\E*{\omega_i} \le \frac{2\alpha\log K}{K} + \frac{1}{|I|} + \exp\left\{-\frac{\alpha|I|(\log K)^2}{K}\right\}\]

First, I bound the bracketed terms in Claim 4:
\begin{align*}
    \log\frac{e}{\mu} - \frac{K\mu}{2C\bar\rho} &\le 1 + \log K - \frac{\alpha\log K}{2C\bar\rho} \\
    &\le 2\log K - 3\log K = -\log K
\end{align*}
The first inequality holds because $K\mu = \alpha\log K$ and $\mu \ge 1/K$, since $\alpha\log K \ge 1$. The second holds because $\log K \ge 1$ and $\alpha/(2C\bar\rho) \ge 3$. Similarly:
\begin{align*}
    \log\frac{e\bar\rho}{\mu} - \frac{K\mu}{2Cc} &\le 1 + \log\bar\rho + \log K - \frac{\alpha\log K}{2Cc} \\
    &\le 2\log K - 3\log K = -\log K
\end{align*}
where the second inequality holds because $\log K \ge 1 + \log\bar\rho$ and $\alpha/(2Cc) \ge 3$.

Next, by Claim 4:
\begin{align*}
    \Pr\left(M \ge 2\mu + \frac{1}{|I|}\right) &\le \exp\big\{-(m + s)\log K\big\} \\
    &\le \exp\big\{-\mu|I|\log K\big\} \\
    &= \exp\left\{-\frac{\alpha|I|(\log K)^2}{K}\right\}
\end{align*}
The second inequality holds because $m \ge \mu|I|$. The equality substitutes $\mu = \alpha\log K/K$.

Finally:
\begin{align*}
    \E*{\omega_i} &= \frac{1}{|I|}\sum_{j \in I} \E*{\omega_j} = \E*{M} \\
    &= \E*{M \mathbbm{1}\left\{M < 2\mu + \frac{1}{|I|}\right\}} + \E*{M \mathbbm{1}\left\{M \ge 2\mu + \frac{1}{|I|}\right\}} \\
    &\le 2\mu + \frac{1}{|I|} + \Pr\left(M \ge 2\mu + \frac{1}{|I|}\right) \\
    &\le \frac{2\alpha\log K}{K} + \frac{1}{|I|} + \exp\left\{-\frac{\alpha|I|(\log K)^2}{K}\right\}
\end{align*}
The first equality holds because, under an anonymous selection of balanced equilibria and i.i.d.\ valuations, $\E*{\omega_j}$ is constant across $j \in I$. The second holds because $\sum_{j \in I} \omega_j = |I|M$ by definition. The first inequality holds because $M \le \sum_{i \in I}\sum_{x \in X} \theta_i^x = 1$. The second holds by the bound above. The claim holds.

\textbf{Claim 6.}  There exists $L$ such that, for all $\ell \ge L$:
\[\mathcal{P}(\pi_{DPM}^\ell) > \mathcal{P}(\pi_{DPM}^{2,\ell})\]

Let $\psi(v_i(x)) := \CE[v_i(x)]{Y_{i,t}^x}$ and $\bar\psi := \sup\{\psi(a) : a \in \Supp{v_i(x)}, a < \bar V\}$. First, I re-write expected productivity:
\begin{align*}
    \mathcal{P}(\pi_{DPM}^\ell) &= \E*{\frac{1}{T}\sum_{t=1}^T \sum_{x \in X} \sum_{i \in I} a_{i,t}^x Y_{i,t}^x} \\
    &= \E*{\frac{1}{T}\sum_{t=1}^T \sum_{x \in X} \sum_{i \in I} \CE[v,\beta,s]{a_{i,t}^x} Y_{i,t}^x} \\
    &= \E*{\frac{1}{T}\sum_{t=1}^T \sum_{x \in X} \sum_{i \in I} F_t(x) q_{i,t}^{*,x} Y_{i,t}^x} \\
    &= \E*{\CELR[v]{\frac{1}{T}\sum_{t=1}^T \sum_{x \in X} \sum_{i \in I} F_t(x) q_{i,t}^{*,x} Y_{i,t}^x}}\\
    &= \E*{\sum_{x \in X} \sum_{i \in I} \frac{1}{T}\sum_{t=1}^T F_t(x) q_{i,t}^{*,x} \psi(v_i(x))}\\
    &= \sum_{i \in I}\sum_{x \in X} \E*{\theta_i^x \psi(v_i(x))} \\
    &= \frac{1}{|I|}\sum_{i \in I} \E*{\Psi_i} = \E*{\Psi_i}
\end{align*}
where $\Psi_i := \sum_{x \in X} |I|\theta_i^x \psi(v_i(x))$. The first equality holds by Assumption \ref{assumption:linear}. The second holds by the law of iterated expectations, since $Y_{i,t}^x$ is a function of $(\beta, s)$. The third holds because, given $v$, the ex-post allocation depends only on arrivals and DPM's allocation. These are independent of $(\beta, s)$, so $\CE[v,\beta,s]{a_{i,t}^x} = F_t(x) q_{i,t}^{*,x}$ by the definition of $q$. The fourth holds by the law of iterated expectations. The fifth holds because $q^*$ is a function of $v$ and $\CE[v]{Y_{i,t}^x} = \psi(v_i(x))$ by Assumption \ref{assumption:top}. The sixth holds by linearity of expectation and the definition of $\theta_i^x$. The last equality holds using an anonymous selection and i.i.d.\ valuations across $i$.

For the lower bound, let $\bar v_i^K := \max_{y \notin S_i^K} v_i(y)$ be $i$'s $(K+1)$-th highest valuation, and fix $\delta > 0$. Then:
\begin{align*}
    \mathcal{P}(\pi_{DPM}^\ell) = \E*{\Psi_i} &\ge \E*{(1 - \omega_i)\psi(\bar v_i^K)} \\
    &= \E*{(1 - \omega_i)\psi(\bar v_i^K)\mathbbm{1}\{\bar{v}_i^K > \bar{V} - \delta\} + (1 - \omega_i)\psi(\bar v_i^K)\mathbbm{1}\{\bar{v}_i^K \le \bar{V} - \delta\}} \\
    &\ge \psi(\bar V - \delta)\,\E*{(1 - \omega_i)\mathbbm{1}\{\bar v_i^K > \bar V - \delta\}} \\
    &\ge \psi(\bar V - \delta)\Big(1 - \E*{\omega_i} - \Pr\big(\bar v_i^K \le \bar V - \delta\big)\Big)
\end{align*}
The first inequality drops $x \notin S_i^K$ since $\psi \ge 0$ given $Y_{i,t}^x \ge 0$. It also uses that every $x \in S_i^K$ has $v_i(x) \ge \bar v_i^K$, that $\psi$ is increasing by Assumption \ref{assumption:non-decreasing}, and that the weights on $S_i^K$ sum to $1 - \omega_i$. The second inequality holds because $\psi \ge 0$ is increasing. The third inequality holds because $(1 - \omega_i)\mathbbm{1}\{A\} \ge 1 - \omega_i - \mathbbm{1}\{A^c\}$ for any event $A$ when $\omega_i \in [0,1]$.

Fix $\epsilon \in (0,1]$, and choose $K$ large enough that $2\alpha\log K/K < \epsilon/2$ and $\log K \ge 1 + \log\bar\rho$. Then, for all $\ell$ large enough:
\begin{align*}
    \mathcal{P}(\pi_{DPM}^\ell) &\ge \psi(\bar V - \delta)\Big(1 - \E*{\omega_i} - \Pr\big(\bar v_i^K \le \bar V - \delta\big)\Big) \\
    &\ge \psi(\bar V - \delta)(1 - 2\epsilon)
\end{align*}
This holds by Claim 5 and $\Pr(\bar v_i^K \le \bar V - \delta) < \epsilon$ for $|X| = \ell, |I| = h(\ell)$ large enough by Assumption \ref{assumption:top}.

Finally, the comparison with $|X| = 2$. For any number of workers:
\begin{align*}
    \mathcal{P}(\pi_{DPM}^{2,\ell}) &\le \E*{\psi(\max\{v_i(x), v_i(y)\})} \\
    &< \bar\psi
\end{align*}
The first inequality holds because $\mathcal{P}(\pi_{DPM}^{2,\ell}) = \E*{\Psi_i}$, where $\Psi_i$ is a weighted average of $\psi(v_i(x))$ and $\psi(v_i(y))$. The second holds because $\max\{v_i(x), v_i(y)\} < \bar V$ with probability one and $\psi$ is strictly increasing, so $\psi(\max\{v_i(x), v_i(y)\}) < \bar\psi$ with probability one.

Since $\epsilon$ and $\delta$ are arbitrary and $\lim_{\delta \to 0}\psi(\bar V - \delta) = \bar\psi$, choose them such that $\psi(\bar V - \delta)(1 - 2\epsilon) > \E*{\psi(\max\{v_i(x), v_i(y)\})}$. Then, for all $\ell$ large enough:
\[\mathcal{P}(\pi_{DPM}^\ell) \ge \psi(\bar V - \delta)(1 - 2\epsilon) > \E*{\psi(\max\{v_i(x), v_i(y)\})} \ge \mathcal{P}(\pi_{DPM}^{2,\ell})\]
$\blacksquare$

The final inequality in the main text holds by Corollary \ref{corollary:conservative}, since $\mathcal{P}(\pi_{DPM}^{2,\ell}) \ge \underline{\mathcal{P}}$ for any number of workers under Assumption \ref{assumption:relative}.

\subsection{Linearity}\label{appendix:linearity}

Assumption \ref{assumption:linear} is a substantive assumption on productivity. The specification allows any linear effects. This includes worker and task fixed effects or a worker-specific caseload burden. One can accommodate these through re-writing $\beta_i^x$. For example, if each task $x$ lowers worker $i$'s total output by $c^x$, then $\beta_i^x - c^x$ replaces $\beta_i^x$. Linearity rules out effects such as complementarities between a worker's tasks, interactions across workers, or learning that depends on the sequence of tasks.

\subsection{Estimation}\label{appendix:estimation}

In the simulations, I compute the bound using a subsample-then-estimate method. I randomly sample two tasks $x,y \in X$, normalize the primitives so that $\tilde f^z = f^z/(f^x + f^y)$, where $f^z := \frac{1}{T}\sum_{t=1}^T F_t(z)$, and $\tilde w^z = w^z/(w^x + w^y)$ for $z \in \{x,y\}$, then I estimate the relevant statistics ($\tilde g^x, \tilde g^y, \tilde g^{x,y}, \tilde g^{y,x}$) in Rotation markets using workers' ordinal preferences. I use all possible subsamples (if $|X| \leq 6$) or up to twenty random, unique subsamples (if $|X| \geq 8$). $a$ is the ex-post Rotation allocation. For $z \in \{x,y\}$, let:
\[|a_i|_z := \sum_{t = 1}^T a_{i,t}^z \quad \text{ and } \quad \mathcal{P}_z(a_i) := \frac{1}{|a_i|_z} \sum_{t = 1}^T Y_{i,t}^z a_{i,t}^z\]
I use the following estimate from simulated preferences for $o_i$:
\begin{align*}
    o_i := \mathbbm{1}\{ v_i(x) > v_i(y) \}
\end{align*}
$I_{\hat o} = \sum_{i \in I} \mathbbm{1}\{ o_i = \hat o \}$. I estimate $\tilde g^x$:
\begin{align*}
    \tilde g^x &:= \frac{1}{I_1} \sum_{i \in I} \mathbbm{1}\left\{ o_i = 1 \right\} \cdot \mathcal{P}_x(a_i)
\end{align*}
and $\tilde g^{x,y}$:
\begin{align*}
    \tilde g^{x,y} &:= \frac{1}{I_1} \sum_{i \in I} \mathbbm{1}\left\{ o_i = 1 \right\} \cdot \mathcal{P}_y(a_i)
\end{align*}
and analogously for $\tilde g^y$ and $\tilde g^{y,x}$. I compute $\underline{\mathcal{P}}$ using the estimated quantities.

With $|X| > 2$, I average the pairwise bounds and the pairwise Rotation productivities then take the difference. I weight each pair $(x,y)$ by $f^x + f^y$. The weighting scheme preserves Rotation's productivity. Specifically:
\[\frac{\sum_{(x,y)} (f^x + f^y)\left(\tilde f^x g^x + \tilde f^y g^y\right)}{\sum_{(x,y)} (f^x + f^y)} = \frac{(|X| - 1)\sum_{x \in X} f^x g^x}{(|X| - 1)\sum_{x \in X} f^x} = \mathcal{P}(\pi_R)\]
where the first equality holds because $\tilde f^z = f^z/(f^x + f^y)$ and each task type appears in $|X| - 1$ pairs, and the second holds because $\sum_{x \in X} f^x = 1$. Sampling random pairs approximates this quantity. 

\subsection{Equilibrium Computation}\label{appendix:computation}

\textit{Algorithm: Balanced Equilibrium Computation} (Input: $\mathcal{M}$.)
\begin{enumerate}
    \item Set $p_t^x := 1$ for all $(x,t) \in X \times [T]$.
    \item For each $i \in I$ sequentially, given $q_{-i}$:
    \begin{enumerate}
        \item Let $\Lambda_i := \left\{(x,t) : F_t(x) v_i(x)/p_t^x \ge 0.99 \max_{(y,t')} F_{t'}(y) v_i(y)/p_{t'}^y \right\}$.
        \item Set $q_{i,t}^x := 0$ for $(x,t) \notin \Lambda_i$ and, for $(x,t) \in \Lambda_i$:
        \[q_{i,t}^x = \frac{W_F}{|I|} \cdot \frac{1 - \sum_{j \ne i} q_{j,t}^x}{\sum_{(y,t') \in \Lambda_i} F_{t'}(y) w^y \left(1 - \sum_{j \ne i} q_{j,t'}^y\right)}\]
        \item Set $r_i := 1/(p \cdot q_i)$.
    \end{enumerate}
    \item Set $p_t^x := p_t^x \exp\left\{0.01\left(\sum_{i \in I} q_{i,t}^x - 1\right)\right\}$ for all $(x,t) \in X \times [T]$.
    \item Repeat Steps 2--3 until $\max_{(x,t)} \left|\sum_{i \in I} q_{i,t}^x - 1\right| < 0.001$ or for 20,000 iterations. Keep the iterate with the smallest maximum excess demand.
    \item For each $(x,t)$ with $\sum_{i \in I} q_{i,t}^x > 0$, rescale $q_{i,t}^x$ so that $\sum_{i \in I} q_{i,t}^x = 1$. \\~\
\end{enumerate}

Step 2(a) satisfies utility maximization up to tie tolerance. Step 2(b) satisfies balance. Step 2(c) satisfies the budget constraint. A fixed point of Step 3 is market clearing. This algorithm converged in 99.5\% instances of Figure \ref{figure:large} which contains the largest market simulations.

\subsection{Asymptotic Strategyproofness}\label{appendix:asymptotic}

DPM is not strategyproof in large markets when only the number of workers increases. Generally, the mapping from preferences to competitive equilibria prices tends to stabilize as the number of agents in the market grows. Consequently, agents effectively purchase goods at fixed prices. Truthful reporting is optimal because the mechanism purchases the agent's favorite bundle at the fixed prices. The primary issue in balanced equilibria is that endogenous budgets can be very discontinuous in reported preferences. The budgets always adjust to satisfy balance, and workers can exploit this fact to manipulate the market through budgets rather than prices. 

This point is technically evident from Proposition \ref{proposition:threshold}. Suppose weights are all equal. For a profile $v$, let $\hat v = (\hat v_i, v_{-i})$ be a unilateral manipulation by $i$. $\hat V_i = \hat v_i(x) / \hat v_i(y)$. $\hat k$ is $i$'s relative valuation's rank in $(\hat V_i, V_{-i})$. Under any $n := |I|$ and $|X| = 2$, a worker $i$ receives $x$ if $\hat k \le \lfloor \kappa \rfloor$. It is trivial that this is a profitable manipulation for any $n$ if $v_i(x) > v_i(y)$ but $V_i = V_{(k)}$ for some $k > \kappa + 1$.